\documentclass[12pt,letterpaper]{article}
\usepackage[letterpaper,margin=1.25in]{geometry}
\usepackage[T1]{fontenc}
\usepackage[utf8]{inputenc}
\usepackage{lmodern}
\usepackage{amssymb}
\usepackage{graphicx}
\usepackage{amsmath}
\usepackage{amsthm}
\usepackage{color}
\usepackage{hyperref}
\usepackage{float}

\definecolor{linkblue}{rgb}{0,0,0.63}
\definecolor{citered}{rgb}{0.63,0,0}
\hypersetup{
    colorlinks=true,
    linkcolor=linkblue,
    citecolor=citered,
    urlcolor=linkblue
}
\newtheorem{theorem}{Theorem}

\newtheorem{definition}[theorem]{Definition}

\newtheorem{lemma}[theorem]{Lemma}

\newtheorem{proposition}[theorem]{Proposition}

\theoremstyle{definition}
\newtheorem{remark}[theorem]{Remark}
\providecommand{\tbigcup}{\bigcup}

\providecommand{\FRAME}[6]{}

\begin{document}

\title{Uniqueness of Recursive Utility with Unbounded Consumption and
Non-Expected Utility\thanks{%
We are grateful to Pablo F. Beker, Jaden Yang Chen, Roberto Marfè, and Veronica Merlone for their helpful comments and suggestions.}}
\author{Luigi Montrucchio\thanks{%
ESOMAS Department and Collegio Carlo Alberto, Universit\`a di Torino, Corso
Unione Sovietica 218 Bis, 10134 Turin, Italy. Emails:
luigi.montrucchio@unito.it, lorenzomaria.stanca@unito.it} \and Lorenzo Stanca%
\footnotemark[2]}
\maketitle

\begin{abstract}
We derive sufficient conditions for the existence, non-existence, uniqueness
and global attractivity of stochastic recursive utilities. Our approach uses
abstract order-theoretic fixed-point arguments, allowing us to accommodate a
large class of Markov consumption-growth models with unbounded per-period
consumption and a broad range of recursive preferences, including
non-expected utility specifications. We illustrate the usefulness of the
framework through macro-finance applications to widely used
consumption-growth models under several preference specifications, including
Epstein--Zin preferences, multiple priors, smooth ambiguity, and
Hansen--Sargent robust-control preferences. In each application, our general
conditions reduce to a simple, economically interpretable parameter test
that shows how fundamental parameters influence existence and uniqueness.
While uniqueness of recursive utility can be difficult to guarantee in
unbounded environments, in a Markov setting with non-unit elasticity we
establish uniqueness within the class of recursive utilities having a
so-called separability property.
\end{abstract}



\section{Introduction}

Recursive utility plays an important role in macroeconomics and finance
because it provides a tractable way to separate intertemporal substitution
from risk aversion and, more generally, from attitudes toward uncertainty.
This flexibility has led to applications across a wide range of topics, from
household finance~\cite{CalvetCampbellGomesSodini2026} to models that
combine economic and climate risks~\cite{CaiLontzek2019}. Following the
seminal contributions of Weil~\cite{Weil1989,Weil1990}, recursive utility
has become a standard tool in macro-finance for incorporating ambiguity,
multiple priors, robust control, and other departures from expected utility;
see Backus, Routledge, and Zin~\cite{BackusRoutledgeZin2005} for a review.

This flexibility comes at a cost. Recursive utility is a fixed point of a
nonlinear stochastic operator, so neither existence nor uniqueness can be
taken for granted. Moreover, applications often feature non-expected utility
or ambiguity, which further complicates the analysis. Hence, before using a
recursive preference specification in applications, one must know whether
the recursion has a solution and whether that solution is unique. These
questions are not merely technical. If multiple solutions exist, the model
may generate ambiguous welfare rankings or indeterminate asset-pricing
implications.

Following earlier contributions, including Marinacci and Montrucchio~\cite%
{marmont}, one strand of the literature has relied on finding metrics to
recover contraction properties. These methods, however, typically require
restrictions that are strong for macroeconomic applications, such as
consumption being bounded at every date. Such restrictions are automatically
violated by standard consumption processes with normally distributed error
terms.

This paper aims to fill this gap by providing general conditions for the
existence, non-existence, uniqueness, and global attractivity of stochastic
recursive utilities with unbounded consumption. Our approach relies on
fixed-point methods in ordered spaces, especially the techniques of
Kantorovich~\cite{kanto} and Krasnosel'skii~\cite{Krasnoselskii}, as already
used by a few authors (see, for example, Christensen~\cite{Christensen2022}).

The key novelty is our uniqueness argument. A result by Marinacci and
Montrucchio~\cite[Theorem~1]{marmont1} allows us to establish uniqueness
using a method that is comparatively simple relative to existing approaches.
It is based on inspecting the behavior of the concave operator on the
so-called \emph{lower perimeter} of its domain. This approach is often
simpler than computing the spectral radius of a subgradient/supergradient of
the convex/concave operator, which is the alternative route to uniqueness
(see for instance \cite{Amann} and \cite{Christensen2022}).

Recursive utility consists of two main ingredients: a \emph{time aggregator}
and a \emph{certainty equivalent}. The time aggregator determines how
current consumption is combined with future continuation utility, and
therefore governs the elasticity of intertemporal substitution (EIS). The
certainty equivalent summarizes the agent's evaluation of uncertainty about
future continuation values, capturing risk attitudes and, more generally,
attitudes toward uncertainty. We focus on the two aggregators most commonly
used in applications, CES and Cobb--Douglas. Our assumptions on certainty
equivalents are deliberately minimal: monotonicity and subhomogeneity. These
conditions have natural behavioral interpretations. For instance,
subhomogeneity is related to increasing relative risk aversion (see \cite%
{marmont}).

Our sufficient conditions for existence have a simple economic
interpretation. For CES aggregators, they depend on an \emph{adjusted
consumption growth measure} that incorporates attitudes toward uncertainty,
discounting, and intertemporal substitutability. When the EIS is below one,
existence requires this adjusted measure to be uniformly bounded away from
zero; that is, consumption cannot grow too slowly.

In contrast, when the EIS is above one, existence requires the adjusted
measure to be uniformly bounded from above; that is, consumption cannot grow
too fast. In the unit-EIS case, a special case of the Cobb--Douglas
aggregator, the analogous condition is the integrability of consumption
adjusted for attitudes toward uncertainty. The non-existence results mirror
these conditions: recursive utility may fail to exist if adjusted growth is
uniformly too low or uniformly too high.

These existence results may yield more than one fixed point. The key
assumption that delivers uniqueness of recursive utility is order concavity
of the certainty equivalent, which has a natural economic interpretation in
terms of risk aversion or ambiguity aversion. Finally, we show that
uniqueness also implies global attractivity: starting from any process
comparable to consumption, repeated application of the recursive operator
converges, in the relevant sense, to the unique fixed point.

Thus far, our analysis has been deliberately abstract. We now specialize the
results to a standard Markov environment that covers many consumption growth
processes used in applications. The economy has a Markov state variable,
which represents an economically relevant variable such as the growth regime
or disaster risk. Consumption growth is allowed to depend on the current
state, the next state, and an independent shock.

In this Markov setting, when the certainty equivalent is also positively
homogeneous and stationary, the problem can be simplified by considering
recursive utilities that take a particularly simple form: current
consumption multiplied by a function of the current state. We call such
recursive utilities \emph{separable}. This class of recursive utilities is
useful because it replaces the original problem of solving for an entire
utility process with the simpler problem of solving for a single function of
the state. Notably, with non-unit EIS, our conditions imply that whenever a
separable recursive utility exists, it is unique. By contrast, at unit EIS
multiple separable recursive utilities may coexist, and for Epstein--Zin
preferences we characterize the set of separable recursive utilities.

We illustrate the usefulness of the framework through applications to widely
used consumption-growth models in macro-finance under several preference
specifications. We begin with the classic Mehra--Prescott regime switching
economy. In the parameter range considered by Weil~\cite{Weil1989}, we show
that our conditions for existence and uniqueness of recursive utility are
always satisfied. For certain parameters in this range, however, relatively
small changes in consumption growth rates are enough to obtain
non-existence: an increase of about $5\%$ in a high-EIS case and a reduction
of about $3\%$ in a low-EIS case.

The second application is a two-state hidden-regime model, in which
consumption growth depends on an unobserved regime. Models of this type are
standard in macro-finance; see, for instance, Lettau, Ludvigson, and Wachter~%
\cite{LettauLudvigsonWachter2008} and Ju and Miao~\cite{JuMiao2012}. Under
the Ju--Miao calibration, our conditions imply existence and uniqueness
under smooth ambiguity (non-unit EIS), Hansen--Sargent robust control (unit
EIS), and multiple priors (infinite ambiguity aversion). Notably, for the
smooth-ambiguity specification, uniqueness of recursive utility had not been
established. Our results close this gap by proving uniqueness under the
Ju--Miao calibration.

To illustrate the non-existence conditions, we consider a counterfactual
with a less persistent low-growth regime. Under smooth ambiguity, an
increase of about $8.3\%$ in gross growth in both regimes is sufficient for
non-existence. Under multiple priors, by contrast, a much larger increase of
about $15.5\%$ in both growth rates is required to satisfy the sufficient
condition for non-existence.

The third application is a rare-disaster model in which consumption growth
is exposed to persistent disaster risk, as in Wachter~\cite{Wachter2013} and
Bidder and Smith~\cite{BidderSmith2018}. In this environment, the conditions
guaranteeing uniqueness of recursive utility may fail even though our
existence conditions are satisfied. In particular, with unit EIS the
rare-disaster model can admit multiple separable recursive utilities.
Surprisingly, with non-unit EIS our results imply uniqueness within the
class of separable recursive utilities, even though they do not rule out
other non-separable recursive utilities. 

Finally, we apply our characterization of separable recursive utilities at
unit EIS to the constant-volatility consumption specification of Bansal and
Yaron~\cite{BansalYaron2004}. This yields two distinct utilities belonging
to disjoint cones (Section~\ref{sec:long-run-risk-multiplicity}).

The paper is organized as follows. Section~\ref{sect. intro} introduces the
mathematical framework. Section~\ref{sec:exandunique} establishes the main
existence, uniqueness, global-attractivity, and non-existence results.
Section~\ref{sec:kappa-model} applies the general framework to standard
macro-finance models. Section \ref{sec:conc} concludes. Appendix~\ref{append}
contains the proofs and auxiliary theoretical results, while the
Supplemental Appendix (Appendix \ref{suppapend}) provides additional
derivations and calculations for the applications. 

\paragraph{Related literature.}

The literature on recursive preferences with unbounded consumption remains
limited, particularly with respect to uniqueness, which is often established
only in special cases. The contributions most closely related to ours are
Borovi\v{c}ka and Stachurski~\cite{BorovickaStachurski2020}, Hansen and
Scheinkman~\cite{HansenScheinkman2009,HansenScheinkman2012}, Christensen~%
\cite{Christensen2022}, and Pohl, Schmedders, and Wilms~\cite%
{PohlSchmeddersWilms2024}. We return in Section~\ref{section:EZ} to the more
technical connection between their approaches and the Epstein--Zin special
case of our framework.

Borovi\v{c}ka and Stachurski~\cite{BorovickaStachurski2020} obtain necessary
and sufficient conditions for existence and uniqueness in the Epstein--Zin
model under the assumption that the state space is compact. In contrast, our
conditions apply to a broader class of certainty equivalents and time
aggregators, including non-expected-utility specifications. Moreover, in
certain cases, we can allow for an unbounded state space.

Hansen and Scheinkman~\cite{HansenScheinkman2009,HansenScheinkman2012}
develop a spectral approach to Epstein--Zin recursive utility, establishing
existence and, under additional restrictions, uniqueness. Our approach
accommodates a broader class of certainty equivalents.

Christensen~\cite{Christensen2022} studies recursive utility without
boundedness assumptions using exponential-Orlicz spaces, which are related
to thin-tailed conditions. His framework allows noncompact state spaces and
unbounded consumption, and provides existence and uniqueness results for
several important recursive specifications with unit EIS. By contrast, our
approach yields existence, uniqueness, and global attractivity for a broader
class of preferences.

Pohl, Schmedders, and Wilms~\cite{PohlSchmeddersWilms2024} derive existence
and non-existence conditions for several recursive preference
specifications, including Epstein--Zin and smooth ambiguity. Our existence
conditions do not rely on Jensen's inequality. For Epstein--Zin preferences,
they neither imply nor are implied by the sufficient conditions obtained
from their explicit Jensen and Minkowski bounds.

A related dynamic-programming literature studies Bellman equations with
recursive utility. Le Van and Vailakis~\cite{LeVanVailakis2005} analyze
recursive utility in optimal growth models with unbounded returns, and
Martins-da-Rocha and Vailakis~\cite{MartinsDaRochaVailakis2013} use
local-contraction methods. Bloise and Vailakis~\cite{BloiseVailakis2018},
Balbus~\cite{Balbus2020}, Ren and Stachurski~\cite{RenStachurski2021}, and
Bloise, Le Van, and Vailakis~\cite{BloiseLeVanVailakis2024} use monotone,
concave, or Tarski-type methods in related recursive-utility and
dynamic-programming problems. In particular, Ren and Stachurski~\cite%
{RenStachurski2021} use fixed-point arguments due to Du \cite{Du1990} (see
in particular his Theorem 3.1). This fixed-point argument requires the
positive cone of the relevant space to have nonempty interior. In our $L^{1}$
setting, the positive cone has empty interior.

Finally, observe that solving forward-looking equations with unbounded
variables of the type studied here is also relevant for general equilibrium
models, as they are required to characterize equilibrium allocations and
prices (e.g., see Beker and Chen \cite{BekerChen2026}).

\section{Mathematical preliminaries\label{sect. intro}}

\paragraph{Stochastic processes.}

Let $(\Omega ,\mathbf{\Sigma },\mathbb{P})$ be a probability space equipped
with an increasing filtration $\mathbf{\Sigma }=(\Sigma _{t})_{t\in T}$,
where $T=\left\{ 0,1,2,...\right\} $. A process $X=(X_{t})_{t\in T}$ is
adapted if $X_{t}$ is $\Sigma _{t}$-measurable for all $t\in T$. Denote by $%
L $ the set of all finitely valued adapted processes, $L_{+}\subset L$ the
nonnegative processes and $L_{++}$ the strictly positive ones.

It is conceptually useful to regard a discrete-time adapted process as a
measurable function defined on a suitable $\sigma $-algebra (see also
Appendix A of \cite{marmont}). It is enough to introduce the $\sigma $%
-algebra $\mathcal{F}$ in $\Omega \times T$, defined as $A\in \mathcal{F}$
if and only if the $t$-section $A_{t}\in \Sigma _{t}$ for all $t\geq 0.$

In what follows, some subspaces of $L_{+}$ will be relevant. Given a scalar $%
0<\delta <1$, consider the discrete measure $c_{\delta }\left( \left\{
t\right\} \right) =\delta ^{t}$ for $t\in T$ and the corresponding $L^{1}$%
-space $L^{1}\left( \Omega \times T,\mathcal{F},\text{ }\mathbb{P}\otimes
c_{\delta }\right) \subset L$, denoted by $L^{1}\left( \delta \right) $ for
simplicity. The norm in $L^{1}\left( \delta \right) $ is given by 
\begin{equation*}
\left\Vert X\right\Vert =\int \left\vert X\left( \omega ,t\right)
\right\vert d\left( \mathbb{P}\otimes c_{\delta }\right)
\end{equation*}%
where $X$ is an equivalence class of measurable functions rather than a
single function. We use the same symbol to denote the $L^{1}$-norms of the
functions $X_{t},$ if no confusion arises. Namely, 
\begin{equation*}
\left\Vert X_{t}\right\Vert =\int \left\vert X_{t}\right\vert d\mathbb{P} =%
\mathbb{E}\left[\left\vert X_{t}\right\vert\right] \text{ \ for each }t\in T%
\text{.}
\end{equation*}%
By Fubini's theorem we have the relation\footnote{%
This $L^{1}\left( \delta \right) $-norm is the $L^{1}$-counterpart of
weighted supnorms adopted by Boyd \cite{Boyd1990} and \cite{marmont}.} 
\begin{equation*}
\left\Vert X\right\Vert =\sum_{t\in T}\left\Vert X_{t}\right\Vert \delta
^{t}.
\end{equation*}%
The order relation $X\leq Y$ in $L$, as well as in $L^{1}\left( \delta
\right) $, means $X_{t}\leq Y_{t}$, a.s., for all $t\in T$. The cone $%
L_{+}^{1}\left( \delta \right) $ consists of the elements $X\in L^{1}\left(
\delta \right) $ satisfying $X\geq 0$. The cone $L_{++}\left( \delta \right) 
$ consists of the elements $X>0,$ namely $X_{t}>0$, a.s., for all $t.$ We
will also use the stronger order relation $Y\gg X$, meaning that $Y\geq
X+\varepsilon $ for some $\varepsilon >0$.

The space $L^{1}\left( \delta \right) $ equipped with the order $\leq $ is a
Banach lattice (i.e., a Riesz space endowed with a complete lattice norm,
see for instance \cite{Aliprantis,Ok2011}). It has the property of being
Dedekind $\sigma $-complete\footnote{%
Indeed, it is Dedekind complete too (see for instance \cite[Proposition
II.4.1]{neveu}), but the completeness for countable sets suffices for our
purposes, as will become apparent later on.}. This entails the
so-called \emph{monotone convergence property}. In symbols, if $%
X_{n}\uparrow $ $\leq X$ for some $X\in L^{1}\left( \delta \right) $, then
there is an element $X^{\ast }\in L^{1}\left( \delta \right) $ such that $%
X_{n}\uparrow X^{\ast }=\sup_{n}X_{n}$ and the convergence holds in the
space $L^{1}\left( \delta \right) $. This last property is just a
consequence of the monotone convergence theorem. In the Russian literature
ordered vector spaces having these properties are said to have a \emph{%
regular positive cone} (see \cite[Section 1.5]{Krasnoselskii}).

Given an element $X\in L_{+}$, $\mathcal{K}\left( X\right) $ denotes the
convex cone 
\begin{equation*}
\mathcal{K}\left( X\right) =\tbigcup_{0<\lambda <\mu <\infty }\left[ \lambda
X,\mu X\right] .
\end{equation*}%
An element $Y$ belongs to $\mathcal{K}\left( X\right) $ if and only if $%
\lambda X\leq Y\leq \mu X$ holds for some $\lambda ,\mu >0.$ For future
reference, notice that in this case there exist a least $\overline{\mu }$
and a greatest $\overline{\lambda }$ for which one has $\overline{\lambda }%
X\leq Y\leq \overline{\mu }X$.\footnote{%
This property is true in any Archimedean ordered vector space.}

To gain economic intuition for the condition $Y\in\mathcal{K}(X)$, observe
that it implies that $Y$ and $X$ have comparable long-run growth rates: more
precisely, these rates do not diverge from one another on average. The cone $%
\mathcal{K}\left( C\right) $ of consumption processes comparable with $C$
will play an important role throughout the paper. Indeed, it will be the
natural domain for the recursive preferences associated with the consumption
process $C$. Let us add some simple properties of it. If $C\in L_{++\text{ }%
} $, then $\mathcal{K}\left( C\right) \subset L_{++\text{ }}$. If $C\in L_{++%
\text{ }}^{1}\left( \delta \right) $, then $\mathcal{K}\left( C\right)
\subset L_{++\text{ }}^{1}\left( \delta \right) $.

\paragraph{Aggregators.}

Consider the CES aggregator $W:\mathbb{R}_{++}^{2}\rightarrow \mathbb{R}_{+}$
given by 
\begin{equation*}
W(x,y)=\big((1-\beta )x^{r}+\beta \,y^{r}\big)^{1/r},\text{ }x,y>0
\end{equation*}%
where $0<\beta<1$ and $r\in(-\infty,0)\cup(0,1)$. The aggregator $W$ is
monotone, positively homogeneous and concave. Clearly, for $r\in \left(
0,1\right) $ it can be extended continuously on $x,y\geq 0$. Note that the
economic counterpart of $r$ is the \emph{elasticity of intertemporal
substitution}, defined by 
\begin{equation*}
\mathrm{EIS}:=\frac{1}{1-r}.
\end{equation*}%
Hence, $r<0\Longleftrightarrow 0<\mathrm{EIS}<1,$ while $0<r<1%
\Longleftrightarrow \mathrm{EIS}>1.$ Moreover, observe that the limit case $%
r\rightarrow 0$ yields the Cobb--Douglas aggregator $W(x,y)=x^{1-\beta
}y^{\beta }$, corresponding to unit elasticity of intertemporal
substitution, $\mathrm{EIS}=1$.

\paragraph{Certainty equivalents.}

The certainty equivalent operator $\mathcal{M}$, in general, maps a portion
of $L_{+}$ into $L_{+}$. The notation $\mathcal{M}_{t}\left( X\right) $
denotes the $t$-th component of $\mathcal{M}\left( X\right) .$ At present,
the certainty equivalent will not be entirely specified. We list below some
classes of certainty operators that will be used below.

\bigskip

i) $\mathcal{M}$ is \emph{monotone}, i.e., $X\leq Y\Longrightarrow \mathcal{M%
}X\leq \mathcal{M}Y$

ii) $\mathcal{M}$ is $\sigma $\emph{-order continuous}, i.e., $X^{n}\uparrow
X$ $\Longrightarrow \mathcal{M}(X^{n})\uparrow \mathcal{M}(X)$, as well as $%
X^{n}\downarrow X\Longrightarrow \mathcal{M}(X^{n})\downarrow \mathcal{M}(X)$

iii) $\mathcal{M}$ is \emph{subhomogeneous}, i.e., $\mathcal{M}\left( \alpha
X\right) \geq \alpha \mathcal{M}\left( X\right) $ for $\alpha \in \left(
0,1\right) \Longleftrightarrow \mathcal{M}\left( \lambda X\right) \leq
\lambda \mathcal{M}\left( X\right) $ for $\lambda >1$.

iv) $\mathcal{M}$ is \emph{order concave}, i.e., 
\begin{equation*}
\mathcal{M}\left( \alpha X_{1}+\left( 1-\alpha \right) X_{2}\right) \geq
\alpha \mathcal{M}\left( X_{1}\right) +(1-\alpha )\mathcal{M}\left(
X_{2}\right)
\end{equation*}%
holds for every $\alpha \in \left[ 0,1\right] $ and $0<X_{1}\leq X_{2}$.
Order convexity is defined analogously. Note that (ii) implies (i).

\begin{remark}
A broad class of subhomogeneous operators is given by the so-called \emph{%
quasi-arithmetic means}, i.e., $\mathcal{M}_{t}\left( X\right) =\phi
^{-1}\left( \mathbb{E}_{t}\left( \phi \circ X_{t+1}\right) \right) $ where $%
\phi $ is a twice-differentiable and strictly monotone function that
exhibits increasing relative risk aversion (IRRA); see \cite[Proposition 7]%
{marmont}. In the following, a central role will be played by the
Kreps--Porteus case $\phi (t)=t^{s}$, i.e., $\mathcal{M}_{t}\left( X\right)
=\left( \mathbb{E}_{t}\left( X_{t+1}^{s}\right) \right) ^{1/s}$ where $1-s$
is the coefficient of relative risk aversion, 
\begin{equation*}
\mathrm{RRA}:=1-s.
\end{equation*}%
We refer also to the discussion in Section \ref{section:EZ} and Proposition %
\ref{prop:conditional-power-mean-concavity}. Its Markovian counterpart is
the Epstein--Zin certainty equivalent treated in subsection \ref%
{subsec:kappa-general} (see (\ref{eq:EZce})).
\end{remark}

\paragraph{Recursive utility.}

Recursive utilities will be identified with the fixed points of an operator
acting on the space $L^{1}\left( \delta \right) $. Given a consumption
process $C\in L_{++}$, the operator $\mathbf{T}$ under study is\footnote{%
Indeed $\mathbf{T}$\textbf{\ }depends on the process $C$, and so one should
write more correctly $\mathbf{T}_{C}$. Analogously a fixed point $U$ of $%
\mathbf{T}_{C}$ should be labelled as $U_{C}$.} 
\begin{equation}
\mathbf{T}X:=W\!\big(C,\,\mathcal{M}(X)\big).  \label{eq:opT}
\end{equation}%
A recursive utility is a fixed point $U$ of $\mathbf{T}$, i.e., 
\begin{equation*}
U=\mathbf{T}(U)\quad \Longleftrightarrow \quad U_{t}=W\left( C_{t},\,%
\mathcal{M}_{t}(U)\right) \ 
\end{equation*}%
for all $t\in T.$

In the following we will restrict our analysis to the most relevant
certainty operators for which $\mathcal{M}_{t}(U)\equiv \mathcal{M}%
_{t}(U_{t+1}).$ Written explicitly, the recursive utilities $U_{t}$ solve
the forward-looking difference equation 
\begin{equation}
U_{t}=\left[ (1-\beta )C_{t}^{r}+\beta \left( \mathcal{M}_{t}(U_{t+1})%
\right) ^{r}\right] ^{1/r}\text{ for }t\in T  \label{eq:pref}
\end{equation}%
Here $C_{t}$ is the current consumption and the stream $\left( U_{t}\right)
_{t\geq 0}$ associated with the consumption process $C$ is known in the
macroeconomics literature as the \emph{continuation value }(see for instance
Hansen et al. \cite{HansenHeatonLeeRoussanov2007}). In this setting, $\beta $
is the \emph{subjective discount factor}, while $r=1-%
\mathrm{EIS}^{-1}$ is the parameter associated with the elasticity of
intertemporal substitution.

We note that the operator $\mathbf{T}$ is subhomogeneous, provided $\mathcal{%
M}$ is. Actually, since our aggregator function $W$ is monotone and
positively homogeneous, we have, for $\alpha \in \left( 0,1\right) $ 
\begin{equation*}
\mathbf{T}\left( \alpha X\right) =W\left( C,\mathcal{M}\left( \alpha
X\right) \right) \geq W\left( \alpha C,\alpha \mathcal{M}\left( X\right)
\right) =\alpha W\left( C,\mathcal{M}\left( X\right) \right) =\alpha \mathbf{%
T}\left( X\right)
\end{equation*}%
Moreover, as $W$ is continuous, if $\mathcal{M}$ is $\sigma $-order
continuous then $\mathbf{T}$ enjoys the same property.

Let $\mathcal{M}$ be monotone and subhomogeneous. If $\mathcal{M}\left(
C\right) \in L_{+}$, then $\mathcal{K}\left( C\right) \subset L_{+}$.
Indeed, if $X\in \mathcal{K}\left( C\right) $, then $X\leq \lambda C$ for
some $\lambda >1$. Hence, under our hypotheses 
\begin{equation}
\mathcal{M}(X)\leq \mathcal{M}\left( \lambda C\right) \leq \lambda \mathcal{M%
}\left( C\right)  \label{eq:HHH}
\end{equation}%
and so $\mathcal{M}(X)\in L_{+}$.

\section{Existence and uniqueness}

\label{sec:exandunique} In this section we study in detail the recursive
preference (\ref{eq:pref}), namely the fixed points of the operator $\mathbf{%
T}$ of equation (\ref{eq:opT}). Since the sign of $r$ determines the
relevant existence and uniqueness conditions, it is convenient to
distinguish the low-EIS case $r<0$ ($\mathrm{EIS}<1$) from the high-EIS case 
$0<r<1$ ($\mathrm{EIS}>1$). We then consider the Cobb--Douglas limit $r=0$ ($%
\mathrm{EIS}=1$), before establishing global attractivity and the related
non-existence results.

\subsection{Low EIS: $r<0$}

Since the function $\varphi \left( t\right) =\left[ (1-\beta )x^{r}+t\right]
^{1/r}$ is decreasing, %
\begin{equation*}
(1-\beta )^{1/r}x=\varphi \left( 0\right) \geq \left[ (1-\beta )x^{r}+\beta
y^{r}\right] ^{1/r}=W\left( x,y\right)
\end{equation*}%
and so we have the upper bound%
\begin{equation}
W\left( x,y\right) \leq (1-\beta )^{1/r}x  \label{eq:LLL}
\end{equation}

To get results on the fixed points of $\mathbf{T}$ we impose restrictions on
the consumption processes under analysis.

\begin{lemma}
\label{lemRT}Let $\mathcal{M}$ be monotone and subhomogeneous and $C\in
L_{++}$. Suppose that the following two assumptions hold:\newline
a) $\mathcal{M}\left( C\right) \in L_{+};$\newline
b) 
\begin{equation}
\frac{\mathcal{M}\left( C\right) }{C}\gg \beta ^{\frac{1}{\left\vert
r\right\vert }}\text{.}  \label{eq:BBB}
\end{equation}%
Then 
\begin{equation}
\frac{\mathbf{T}\left( \varepsilon C\right) }{\varepsilon C}\gg 1
\label{eq:HAZ}
\end{equation}%
holds for all $\varepsilon >0$ sufficiently small.
\end{lemma}

Observe that (\ref{eq:HAZ}) implies that $\mathbf{T}\left( \varepsilon
C\right) \geq \varepsilon C$. Moreover (\ref{eq:LLL}) implies that $\mathbf{T%
}\left( X\right) \leq (1-\beta) ^{1/r}C$ whenever $\mathbf{T}\left( X\right) 
$ is well defined. Therefore, under the hypothesis of this Lemma, we deduce
that $\mathbf{T}$ maps the interval $\left[ \varepsilon C,(1-\beta) ^{1/r}C%
\right] $ into itself for all sufficiently small $\varepsilon>0$.

This is the main ingredient to get the next result on the existence of fixed
points. Its proof uses an argument based on Kantorovich's fixed-point
theorem \cite[p. 68]{kanto}; see also \cite[Th. 4.1(b)]{Krasnoselskii} and 
\cite[Prop. 3.1 (i)]{Christensen2022}.

\begin{proposition}
\label{prop:exist-neg}Let $r<0$. Assume that (i)-(ii)-(iii) hold for the
certainty operator $\mathcal{M}$. Let $C\in L_{++}^{1}\left( \delta \right) $
for some $\delta >0$ and $\mathcal{M}\left( C\right) \in L_{+}$. If 
\begin{equation}
\frac{\mathcal{M}\left( C\right) }{C}\gg \beta ^{\frac{1}{\left\vert
r\right\vert }},  \label{eq:JJP}
\end{equation}%
then $\mathbf{T}:\mathcal{K}\left( C\right) \rightarrow \mathcal{K}\left(
C\right) $. More precisely, for every $X\in \mathcal{K}\left( C\right) $
there exists an $\varepsilon >0$ such that 
\begin{equation}
\varepsilon C\leq \mathbf{T}\left( X\right) \leq (1-\beta) ^{1/r}C
\label{eq:YUP}
\end{equation}%
Moreover, the set of the fixed points of $\mathbf{T}$ in $\mathcal{K}\left(
C\right) $ is not empty and contains the greatest fixed point.
\end{proposition}

A few comments are in order. The condition $C\in L_{++}^{1}\left( \delta
\right) $ is a very mild restriction on the growth of the paths of
consumption as $\delta $ approaches $0.$ Though the value $\delta $ does not
appear in condition (\ref{eq:JJP}), its positivity is crucial to guarantee
the Dedekind $\sigma $-completeness of the space.

As will be seen in the proof, a consequence of the fact that $\mathbf{T}$
maps $\left[ \varepsilon C,(1-\beta )^{1/r}C\right] $ into itself is that $%
\mathbf{T}$  also has a least fixed point $U_{\ast }^{\varepsilon }$ in this
interval; but it depends on the scalar $\varepsilon $ and so the existence
of a least fixed point in $\mathcal{K}\left( C\right) $ is not assured.

Condition~\eqref{eq:BBB} has a simple economic interpretation. It is
equivalent to 
\begin{equation}
\frac{1}{\beta }\left( \frac{\mathcal{M}(C)}{C}\right) ^{|r|}\gg 1.
\label{eq:GLA}
\end{equation}%
The ratio $\mathcal{M}(C)/C$ is the certainty-equivalent value of future
consumption relative to current consumption. Hence the left-hand side is an
adjusted measure of consumption growth that incorporates attitudes toward
uncertainty, discounting, and intertemporal substitutability. When the EIS
is below one, existence requires this adjusted growth measure to be
uniformly large enough: consumption cannot grow too slowly.

\subsubsection{Uniqueness}

We need the concept of lower perimeter $\partial _{\diamond }A$, for a
subset $A$ of an ordered vector space, introduced by Marinacci and
Montrucchio \cite{marmont1}. We refer to that paper for the general
definition and focus here only on the lower perimeter $\partial _{\diamond }%
\left[ \varepsilon C,(1-\beta )^{1/r}C\right] $ of the interval $\left[
\varepsilon C,(1-\beta )^{1/r}C\right] $ of Proposition \ref{prop:exist-neg}.

The elements of such a lower perimeter are\footnote{%
Recall that all our random variables are indeed equivalence classes. Thus, $%
\inf_{\omega }Y\left( \omega \right) $ means $\inf \left\{ x:\mathbb{P}%
\left( Y<x\right) >0\right\} $.}%
\begin{equation}
\partial _{\diamond }\left[ \varepsilon C,(1-\beta )^{1/r}C\right] =\left\{
H\in \left[ \varepsilon C,(1-\beta )^{1/r}C\right] :\inf \frac{H}{C}%
=\varepsilon \right\}  \label{eq:upper}
\end{equation}

Actually, by Proposition 4 of \cite{marmont1}, the process $H\in \partial
_{\diamond }\left[ \varepsilon C,(1-\beta )^{1/r}C\right] $ if and only if 
\begin{eqnarray*}
\inf \frac{H-\varepsilon C}{((1-\beta )^{1/r}-\varepsilon )C}
&=&0\Longleftrightarrow \inf \frac{H}{((1-\beta )^{1/r}-\varepsilon )C}=%
\frac{\varepsilon }{((1-\beta )^{1/r}-\varepsilon )} \\
&\Longleftrightarrow &\inf \frac{H}{C}=\varepsilon
\end{eqnarray*}

We are in a position to establish sufficient conditions for uniqueness of
the fixed points of $\mathbf{T}$. Observe that order concavity of $\mathcal{M%
}$ implies that $\mathbf{T}$\textbf{\ }is order concave as well, as it is
easy to check.

\begin{proposition}
\label{prop:unique1}Under the assumptions of Proposition~\ref{prop:exist-neg}%
, suppose in addition that $\mathcal{M}$ is order concave. Then the operator 
$\mathbf{T}$ has a unique fixed point in the cone $\mathcal{K}\left(
C\right) .$
\end{proposition}

This result shows that, once existence is guaranteed, no further growth
restriction is needed for uniqueness in the case $r<0$. Order concavity of
the certainty equivalent is sufficient to obtain uniqueness. Economically,
this means that sufficiently strong risk or uncertainty aversion rules out
multiple continuation-value processes compatible with the same consumption
stream. We now turn to the high-EIS case.

\subsection{High EIS: $r\in (0,1)$}

In this case the CES aggregator $W$ has the obvious lower bound 
\begin{equation}
W\left( x,y\right) \geq (1-\beta )^{1/r}x  \label{eq:ZAZ}
\end{equation}%
that implies $\mathbf{T}\left( X\right) \geq (1-\beta )^{1/r}C.$

An argument analogous to that used for the case $r<0$ shows that the
condition%
\begin{equation*}
\frac{\mathcal{M}\left( C\right) }{C}\ll \beta ^{-1/r}\text{ }
\end{equation*}%
implies that 
\begin{equation*}
\mathbf{T}\left( \lambda C\right) \leq \lambda C\text{ for }\lambda >1\text{
large enough.}
\end{equation*}

Therefore, a proof similar to that of Proposition \ref{prop:exist-neg}
delivers the next result.

\begin{proposition}
\label{propB}Let $0<r<1$. Assume that (i)-(ii)-(iii) hold for the certainty
operator $\mathcal{M}$. Let $C\in L_{++}^{1}\left( \delta \right) $ for some 
$\delta >0$ and $\mathcal{M}(C)\in L_{+}$. If%
\begin{equation*}
\frac{\mathcal{M}\left( C\right) }{C}\ll \beta ^{-1/r}
\end{equation*}%
then $\mathbf{T}:\mathcal{K}\left( C\right) \rightarrow \mathcal{K}\left(
C\right) $. More precisely, for all $X\in \mathcal{K}\left( C\right) $ there
exists a scalar $\lambda >1$ for which 
\begin{equation*}
(1-\beta )^{1/r}C\leq \mathbf{T}\left( X\right) \leq \lambda C.
\end{equation*}%
Moreover, the set of the fixed points of $\mathbf{T}$ is not empty in $%
\mathcal{K}\left( C\right) $ and the least fixed point is included in it.
\end{proposition}

The analogue of (\ref{eq:GLA}) is here%
\begin{equation*}
\beta \left( \frac{\mathcal{M}(C)}{C}\right) ^{r}\ll 1.
\end{equation*}%
This condition also has  a simple economic interpretation. The left-hand side
is an adjusted measure of consumption growth that incorporates attitudes
toward uncertainty, discounting, and intertemporal substitutability. When
the EIS is above one, existence requires this adjusted growth measure to be
uniformly bounded from above: consumption cannot grow too fast.

In this case the lower perimeter is given by 
\begin{equation*}
\partial _{\diamond }\left[ (1-\beta )^{1/r}C,\lambda C\right] =\left\{ H\in %
\left[ (1-\beta )^{1/r}C,\lambda C\right] :\inf \frac{H}{C}=(1-\beta
)^{1/r}\right\}
\end{equation*}%
Therefore for uniqueness it is required that $\mathbf{T}\left( (1-\beta
)^{1/r}C\right) /C\gg (1-\beta )^{1/r}$ together with the assumption of
order concavity. Notice that $\mathbf{T}\left( (1-\beta )^{1/r}C\right) \gg
(1-\beta )^{1/r}C$ is equivalent to the condition $\mathcal{M}\left(
C\right) /C\gg 0.$ The next proposition formalizes these facts.

\begin{proposition}
\label{prop:unique2}Let $0<r<1$. Assume that (i)-(ii)-(iii) hold for the
certainty operator $\mathcal{M}$. Under the additional condition that $%
\mathcal{M}$ is order concave, $\mathbf{T}$ has a unique fixed point in $%
\mathcal{K}\left( C\right) $ if 
\begin{equation*}
0\ll \frac{\mathcal{M}\left( C\right) }{C}\ll \beta ^{-1/r}
\end{equation*}
\end{proposition}

The key point of the proof of this result is that the relation $\mathbf{T}%
\left( (1-\beta )^{1/r}C\right) /C\gg (1-\beta )^{1/r}$ is equivalent to $%
\mathcal{M}\left( C\right) /C\gg 0,$ which in turn is nothing but the lower
perimeter condition for uniqueness discussed earlier. 

\subsection{Unit EIS: $r=0$ and Cobb--Douglas case\label{subsection Unit}}

We now look at the Cobb--Douglas aggregator 
\begin{equation*}
W\left( x,y\right) =x^{\alpha }y^{\beta }\text{ }
\end{equation*}%
with $\alpha >0$ and $\beta \in \left( 0,1\right) .$ The specification $%
\alpha +\beta =1$ is of some interest since it can be seen as the limit of
the CES aggregator as $r\rightarrow 0$ $($i.e., EIS $\rightarrow 1)$. The
associated operator is given by 
\begin{equation}
\mathbf{T}\left( X\right) =C^{\alpha }\left[ \mathcal{M}\left( X\right) %
\right] ^{\beta }  \label{eq:cobb}
\end{equation}

Set $p=\left( 1-\alpha \right) /\beta .$ Clearly $p\geq 1$ when $\alpha
+\beta \leq 1$ and $p=1$ if $\alpha +\beta =1$.

The search for fixed points of (\ref{eq:cobb}) is quite simplified with
respect to the CES case, thanks to a property which we are going to
highlight in the following definition.

\begin{definition}
\label{def:ssub}An operator $\mathbf{T}:L_{+}\longrightarrow L_{+}$ is said
to be $q$-\emph{subhomogeneous}, $q\in \left( 0,1\right) ,$ if 
\begin{equation*}
\mathbf{T}\left( \varepsilon X\right) \geq \varepsilon ^{q}\mathbf{T}\left(
X\right) \Longleftrightarrow \mathbf{T}\left( \lambda X\right) \leq \lambda
^{q}\mathbf{T}\left( X\right)
\end{equation*}%
holds for all $X\in L_{+}$, $\varepsilon \in \left( 0,1\right) $ and $%
\lambda >1.$
\end{definition}

We refer to \cite{marmont1} for more details on this strengthened version of
subhomogeneity. This immediately yields the following result. Its trivial
proof is omitted.

\begin{proposition}
\label{prop:FFF}If the certainty equivalent $\mathcal{M}$ is subhomogeneous,
then the operator $\mathbf{T}$ in (\ref{eq:cobb}) is $\beta $-subhomogeneous
with $\beta \in \left( 0,1\right) .$
\end{proposition}

Unlike the CES case, the next proposition proves that existence and
uniqueness in $\mathcal{K}(C)$ coincide, and concavity is not needed, thanks
to Proposition \ref{prop:FFF}.

\begin{proposition}
\label{prop:cobbdouglas} Assume that (i)-(ii)-(iii) hold for $\mathcal{M}$.
Then $\mathbf{T}:\mathcal{K}\left( C\right) \rightarrow \mathcal{K}\left(
C\right) $ and the set of its fixed points is not empty in $\mathcal{K}%
\left( C\right) $ if and only if 
\begin{equation}
\mathcal{M}(C)\in \mathcal{K}\left( C^{p}\right) .  \label{eq:HCM}
\end{equation}%
In such a case $\mathbf{T}$ has exactly one fixed point in $\mathcal{K}%
\left( C\right) $.
\end{proposition}

If $r=0$ then $p=1$, and condition (\ref{eq:HCM}) reduces to $\mathcal{M}%
(C)\in \mathcal{K}(C)$. Thus, in the unit-EIS case, this growth condition
guarantees the existence and uniqueness in $\mathcal{K}(C)$.

\begin{remark}
Uniqueness cannot hold in all of $L_{++}^{1}\left( \delta \right) .$ Suppose
that $\mathcal{M}$ is positively homogeneous. If $U$ is a fixed point in $%
L_{++}^{1}\left( \delta \right) $ then $V\left( \lambda \right) =\left(
\lambda ^{\beta ^{-t}}U_{t}\right) $ is a fixed point of $\mathbf{T}$ in $%
L_{++}^{1}\left( \delta \right) $, for every $\lambda \in \left( 0,1\right] $%
. In other words, every fixed point generates a continuum of non-comparable
fixed points. Indeed, $\mathcal{K}\left( V\left( \lambda \right) \right)
\cap \mathcal{K}\left( V\left( \lambda _{1}\right) \right) =\emptyset $ for 
$\lambda \neq \lambda _{1}$. To see this, let $V\equiv V\left( \lambda
\right) $ for short. Then%
\begin{eqnarray*}
\left( \mathbf{T}V\right) _{t} &=&C_{t}^{\alpha }\left[ \mathcal{M}%
_{t}\left( V_{t+1}\right) \right] ^{\beta }=C_{t}^{\alpha }\left[ \mathcal{M}%
_{t}\left( \lambda ^{\beta ^{-\left( t+1\right) }}U_{t+1}\right) \right]
^{\beta } \\
&=&\lambda ^{\beta ^{-t}}C_{t}^{\alpha }\left[ \mathcal{M}_{t}\left(
U_{t+1}\right) \right] ^{\beta }=\lambda ^{\beta ^{-t}}U_{t}=V_{t}
\end{eqnarray*}%
Since $V\left( \lambda \right) \leq U$, $V\left( \lambda \right) \in
L_{++}^{1}\left( \delta \right) $ and so it is a fixed point.
\end{remark}

The condition $\mathcal{M}(C)\in \mathcal{K}\left( C\right)$ for $\alpha
+\beta =1$ is quite strong and, as we shall see later in  Section \ref%
{subsec:cobb}, it may fail in economically relevant cases. This does not imply the non-existence of
fixed points outside the cone $\mathcal{K}\left( C\right) .$ Looking at the
proof of Proposition \ref{prop:FFF}, it is easy to realize that it can be
generalized by replacing the cone $\mathcal{K}\left( C\right) $ with any
other cone $\mathcal{K}\left( Z\right) $, with $Z\in L_{++}^{1}\left( \delta
\right) $.

\begin{proposition}
\label{prop:FFP} Assume that $\mathcal{M}$ satisfies (i)--(iii). Let $%
\alpha+\beta=1$ and $Z\in L_{++}^{1}(\delta)$. Then $\mathbf{T}:\mathcal{K}%
(Z)\rightarrow\mathcal{K}(Z)$ and has a unique fixed point in $\mathcal{K}%
(Z) $ if and only if 
\begin{equation}
\mathcal{M}(Z)\in \mathcal{K}\left(C^{-\alpha/\beta}Z^{1/\beta}\right).
\label{eq:FFP}
\end{equation}
\end{proposition}

In other words, to each admissible cone $\mathcal{K}(Z)$ there corresponds a
unique fixed point $U\in \mathcal{K}(Z)$, provided (\ref{eq:FFP}) holds.

For the Epstein--Zin certainty equivalent $\mathcal{M}_{t}(X)=\left( \mathbb{%
E}_{t}X_{t+1}^{s}\right) ^{1/s},$ we can go further and characterize
explicitly the processes $Z$ satisfying \eqref{eq:FFP}. Indeed, condition~%
\eqref{eq:FFP} is equivalent to the existence of an adapted process $B\in 
\mathcal{K}(\mathbf{1})$ such that (see the proof of Proposition \ref{prop:FFP}) 
\begin{equation}
\mathbb{E}_{t}\left[ Z_{t+1}^{s}\right] =C_{t}^{-\alpha s/\beta
}Z_{t}^{s/\beta }B_{t}.  \label{eq:DDD}
\end{equation}

It is worth remarking that, whenever in (\ref{eq:DDD}) the process $Z$ is
generated by a constant $B_{t}\equiv \overline{B}>0$, we can calculate the
unique recursive utility in $\mathcal{K}\left( Z\right) $ given by 
\begin{equation}
U=\overline{B}^{\beta /\left( s\alpha \right) }Z.  \label{eq:LSX}
\end{equation}

\subsection{Attractivity}

Uniqueness of the fixed point implies that the iterates converge to it from
any initial point. Recall that in a Riesz space a sequence $\left(
X_{n}\right) $ order converges to the point $X$, in symbols $X_{n}$ $%
\mathring{\rightarrow}$ $X$, whenever 
\begin{equation*}
\left\vert X-X_{n}\right\vert \leq Z_{n}
\end{equation*}%
holds for a nonnegative and decreasing sequence $Z_{n}\downarrow 0$. Clearly
in $L^{1}$ spaces order convergence implies pointwise convergence a.s., as
well as $L_{1}$-convergence thanks to Beppo Levi's theorem.

\begin{proposition}
\label{prop:convergence} Assume that $C\in L_{++}^{1}(\delta )$. Let $%
\mathbf{T}:\mathcal{K}(C)\rightarrow $ $\mathcal{K}(C)$ \ be the operator of
Proposition \ref{prop:unique1} or Proposition \ref{prop:unique2} having a
unique fixed point $\overline{U}$. For every initial point $\overline{X}\in 
\mathcal{K}(C)$, we have $\mathbf{T}^{n}(\overline{X})\mathring{\rightarrow}$
$\overline{U}$ and $\ \left\Vert \overline{U}-\mathbf{T}^{n}(\overline{X}%
)\right\Vert _{L^{1}(\delta )}\rightarrow 0$.
\end{proposition}

\bigskip

The next simple example illustrates why uniqueness is formulated for the
cone $\mathcal{K}(C)$. Even in a deterministic economy, the recursive
utility equation may admit several positive solutions, although only one of
them belongs to $\mathcal{K}(C)$.

Assume there is no uncertainty and that consumption is constant, $C_{t}:=1.$
Consider the linear certainty equivalent $\mathcal{M}_{t}(X)=X_{t+1}.$ The
recursive equation is 
\begin{equation*}
U_{t}=\left[ (1-\beta) +\beta U_{t+1}^{r}\right] ^{1/r}.
\end{equation*}%
Clearly, the constant process $U_{t}:= \mathbf{1}\in \mathcal{K}(\mathbf{1})$ is a solution of
the recursive equation. If we now, for every $A\geq 0$, define 
\begin{equation*}
U_{t}^{A}:=\left[ 1+A\beta^{-t}\right] ^{1/r},
\end{equation*}%
then 
\begin{equation*}
\mathbf{T}\left( U^{A}\right) _{t}=\left[ (1-\beta) +\beta \left(
U_{t+1}^{A}\right) ^{r}\right] ^{1/r}=U_{t}^{A}
\end{equation*}%
Hence $U^{A}$ is a fixed point of $\mathbf{T}$ for every $A\geq 0$. If $A>0$
and $r\in (0,1)$, then $U_{t}^{A}\uparrow \infty $, while $%
U_{t}^{A}\downarrow 0$ if $r<0.$ In both cases $U^{A}\notin \mathcal{K}(C)$
while the only fixed point in $\mathcal{K}(C)$ is $U_{t}\equiv 1$.

\subsection{Non-existence}

Propositions \ref{prop:exist-neg} and \ref{propB} state sufficient
conditions for the existence of recursive utilities. We show that those
conditions are \textquotedblleft almost\textquotedblright\ necessary, under
slightly different hypotheses on the certainty equivalent operator. 

\begin{proposition}
\label{prop:nonexist-positive} Let $r\in\left( 0,1\right) $ and $\mathcal{M} 
$ be monotone and positively homogeneous. If 
\begin{equation}
\frac{\mathcal{M}(C)}{C}\geq\beta^{-1/r},  \label{eq:non2}
\end{equation}
then $\mathbf{T}$ has no fixed point in $L_{+}$.
\end{proposition}

Thus, when the EIS is above one, adjusted consumption growth cannot be
uniformly too high. If the certainty-equivalent value of future consumption
is everywhere large enough relative to current consumption, there is no
recursive utility.

A perfectly analogous result is valid for $r<0$, by replacing (\ref{eq:non2}%
) with the condition 
\begin{equation}
\frac{\mathcal{M}(C)}{C}\leq \beta ^{1/|r|},  \label{eq:non3}
\end{equation}%
which implies there is no fixed point in $L_{++}$.

\bigskip

The severe restrictions (\ref{eq:non2}) and (\ref{eq:non3}) guarantee the
everywhere non-existence of fixed points. We next provide a weaker result
which rules out fixed points in the cone $\mathcal{K}(C)$ only.

\begin{proposition}
\label{prop:nonexistweak} Let $\mathcal{M}$ be monotone and subhomogeneous.
For $r\in \left( 0,1\right) ,$ under the condition 
\begin{equation*}
\sup \frac{\mathcal{M}\left( C\right) }{C}=\infty
\end{equation*}%
no fixed point of $\mathbf{T}$ lies in $\mathcal{K}\left( C\right) .$\newline
For $r<0$ the same result holds whenever 
\begin{equation*}
\inf \frac{\mathcal{M}\left( C\right) }{C}=0.
\end{equation*}
\end{proposition}

\subsection{Epstein--Zin preferences and related literature\label{section:EZ}%
}

As discussed in the introduction, much of the literature related to ours on
recursive utility with unbounded consumption or unbounded state variables
focuses on Epstein--Zin preferences. We now make the connection to our
framework more explicit. In this case, the certainty equivalent is $\mathcal{%
M}_{t}(X)=(\mathbb{E}_{t}X_{t+1}^{s})^{1/s}$, so the recursive utility
satisfies 
\begin{equation}
U_{t}=\left[ (1-\beta) C_{t}^{r}+\beta \left( \mathbb{E}_{t}U_{t+1}^{s}%
\right) ^{r/s}\right] ^{1/r}.  \label{eq:EPZIN}
\end{equation}

In terms of the economically relevant parameters $\mathrm{EIS}=(1-r)^{-1}$
and $\mathrm{RRA}=1-s$, the exponent ratio above is $r/s=(1-\mathrm{EIS}%
^{-1})/(1-\mathrm{RRA})$. The standard expected utility model corresponds to 
$\mathrm{RRA}=\mathrm{EIS}^{-1}$, while $\mathrm{RRA}>\mathrm{EIS}^{-1}$
corresponds to a preference for early resolution of uncertainty. Markov
specifications are discussed more extensively in the next section.

We have already remarked that the operator $\mathbf{T}$ associated with %
\eqref{eq:EPZIN} is concave for $s\leq 1$. Setting $g_{t+1}:=\log
(C_{t+1}/C_{t})$, we obtain 
\begin{equation*}
\frac{\mathcal{M}_{t}(C)}{C_{t}}=\mathcal{M}_{t}\left( \frac{C_{t+1}}{C_{t}}%
\right) =\mathcal{M}_{t}\left( e^{g_{t+1}}\right) =\left( \mathbb{E}%
_{t}e^{sg_{t+1}}\right) ^{1/s}.
\end{equation*}%
Hence, in the Epstein--Zin case, the abstract conditions in Propositions~\ref%
{prop:unique1} and~\ref{prop:unique2} become restrictions on the conditional
certainty-equivalent growth rate of consumption. This is the same object
that appears, in transformed form, in the existing literature. The main
difference is that those papers often work with transformations of %
\eqref{eq:EPZIN}, whereas our results are stated directly in terms of the
order properties of $\mathbf{T}$.

\bigskip

For instance, setting $u_{t}=U_{t}/C_{t}$ gives 
\begin{equation}
u_{t}=\left[ (1-\beta) +\beta \left( \mathbb{E}_{t}\left[
u_{t+1}^{s}e^{sg_{t+1}}\right] \right) ^{r/s}\right] ^{1/r}=\Gamma
_{r}\left( \left( \mathbb{E}_{t}\left[ u_{t+1}^{s}e^{sg_{t+1}}\right]
\right) ^{1/s}\right) ,  \label{eq:due}
\end{equation}%
where $\Gamma _{r}(t)=[(1-\beta) +\beta t^{r}]^{1/r}$.

If $U$ solves \eqref{eq:EPZIN} in $\mathcal{K}(C)$, then $u$ solves %
\eqref{eq:due} in the cone $\mathcal{K}(\mathbf{1})$. The operator associated with %
\eqref{eq:due} inherits the relevant order properties of $\mathbf{T}$ and
will be useful below; in fact, it is closely related to the operator $%
\mathbf{A}$ of the next section.

Alternatively, we can make a further change of variable by setting 
\begin{equation*}
\psi _{t}=u_{t}^{r}=\left( \frac{U_{t}}{C_{t}}\right) ^{r}.
\end{equation*}%
The variable $\psi _{t}/(1-\beta )$ is known in the literature as the \emph{%
wealth-consumption ratio }(see \cite{EpsteinZin1989}). This gives 
\begin{equation*}
\psi _{t}=(1-\beta )+\beta \left( \mathbb{E}_{t}\left[ \psi
_{t+1}^{s/r}e^{sg_{t+1}}\right] \right) ^{r/s}.
\end{equation*}%
This formulation is used, in different forms, by Hansen and Scheinkman~\cite%
{HansenScheinkman2012} and Pohl, Schmedders, and Wilms~\cite%
{PohlSchmeddersWilms2024}. It is especially convenient because it can also
handle parameter regions not directly covered by the concavity argument
above.

For example, let $r\in (0,1)$ and $s>r$. The associated operator 
\begin{equation*}
\mathbf{S}(\psi )=(1-\beta )+\beta \left( \mathbb{E}_{t}\left[ \psi
^{s/r}e^{sg_{t+1}}\right] \right) ^{r/s}
\end{equation*}%
is monotone and convex. Set 
\begin{equation*}
M=\sup_{t,\omega }\left( \mathbb{E}_{t}e^{sg_{t+1}}\right) ^{r/s}.
\end{equation*}%
If $M<\beta ^{-1}$, then $\mathbf{S}$ maps any interval $[(1-\beta ),\gamma
] $ into itself whenever 
\begin{equation*}
\gamma \geq \overline{\gamma }:=\frac{1-\beta }{1-\beta M}.
\end{equation*}%
Moreover, for every $\gamma >\overline{\gamma }$, 
\begin{equation*}
\gamma -\mathbf{S}(\gamma )\geq \gamma \lbrack 1-\beta M]-(1-\beta
)=(1-\beta )\left( \frac{\gamma }{\overline{\gamma }}-1\right) >0.
\end{equation*}%
Thus no fixed point lies on the \emph{upper perimeter} of the interval $%
[(1-\beta ),\gamma ]$. By the dual formulation of Theorem~1 of \cite%
{marmont1}, it follows that the wealth-consumption recursion has a unique
solution in $\mathcal{K}(\mathbf{1})$ for $r\in (0,1)$ and $s>r$.

\bigskip

Another equivalent transformation is used by Borovi\v{c}ka and Stachurski~%
\cite[Section II]{BorovickaStachurski2020}. Setting $\theta
_{t}=(U_{t}/C_{t})^{s}$ yields 
\begin{equation*}
\theta _{t}=\left[ (1-\beta )+\beta \left( \mathbb{E}_{t}\left[ \theta
_{t+1}e^{sg_{t+1}}\right] \right) ^{r/s}\right] ^{s/r},
\end{equation*}%
which can be treated in the same way as the formulations above.

\bigskip

If now we consider the unit-EIS case, the Epstein--Zin recursion becomes 
\begin{equation*}
U_{t}=C_{t}^{1-\beta }\left( \mathbb{E}_{t}U_{t+1}^{s}\right) ^{\beta /s}.
\end{equation*}%
Defining $v_{t}=s\log \left( U_{t}/C_{t}\right) $, we have the formulation
that, under a convenient restriction, is carefully studied by Christensen~%
\cite{Christensen2022}: 
\begin{equation}
v_{t}=\beta \log \mathbb{E}_{t}\left[ \exp \left( v_{t+1}+sg_{t+1}\right) %
\right]  \label{eq:AUX}
\end{equation}%
as discussed further ahead.

\section{Application to macroeconomics}

\label{sec:kappa-model} We now specialize the abstract framework developed
above to standard macro-finance environments in which recursive utility is
commonly applied. We begin with a general Markov specification and then
consider four applications. These examples show how the general existence
and uniqueness conditions translate into explicit parameter restrictions.

So far, we have formulated intertemporal stochastic models in terms of an
abstract probability space $(\Omega ,\Sigma ,\mathbb{P})$ with information
given by an increasing filtration $\left( \Sigma _{t}\right) _{t\in T}.$ In
the following models the uncertainty is represented by a Markov process. It is thus convenient to provide a brief account of the notation
adopted in such a setting\footnote{%
An excellent description of the present subject can be found in \cite%
{StokeyLucasPrescott1989}. For a more complete mathematical analysis we
refer to \cite[Ch. V]{neveu}.}.

Let $\mathcal{X}$ be a Borel set of a Polish space and set $\Omega =\mathcal{%
X}^{\infty }$. The elements of $\mathcal{X}^{\infty }$ are denoted by $%
\omega \equiv \xi ^{\infty }=\left( \xi _{t}\right) _{t\geq 0}$. Denote by $%
\pi _{t}\left( \omega \right) =\xi _{t}$ the canonical projections for $%
t\geq 0$ and endow $\mathcal{X}^{\infty }$ with the canonical filtration $%
\left( \Sigma _{t}\right) _{t\in T}$ where $\Sigma _{t}$ is the least $%
\sigma $-algebra generated by the projections $\left( \pi _{s}\right)
_{0\leq s\leq t}$. The time $t$-history is given by $\xi ^{t}=\left( \xi
_{0},\xi _{1},...,\xi _{t}\right) $. Therefore an adapted Markov process $X$
is a function $X\left( \omega ,t\right) \equiv X_{t}\left( \xi ^{\infty
}\right) \equiv X\left( \xi ^{t}\right) $ for $t\in T$.

The probability $\mathbb{P}$ over $\Sigma $ is generated by an initial
distribution $\pi _{0}$ on $\mathcal{X}$ and the stationary transition
kernel $Q\left( \xi ,A\right) $ with the interpretation: 
\begin{equation*}
Q(\xi _{t},A)=\mathbb{P}\left\{ \xi _{t+1}\in A\text{ \ }\left\vert \text{ \ 
}\xi _{t}\right. \right\}
\end{equation*}%
for $A\in \mathcal{B}(\mathcal{X})$. In this setting, the recursive utility (%
\ref{eq:pref}) can be written as 
\begin{equation*}
U\left( \xi ^{t}\right) =\left[ (1-\beta )C^{r}\left( \xi ^{t}\right) +\beta
\left( \mathcal{M}_{t}(U\left( \xi ^{t},\cdot \right) )\right) ^{r}\right]
^{1/r}
\end{equation*}%
for all $t\geq 0$ and $\xi ^{t+1}\in \mathcal{X}^{t+1}.$ The associated
operator $\mathbf{T}$ is 
\begin{equation*}
\left( \mathbf{T}X\right) \left( \xi ^{t}\right) =W\left( C\left( \xi
^{t}\right) ,\mathcal{M}_{t}(X\left( \xi ^{t},\cdot \right) \right)
\end{equation*}%
It is convenient in this setting to restrict slightly the functional
framework. Unlike what was postulated in the previous sections, here the
measurable functions are functions and not equivalence classes.\footnote{%
As a matter of fact, we can also use here the space of classes of functions $%
L^{1}\left( \delta \right) $. But for some later constructions (for
instance, the function $h$ and the operator $\mathbf{A}$) we should require
the existence of a stationary distribution for the underlying Markov process.%
} This will be emphasized by using the calligraphic letters for the spaces
involved. For instance, $\mathcal{L}_{+}$ is the cone of the nonnegative,
finitely valued and adapted processes.

The analogue of the space $L^{1}\left( \delta \right) $ will be $\mathcal{L}%
^{1}\left( \delta \right) \subset \mathcal{L}$. A process $X$ belongs to $%
\mathcal{L}^{1}\left( \delta \right) $ if all the functions $X_{t}$ are
integrable, i.e., if 
\begin{equation*}
\left\Vert X_{t}\right\Vert =\int \left\vert X\left( \xi ^{t}\right)
\right\vert d\mathbb{P<\infty }
\end{equation*}%
and 
\begin{equation*}
\left\Vert X\right\Vert =\sum_{t\in T}\left\Vert X_{t}\right\Vert \delta
^{t}<\infty
\end{equation*}

Endowed with the point-wise order, the space $\mathcal{L}^{1}\left( \delta
\right) $ is Dedekind $\sigma $-complete and its positive cone is regular.
This suffices to apply the results of existence and uniqueness stated in
the previous section. Note, incidentally, that $\mathcal{L}^{1}\left( \delta
\right) $ is not Dedekind complete.

\subsection{Epstein--Zin certainty equivalent}

\label{subsec:kappa-general} Let us formulate the Epstein--Zin preferences,
already introduced in subsection \ref{section:EZ}, in a Markovian setting.

The state of the economy is given by pairs $\xi _{t}=\left(
z_{t},y_{t}\right) $ with $z_{t}\in \mathcal{Z}$ and $y_{t}\in \mathcal{Y}$,
where $\mathcal{Z}$ and $\mathcal{Y}$ are Borel sets of Polish spaces,
equipped with their Borel $\sigma $-algebras $\mathcal{B}\left( \mathcal{Z}%
\right) $ and $\mathcal{B}\left( \mathcal{Y}\right) $ respectively. The
variable $z$ represents the exogenous state of the system, while $y$ is the
shock due to the technological innovation.

Assume that the process $\xi _{t}$ is Markov and that the innovation
sequence $(y_{t})_{t\geq 1}$ is i.i.d. with law $\nu $, independent of $%
(z_{t})_{t\geq 0}$. Therefore, the probability $\mathbb{P}$ over $\Sigma $
is generated by the stationary transition kernel 
\begin{equation*}
\widetilde{Q}\left( \xi ,A\times B\right) =Q(z,A)\times \nu \left( B\right)
\end{equation*}%
with the following interpretation: given $\xi _{t}=\left( z_{t},y_{t}\right) 
$, 
\begin{equation*}
Q(z_{t},A)\times \nu \left( B\right) =\mathbb{P}\left\{ \xi _{t+1}\in
A\times B\right\} =\mathbb{P}\left\{ z_{t+1}\in A\right\} \times \nu \left(
B\right)
\end{equation*}%
for every $A\in \mathcal{B}(\mathcal{Z})$ and $B\in \mathcal{B}(\mathcal{Y})$%
.

Assume that consumption growth is given by the following specification (see
Hansen and Scheinkman~\cite{HansenScheinkman2012} and Borovi\v{c}ka and
Stachurski~\cite{BorovickaStachurski2020}) 
\begin{equation}
g_{t+1}:=\log \frac{C_{t+1}}{C_{t}}=\kappa (z_{t},z_{t+1},y_{t+1}),\qquad
t\geq 0,  \label{eq:XXX}
\end{equation}%
where $\kappa \colon \mathcal{Z}\times \mathcal{Z}\times \mathcal{Y}%
\rightarrow \mathbb{R}$ is Borel measurable. Given $\xi=(z,y)$ and $%
\xi^{\prime }=(z^{\prime },y^{\prime })$, we write $\kappa(\xi,\xi^{\prime
}) $ for $\kappa(z,z^{\prime },y^{\prime })$.

The Kreps--Porteus certainty equivalent will be 
\begin{equation}
\begin{aligned} \mathcal{M}_{t}\!\left(X(\xi^{t},\xi_{t+1})\right) &:=
\left( \mathbb{E}^{Q}\!\left[ X^{s}(\xi^{t},\xi_{t+1}) \mid \xi^{t}
\right]\right)^{1/s} \\ &= \left[ \int_{\mathcal{Z}}\int_{\mathcal{Y}}
X^{s}(\xi^{t},z^{\prime},y^{\prime})
\,\nu(dy^{\prime})\,Q(z_{t},dz^{\prime}) \right]^{1/s}. \end{aligned}
\label{eq:EZce}
\end{equation}

Given a consumption process $C\in \mathcal{L}_{++}$, the operator $\mathbf{T}%
\equiv \mathbf{T}_{C}$ is 
\begin{equation*}
\left( \mathbf{T}\left( X\right) \right) _{t}=W\left( C_{t},\mathcal{M}%
_{t}\left( X_{t+1}\right) \right)
\end{equation*}%
for all $t\in T$ and all infinite histories $\xi ^{\infty }=\left( \xi
_{0},\xi _{1},....\right) ,$ whenever $\mathcal{M}\left( X\right) \in 
\mathcal{L}_{+}$. The Epstein--Zin recursive utility is a fixed point $U$ of 
$\mathbf{T}$, i.e., $U=\left( U_{t}\right) _{t\in T}$ solves 
\begin{equation*}
U\left( \xi ^{t}\right) =\left[ (1-\beta) C^{r}\left( \xi ^{t}\right) +
\beta \left[ \int_{\mathcal{Z}}\int_{\mathcal{Y}}U^{s}(\xi ^{t},z^{\prime
},y^{\prime })\,\nu (dy^{\prime })\,Q(z_{t},dz^{\prime })\right] ^{r/s}%
\right] ^{1/r}
\end{equation*}%
for all $t\in T$ and all $\xi ^{\infty }.$

Before proceeding in the analysis, let us introduce a more general class of
certainty equivalents.

\subsection{Positively homogeneous and stationary certainty equivalents}

\label{sec:markovstationary}

In the Epstein--Zin case, the certainty equivalents are
positively homogeneous, i.e., $\mathcal{M}_{t}(\lambda X_{t+1})=\lambda 
\mathcal{M}_{t}(X_{t+1})$ for all scalars $\lambda >0$. In fact, the
following stronger property, which we call \emph{conditional positive
homogeneity}, holds. That is, 
\begin{equation*}
\mathcal{M}_{t}(\lambda _{t}X_{t+1})=\lambda _{t}\mathcal{M}_{t}(X_{t+1})%
\text{ \ }\forall t\in T,
\end{equation*}%
for every admissible $X\in \mathcal{L}_{++}$ and every strictly positive $%
\Sigma _{t}$-measurable function $\lambda _{t}$.

\bigskip

Another property exhibited by the above certainty equivalents is
stationarity. The certainty equivalent $\mathcal{M}$ is called \emph{%
stationary Markov} if for every admissible nonnegative Borel function $F:%
\mathcal{X}\times \mathcal{X}\longrightarrow \mathbb{R}_{+}$ and for every $%
s,t\geq 0$, we have%
\begin{equation*}
\pi _{t}(\xi ^{\infty })=\pi _{s}(\widetilde{\xi }^{\infty })\quad
\Longrightarrow \quad \mathcal{M}_{t}\left( F\left( \xi _{t},\cdot \right)
\right) =\mathcal{M}_{s}\!\left( F(\widetilde{\xi }_{s},\cdot )\right)
\end{equation*}%
Thus, the evaluation of a one-period Markov payoff depends on neither the
calendar date nor the preceding history once the current state is fixed.

Under these two additional assumptions, together with the growth condition (%
\ref{eq:XXX}), a particularly important result can be formulated.

To see this, observe that if $\log \left( C_{t+1}/C_{t}\right) =\kappa (\xi
_{t},\xi _{t+1}),$ then there exists a function $h:\mathcal{X}\rightarrow 
\mathbb{R}\cup \{+\infty \}$ such that 
\begin{equation}
\mathcal{M}_{t}(C_{t+1})=C_{t}h(\xi _{t}).  \label{eq:VVV}
\end{equation}%
Indeed, conditional positive homogeneity gives 
\begin{equation*}
\mathcal{M}_{t}(C_{t+1})=C_{t}\mathcal{M}_{t}\left( e^{\kappa (\xi
_{t},\cdot )}\right) ,
\end{equation*}%
while stationarity implies that $\mathcal{M}_{t}\left( e^{\kappa (\xi ,\cdot
)}\right) $ depends on the current state $\xi $ but not on calendar time.
Hence, defining 
\begin{equation*}
h(\xi ):=\mathcal{M}_{t}\left( e^{\kappa (\xi ,\cdot )}\right) ,
\end{equation*}%
yields~\eqref{eq:VVV}.

For instance, in the Epstein--Zin case the function $h$ is given by 
\begin{equation}
h(z):=\left\{ \int_{\mathcal{Z}}\left[ \int_{\mathcal{Y}}e^{s\kappa
(z,z^{\prime },y^{\prime })}\,\nu (dy^{\prime })\right] Q(z,dz^{\prime
})\right\} ^{1/s}.  \label{eq:GSL}
\end{equation}%
Consequently, using \eqref{eq:VVV} we have 
\begin{equation}
\mathcal{M}(C)\in \mathcal{L}_{+}\quad \Longleftrightarrow \quad h(z)<\infty
\quad \text{for every }z\in \mathcal{Z}.  \label{eq:CCC}
\end{equation}%
Suppose now that $C\in \mathcal{L}_{++}^{1}(\delta )$ for some $\delta >0$,
that \eqref{eq:CCC} holds, and $s\in (-\infty ,0)\cup (0,1]$. Since the
Epstein--Zin certainty equivalent is order concave, Propositions~\ref%
{prop:unique1} and \ref{prop:unique2} imply uniqueness of recursive utility
in $\mathcal{K}(C)$ under the following conditions:

\begin{itemize}
\item for $r<0$: 
\begin{equation*}
\inf_{z\in \mathcal{Z}}h(z)>\beta ^{1/\lvert r\rvert }.
\end{equation*}

\item for $r\in (0,1)$: 
\begin{equation*}
\sup_{z\in \mathcal{Z}}h(z)<\beta ^{-1/r}\qquad \text{and}\qquad \inf_{z\in 
\mathcal{Z}}h(z)>0.
\end{equation*}
\end{itemize}

Notice that the condition $\inf_{z\in \mathcal{Z}}h(z)>0$ is the
\textquotedblleft Markovian\textquotedblright\ translation of the condition $%
\mathcal{M}(C)/C\gg 0$ postulated in Proposition \ref{prop:unique2}.

The two properties of homogeneity and stationarity give rise to another
significant property regarding the structure of the fixed points.

We have seen that the utilities $U_{t}$ in the Epstein--Zin equation (\ref%
{eq:EPZIN}) admit a multiplicative representation $U_{t}=u_{t}C_{t}$ where $%
u_{t}$ solves (\ref{eq:due}). This representation turns out to be true in
general under the assumptions of conditional homogeneity and stationarity of
the operator $\mathcal{M}$. In fact, from the equation 
\begin{equation*}
U\left( \xi ^{t}\right) =W\left( C\left( \xi ^{t}\right) ,\mathcal{M}%
_{t}\left( U\left( \xi ^{t},\cdot \right) \right) \right)
\end{equation*}%
it follows 
\begin{equation*}
\frac{U\left( \xi ^{t}\right) }{C\left( \xi ^{t}\right) }=W\left( 1,\mathcal{%
M}_{t}\left( \frac{U\left( \xi ^{t},\cdot \right) }{C\left( \xi ^{t}\right) }%
\right) \right) =W\left( 1,\mathcal{M}_{t}\left( \frac{U\left( \xi
^{t},\cdot \right) }{C\left( \xi ^{t},\cdot \right) }e^{\kappa (\xi
_{t},\cdot )}\right) \right)
\end{equation*}

Consequently, the substitution $U\left( \xi ^{t}\right) =u\left( \xi
^{t}\right) C\left( \xi ^{t}\right) $ leads to the recursive equation%
\begin{eqnarray}
u\left( \xi ^{t}\right) &=&\Gamma _{r}\left( \mathcal{M}_{t}\left( u\left(
\xi ^{t},\cdot \right) e^{\kappa (\xi _{t},\cdot )}\right) \right)
\label{eq:EEM} \\
\Gamma _{r}\left( t\right) &=&\left( (1-\beta )+\beta t^{r}\right) ^{1/r} 
\notag
\end{eqnarray}%
We conjecture that, by virtue of the Markovian structure, the utility $U$
takes the \emph{separable} form 
\begin{equation}
U\left( \xi ^{t}\right) =\varphi \left( \xi _{t}\right) C\left( \xi
^{t}\right)  \label{eq:TZZ}
\end{equation}%
Namely, $u\left( \xi ^{t}\right) =\varphi \left( \xi _{t}\right) $ for a
certain function $\varphi :\mathcal{X}\rightarrow \mathbb{R}$. Said
differently, in view of (\ref{eq:EEM}), $\varphi $ solves%
\begin{equation}
\varphi \left( \xi _{t}\right) =\Gamma _{r}\left( \mathcal{M}_{t}\left(
\varphi \left( \cdot \right) e^{\kappa (\xi _{t},\cdot )}\right) \right)
\label{eq:BQY}
\end{equation}

To check this guess, let us introduce the space $B\left( \mathcal{X},%
\mathcal{B}(\mathcal{X})\right) $ of the bounded measurable functions
defined on $\mathcal{X}$. Equipped with the supnorm $\left\Vert \varphi
\right\Vert =\sup_{\mathcal{X}}\left\vert \varphi \left( \xi \right)
\right\vert $, this space is a Dedekind $\sigma $-complete Banach lattice.

For the sake of simplicity, set 
\begin{equation*}
\mathcal{M}_{t}\!\left( e^{\kappa (\xi _{t},\cdot )}\varphi (\cdot )\right)
=(\mathcal{H}\varphi )(\xi _{t}).
\end{equation*}%
Notice that $\mathcal{H}1=h$.

In light of equation (\ref{eq:BQY}), define the operator 
\begin{equation*}
\mathbf{A}:B_{+}\left( \mathcal{X},\mathcal{B}(\mathcal{X})\right)
\rightarrow B_{+}\left( \mathcal{X},\mathcal{B}(\mathcal{X})\right)
\end{equation*}

by%
\begin{equation*}
(\mathbf{A}\varphi )(\xi )=\Gamma _{r}\left( \mathcal{H}\varphi (\xi
)\right) 
\end{equation*}%
where 
\begin{equation*}
\Gamma _{r}\left( t\right) =\left( (1-\beta )+\beta t^{r}\right) ^{1/r}
\end{equation*}%
In other words, we assume that $\mathcal{H}\varphi $ depends only on $z$ for
every admissible $\varphi $, and identify $\mathbf{A}$ with its restriction
to functions of $z$. For instance, in the Epstein--Zin case, the operator $%
\mathbf{A}$ takes the explicit form 
\begin{equation*}
(\mathbf{A}\varphi )(z)=\left\{ (1-\beta )+\beta \left[ \int_{\mathcal{Z}%
}\int_{\mathcal{Y}}e^{s\kappa (z,z^{\prime },y^{\prime })}\varphi (z^{\prime
})^{s}\,\nu (dy^{\prime })\,Q(z,dz^{\prime })\right] ^{r/s}\right\} ^{1/r}.
\end{equation*}%
Below we list the main properties of $\Gamma $ and $\mathbf{A}$ that can be
easily checked.

\begin{itemize}
\item $\Gamma $ is increasing for all $r\in \mathbb{R}$ and concave for all $%
r\leq 1$

\item $\mathbf{A}$ is concave and increasing, provided $\mathcal{M}$ is
concave and increasing

\item $\mathbf{A}$ is $\sigma $-order continuous, if $\mathcal{M}$ is $%
\sigma $-order continuous

\item If $r\in \left( 0,1\right) $, then $\mathbf{A}\varphi \geq (1-\beta)
^{1/r} $ for all $\varphi \geq 0$. If $r<0$, then $0\leq \mathbf{A}\varphi
\leq (1-\beta) ^{1/r}$.
\end{itemize}

The next fixed-point result holds under assumptions (i)--(iii), together
with conditional homogeneity and stationarity. It follows directly from the
properties listed above and the existence and uniqueness arguments of
Section~\ref{sec:exandunique}.

\begin{proposition}
\label{prop:operatorA} Let $C\in \mathcal{L}_{++}^{1}(\delta )$ for some $%
\delta >0$. Under assumptions (i)--(iii), conditional homogeneity and
stationarity:

\begin{enumerate}
\item for $0<r<1$, set $\overline{h}:=\sup_{z\in \mathcal{Z}}h(z)$. Under
the condition 
\begin{equation}
\overline{h}<\beta ^{-1/r},  \label{eq;MMM}
\end{equation}%
 $\mathbf{A}$ maps the interval $[(1-\beta )^{1/r},b]$ into itself and
has a least and a greatest fixed point there, for every 
\begin{equation*}
b\geq \overline{b}:=\left( \frac{1-\beta }{1-\beta \overline{h}^{\,r}}%
\right) ^{1/r},
\end{equation*}%
In addition, if $\mathcal{M}$ is order concave, then the fixed point is
unique in $B_{++}(\mathcal{X},\mathcal{B}(\mathcal{X}))$.

\item for $r<0$, set $\underline{h}:=\inf_{z\in \mathcal{Z}}h(z)$ and
suppose $h(z)<\infty $ for every $z\in \mathcal{Z}$. Under 
\begin{equation*}
\underline{h}>\beta ^{1/|r|},
\end{equation*}%
 $\mathbf{A}$ maps $[b,(1-\beta )^{1/r}]$ into itself and has a least
and a greatest fixed point there, for every 
\begin{equation*}
0<b\leq \underline{b}:=\left( \frac{1-\beta }{1-\beta \underline{h}^{\,r}}%
\right) ^{1/r},
\end{equation*}
If $\mathcal{M}$ is order concave, the fixed point is unique in $\mathcal{K}%
(1)$.
\end{enumerate}
\end{proposition}

Unlike the uniqueness result for $\mathbf{T}$ in Section~\ref%
{sec:exandunique}, uniqueness of the fixed point of $\mathbf{A}$ does not
require the condition $\inf_{z\in \mathcal{Z}}h(z)>0$. Thus, even when $%
\inf_{z\in \mathcal{Z}}h(z)=0$, $\mathbf{A}$ may have a unique fixed point
while $\mathbf{T}$ may admit multiple fixed points. Hence, with non-unit
EIS, separability provides a natural selection criterion: the separable
recursive utility is unique in $\mathcal{K}(C)$ even when recursive utility
more generally need not be.

\bigskip

\textbf{Remark. }Let us outline an alternative route to handle the operator $%
\mathbf{A}$. Suppose, for simplicity, that $r\in \left( 0,1\right) $ and $%
\mathbf{A} $ is concave. The lower bound $\mathbf{A}\varphi \geq (1-\beta
)^{1/r}$ implies that 
\begin{equation*}
\mathbf{A}\left( \alpha \varphi \right) =\mathbf{A}\left( \alpha \varphi +%
\overline{\alpha }0\right) \geq \alpha \mathbf{A}\left( \varphi \right) +%
\overline{\alpha }\mathbf{A}\left( 0\right) \geq \alpha \mathbf{A}\left(
\varphi \right) +\overline{\alpha }(1-\beta )^{1/r}
\end{equation*}%
The property $\mathbf{A}\left( \alpha \varphi \right) \geq \alpha \mathbf{A}%
\left( \varphi \right) +\overline{\alpha }k$, for every $\alpha \in \left(
0,1\right) $ and for some $k>0,$ is called $k$\emph{-subconcavity} (at $0$);
see \cite{marmont1}. This in turn implies that the self-map $\mathbf{A}:%
\left[ 0,b\right] \rightarrow \left[ 0,b\right] $ for any given $b\geq 
\overline{b}$ is $p$-subhomogeneous for some $p\in \left( 0,1\right) $; see
Definition \ref{def:ssub} and \cite[Proposition 6]{marmont1}. Namely,%
\begin{equation*}
\mathbf{A}\left( \alpha \varphi \right) \geq \alpha ^{p}\mathbf{A}\left(
\varphi \right) \text{ for }\alpha \in \left( 0,1\right)
\end{equation*}%
If we introduce  the Thompson
distance $d_{T}$ (see \cite{tom} and \cite{marmont1}) in the cone $\mathcal{K}\left( 1\right) $, the map $\mathbf{A}:%
\left[ 0,b\right] \rightarrow \left[ 0,b\right] $ turns out to be a $p$%
-contraction, i.e. $d_{T}\left( \mathbf{A}\varphi _{1},\mathbf{A}\varphi
_{2}\right) \leq pd_{T}\left( \varphi _{1},\varphi _{2}\right).$ One can
then deduce the existence of a unique fixed point for the map $\mathbf{A}$
and also the convergence of the iterates in the norm $\left\Vert \cdot
\right\Vert .$

\bigskip

Proposition \ref{prop:operatorA} shows that our conjecture about the
multiplicative relation (\ref{eq:TZZ}) is true. This is formalized in the
next proposition, whose trivial proof is omitted.

\begin{proposition}
\label{prop:separable-A-implies-fixed-point} Let $0\neq r<1$, $C\in \mathcal{%
L}_{++}^{1}\left( \delta \right) $ and $\overline{\varphi }$ be a fixed
point of $\mathbf{A}$. Then the process $U\in \mathcal{L}_{++}^{1}\left(
\delta \right) $ defined by $U_{t}=\overline{\varphi }\left( z_{t}\right)
C_{t}$ for every $t\geq 0$ is a fixed point of $\mathbf{T}$.
\end{proposition}

Notice that the previous result does not require the fixed point $\varphi =%
\mathbf{A}\varphi $ to belong to $\mathcal{K}(\mathbf{1})$, nor does it
require $\mathcal{M}$ to be concave, for $U_{t}=\varphi (z_{t})C_{t}$ to be
a fixed point of $\mathbf{T}$. Clearly, $U\in \mathcal{K}(C)$ if and only if 
$\varphi \in \mathcal{K}(\mathbf{1})$.

\subsection{Cobb--Douglas aggregator\label{subsec:cobb}}

Finally, let us discuss the case of the Cobb--Douglas aggregator in the
Markov setting. It is now convenient to assume $\alpha +\beta
=1\Longleftrightarrow p=1$. We add the hypotheses of conditional positive
homogeneity and stationarity.

In view of Proposition \ref{prop:cobbdouglas}, if $\alpha +\beta =1$ then
there is a unique solution in $\mathcal{K}\left( C\right) $, provided $C\in
L_{++}^{1}\left( \delta \right) $ and $h\in \mathcal{K}\left( 1\right) $
where $h:\mathcal{Z}\rightarrow \mathbb{R}$ is given by $h(\xi _{t})=%
\mathcal{M}_{t}\left( e^{\kappa (\xi _{t},\cdot )}\right) $.

As we did earlier for the CES aggregator, we associate with $\mathbf{T}$ an
operator $\mathbf{A}$. In fact, if the recursive utility $U$ is a solution
to 
\begin{equation}
U\left( \xi ^{t}\right) =C^{\alpha }\left( \xi ^{t}\right) \left[ \mathcal{M}%
_{t}\left( U\left( \xi ^{t},\cdot \right) \right) \right] ^{\beta }
\label{eq:DDP}
\end{equation}%
the substitution $u\left( \xi ^{t}\right) =U\left( \xi ^{t}\right) /C\left(
\xi ^{t}\right) $ yields the recursion 
\begin{equation*}
u\left( \xi ^{t}\right) =\mathcal{M}_{t}^{\beta }\left[ u\left( \xi
^{t},\cdot \right) e^{\kappa \left( \xi _{t},\cdot \right) }\right]
\end{equation*}%
Hence the operator $\mathbf{A}$ will take the form 
\begin{equation*}
\mathbf{A}\varphi =\mathcal{M}_{t}^{\beta }\left[ \varphi \left( \cdot
\right) e^{\kappa \left( \xi _{t},\cdot \right) }\right]
\end{equation*}

For instance, in the Epstein--Zin case it is 
\begin{equation*}
(\mathbf{A}\varphi )(z)=\left( \int_{\mathcal{Z}}\int_{\mathcal{Y}%
}e^{s\kappa (z,z^{\prime },y^{\prime })}\varphi (z^{\prime })^{s}\,\nu
(dy^{\prime })\,Q(z,dz^{\prime })\right) ^{\beta /s}.
\end{equation*}%
Clearly the operator $\mathbf{A}$ turns out to be $\beta $-subhomogeneous.
Existence and uniqueness are then assured in the cone $\mathcal{K}(\mathbf{1})$,
provided $\varepsilon \leq h\left( z\right) \leq M$ for all $z\in \mathcal{Z}
$.

\bigskip

In the Epstein--Zin case we have seen (eq. \ref{eq:AUX} of Section \ref%
{section:EZ}) the alternative use of the transformation $v\left( \xi
^{t}\right) =s\log \left( U\left( \xi ^{t}\right) /C\left( \xi ^{t}\right)
\right) $ that gives the recursion 
\begin{equation}
v\left( \xi ^{t}\right) =\beta \log \mathbb{E}^{Q}\left[ \exp \left( v\left(
\xi ^{t},\xi _{t+1}\right) +s\kappa \left( \xi _{t},\xi _{t+1}\right)
\right) \left\vert \xi ^{t}\right. \right]  \label{eq:rec1}
\end{equation}%
Looking for the solutions of the kind $v\left( \xi ^{t}\right) \equiv
\varphi \left( \xi _{t}\right) $, we get the recursive equation 
\begin{equation}
\varphi \left( \xi \right) =\beta \log \mathbb{E}^{Q}\left[ \exp \left(
\varphi \left( \xi _{t+1}\right) +s\kappa \left( \xi _{t},\xi _{t+1}\right)
\right) \left\vert \xi _{t}=\xi \right. \right]  \label{eq:rec2}
\end{equation}%
extensively studied by Christensen \cite{Christensen2022}.

To clarify the deep difference between (\ref{eq:rec1}) and (\ref{eq:rec2}),
suppose that $U\in \mathcal{K}\left( C\right) $ is a solution to (\ref%
{eq:DDP}). Hence 
\begin{equation*}
U\left( \xi ^{t}\right) =C\left( \xi ^{t}\right) \overline{\varphi }\left(
\xi _{t}\right)
\end{equation*}%
for some $\overline{\varphi }\in \mathcal{K}\left( 1\right) .$ As discussed
in subsection \ref{subsection Unit} the utility $U$ generates a family of
utilities $V\left( \xi ^{t}\right) =\lambda ^{\beta ^{-t}}C\left( \xi
^{t}\right) \overline{\varphi }\left( \xi _{t}\right) $, with $\lambda \in
\left( 0,1\right] .$ Under the change of variables $v\left( \xi ^{t}\right)
=s\log \left( V\left( \xi ^{t}\right) /C\left( \xi ^{t}\right) \right) $ we
obtain the family of solutions to the recursion (\ref{eq:rec1}): 
\begin{equation*}
v\left( \xi ^{t}\right) =s\beta ^{-t}\log \lambda +s\log \overline{\varphi }%
\left( \xi _{t}\right)
\end{equation*}%
while the unique bounded measurable solution to (\ref{eq:rec2}) is given by $%
\varphi\left(\xi\right) =s\log\overline{\varphi}\left(\xi\right)$.

\bigskip Observe that the condition $h\in\mathcal{K}(\mathbf{1})$, which guarantees
existence and uniqueness in $\mathcal{K}(C)$, may fail even though fixed
points exist in other cones. Section~\ref{sec:bidder-smith-wachter} provides
an economically relevant example of this phenomenon. Moreover, Section~\ref%
{sec:long-run-risk-multiplicity} establishes multiplicity of separable
recursive utilities for the constant-volatility Case~I consumption
specification of Bansal and Yaron~\cite{BansalYaron2004} when the EIS is set
to one.

\subsection{Mehra--Prescott model}

We begin by illustrating our results with the simplest finite-state member
of the general Markov family described in the previous section (see Mehra
and Prescott \cite{MehraPrescott1985}). The state space is the finite set $%
\mathcal{Z}=\{1,\dots ,m\}$ and the Markov chain is generated by the
stochastic matrix $\Pi =\left( \pi _{ij}\right) _{i,j=1}^{m}$. The initial
distribution is the $m$-dimensional row vector $\pi _{0}=\left( \pi
_{0}^{1},\pi _{0}^{2},...,\pi _{0}^{m}\right) $ and 
\begin{equation*}
Q(i,\{j\})=\pi _{ij}.
\end{equation*}%
For each $i\in \{1,\dots ,m\}$, let $\ell _{i}>0$ denote the gross
consumption growth associated with state $i$. The consumption growth depends
only on the next state. Namely, $C_{t+1}=\ell _{i_{t+1}}C_{t}$ and 
\begin{equation*}
g_{t+1}=\kappa (i_{t},i_{t+1})=\log \ell _{i_{t+1}}.
\end{equation*}%
Thus the Mehra--Prescott economy is exactly the general Markov model with a
finite state space and no auxiliary innovation.

Here, we consider the Epstein--Zin certainty equivalent \eqref{eq:EZce}.
Therefore, in this case the function $h:\mathcal{Z}\rightarrow \mathbb{R}$
in (\ref{eq:GSL}) is given by 
\begin{equation*}
h(i)=\left( \sum_{j=1}^{m}\pi _{ij}\ell _{j}^{s}\right) ^{1/s},\qquad i\in 
\mathcal{Z}.
\end{equation*}%
Let $D:=\mathrm{diag}(\ell_1,\dots,\ell_m)$, and let $\mathbf{1}%
=(1,1,\dots,1)$ be the $m$-dimensional row vector of ones. Tedious algebra
leads to the moment condition 
\begin{equation*}
\mathbb{E}[C_{t}]=C_{0}\,\left\langle \pi _{0}(\Pi D)^{t},\mathbf{1}%
\right\rangle \text{ \ }\forall t\geq 0.
\end{equation*}%
Hence 
\begin{equation*}
C\in L_{++}^{1}\left( \delta \right) \iff \sum_{t\geq 0}\delta
^{t}\,\left\langle \pi _{0}(\Pi D)^{t},\mathbf{1}\right\rangle <\infty .
\end{equation*}%
If $\rho \left( A\right) $ denotes the spectral radius of a matrix $A,$ then
the condition $\delta \rho (\Pi D)<1$ suffices to get $C\in L_{++}^{1}\left(
\delta \right) $.

Thus, under the assumption $\delta \rho (\Pi D)<1$, we obtain the following
conditions for existence, uniqueness, and non-existence. If $0<r<1$, then 
\begin{equation*}
\max_i h(i)<\beta^{-1/r}
\end{equation*}
implies existence, uniqueness, and global attractivity, whereas $\min_i
h(i)\geq\beta^{-1/r}$ implies that there is no fixed point in $\mathcal{L}%
_{+}$. If $r<0$, then 
\begin{equation*}
\min_i h(i)>\beta^{1/|r|}
\end{equation*}
implies existence, uniqueness, and global attractivity, whereas $\max_i
h(i)\leq\beta^{1/|r|}$ implies that there is no fixed point in $\mathcal{L}%
_{++}$.

We deduce from Proposition~\ref{prop:separable-A-implies-fixed-point} that
the recursive utility has the separable form $U_{t}=\varphi (i_{t})C_{t},$
where the vector $\varphi =(\varphi _{1},\dots ,\varphi _{m})$ is the unique
solution to the system of non-linear equations 
\begin{equation*}
\varphi _{i}=\left[ (1-\beta )+\beta \left( \sum_{j=1}^{m}\pi _{ij}\ell
_{j}^{s}\varphi _{j}^{s}\right) ^{r/s}\right] ^{1/r},\qquad i=1,\dots ,m.
\end{equation*}%
For instance, in the calibration used by Weil~\cite{Weil1989} for the
Mehra--Prescott model, there are two states, $m=2$, with gross consumption
growth $\ell _{1}=1.054$, $\ell _{2}=0.984$ and transition matrix 
\begin{equation*}
\Pi =\left( 
\begin{array}{cc}
0.43 & 0.57 \\ 
0.57 & 0.43%
\end{array}%
\right) ,\qquad \pi _{0}=\left( 
\begin{array}{cc}
1/2 & 1/2%
\end{array}%
\right) .
\end{equation*}

We use Weil's calibration with $\beta=0.95$ and consider EIS $\in[%
0.1,1)\cup(1,\infty)$ and RRA $\in(0,1)\cup(1,10]$, corresponding to $%
r,s\in[-9,0)\cup(0,1)$. Over this parameter range, 
\begin{equation*}
1.008\leq h(i)\leq1.024,\qquad i=1,2.
\end{equation*}
Moreover, $\beta^{1/|r|}\leq0.95^{1/9}\simeq0.9943$ for $r<0$, while $%
\beta^{-1/r}>0.95^{-1}\simeq1.053$ for $r>0$. Hence the relevant existence
and uniqueness conditions are satisfied for every $r,s\in[-9,0)\cup(0,1)$,
and there is a unique recursive utility $U\in\mathcal{K}(C)$ throughout this
parameter range.

To illustrate what parameters would lead to non-existence, assume that $s=-9$%
. In the high-EIS case, a common $5\%$ increase in the two growth rates,
from $(\ell_1,\ell_2)=(1.054,0.984)$ to $(1.1067,1.0332)$, gives $\min_i
h(i)\simeq1.059$. Hence the non-existence condition holds whenever $\min_i
h(i)\geq0.95^{-1/r}$, that is, for 
\begin{equation*}
0.8964 \simeq -\frac{\log(0.95)}{\log(\min_i h(i))} \leq r<1.
\end{equation*}

In the low-EIS case, a common $3\%$ reduction, from $(1.054,0.984)$ to $%
(1.02238,0.95448)$, gives $\max_i h(i)\simeq0.9874$. Hence the low-EIS
non-existence condition holds whenever $\max_i h(i)\leq0.95^{1/|r|}$, that
is, for $-9\leq r\leq -\frac{\log(0.95)}{\log(\max_i h(i))} \simeq-4.041,$
over the parameter range of interest.

Finally, the unique fixed point can be found by solving the system of
equations 
\begin{equation*}
\begin{array}{rcl}
\varphi_1 & = & \left[ 0.05 + 0.95 \left( 0.43(1.054\varphi_1)^{s} +
0.57(0.984\varphi_2)^{s} \right)^{r/s} \right]^{1/r}, \\[0.5em] 
\varphi_2 & = & \left[ 0.05 + 0.95 \left( 0.57(1.054\varphi_1)^{s} +
0.43(0.984\varphi_2)^{s} \right)^{r/s} \right]^{1/r}.%
\end{array}%
\end{equation*}
For example, when $r=-1$ and $s=-9$, corresponding to $\mathrm{RRA}=10$ and $%
\mathrm{EIS}=1/2$, the unique solution is 
\begin{equation*}
\varphi_1\simeq1.263,\qquad \varphi_2\simeq1.273,
\end{equation*}
so that 
\begin{equation*}
U_t\simeq 
\begin{cases}
1.263\,C_t, & \text{if } i_t=1, \\ 
1.273\,C_t, & \text{if } i_t=2.%
\end{cases}%
\end{equation*}
To be clear, here and below, the approximation sign refers only to the reported numerical values of the coefficients.
\subsection{Hidden-regime consumption growth}

\label{subsec:two-state-hidden-regime} We now specialize the general Markov
consumption-growth model of Section~\ref{subsec:kappa-general} to a
two-state hidden-regime learning economy. Models of this kind are widely
used in macro-finance because they provide a parsimonious way to represent
persistent changes in expected growth and macroeconomic risk. Here we focus
on the model considered by Ju and Miao~\cite{JuMiao2012} (see also Lettau,
Ludvigson and Wachter~\cite{LettauLudvigsonWachter2008}). For this model,
Christensen~\cite{Christensen2022} establishes existence and uniqueness with
unit EIS, while Pohl et al.~\cite{PohlSchmeddersWilms2024} establish
existence with non-unit EIS; uniqueness in the latter case had remained
unresolved.

Consumption growth is determined by a regime that is not observed by the
agent. If the economy is in regime $i\in \{1,2\}$, then consumption growth
is given by 
\begin{equation*}
g_{t+1}=\log \frac{C_{t+1}}{C_{t}}=\bar{g}_{i}+\sigma \varepsilon
_{t+1},\qquad \varepsilon _{t+1}\sim N(0,1),\qquad i=1,2,
\end{equation*}%
where $\bar{g}_{1}>\bar{g}_{2}$ and $\sigma >0$. Thus, $1$ is the
high-growth regime, while regime $2$ is the low-growth regime. Here, the
shocks $(\varepsilon _{t})$ are i.i.d., normal and independent of the regime
process. The regime process follows a Markov chain with transition matrix 
\begin{equation*}
\Lambda =%
\begin{pmatrix}
\lambda _{11} & 1-\lambda _{11} \\ 
\lambda _{21} & 1-\lambda _{21}%
\end{pmatrix}%
,
\end{equation*}%
where $\lambda _{11}$ is the probability of remaining in regime $1$ and $%
\lambda _{21}$ is the probability of moving from regime $2$ to regime $1$;
see equation (13) in \cite{JuMiao2012}. Observe that $C_{t}$ has unbounded
support for every $t\geq 1$.

Because the agent observes consumption growth but not the underlying regime,
the relevant state variable is the agent's belief $\mu _{t}$ that the
next-period regime is the high-growth regime. In particular, the
hidden-regime model of Ju and Miao~\cite{JuMiao2012} can be recast into our
framework as a Markov economy with state $\mu _{t}$, as we formally show in
Supplemental Appendix~\ref{app:two-state-hidden-regime}. In the following,
we work directly with our Markov formulation.

In particular, we can take the state space to be $\mathcal{Z}=[\lambda
_{21},\lambda _{11}]:=K$ and, for every $\mu ^{\prime }\in (\lambda
_{21},\lambda _{11})$, the function $\kappa $ in \eqref{eq:XXX} is given by 
\begin{equation}
\kappa (\mu ,\mu ^{\prime }):=\frac{\bar{g}_{1}+\bar{g}_{2}}{2}+\frac{\sigma
^{2}}{\bar{g}_{1}-\bar{g}_{2}}\left[ \log \frac{\mu ^{\prime }-\lambda _{21}%
}{\lambda _{11}-\mu ^{\prime }}-\log \frac{\mu }{1-\mu }\right] .
\label{eq:kappa-two-state-explicit}
\end{equation}%

For $i=1,2$, let $f_{i}$ denote the density of $N(\bar{g}_{i},\sigma ^{2})$,
and define 
\begin{equation*}
J(\mu ,\mu ^{\prime }):=\frac{\sigma ^{2}(\lambda _{11}-\lambda _{21})}{(%
\bar{g}_{1}-\bar{g}_{2})(\mu ^{\prime }-\lambda _{21})(\lambda _{11}-\mu
^{\prime })}.
\end{equation*}%
For every Borel set $A\subseteq K$, define $Q$ by 
\begin{equation*}
Q(\mu ,A)=\int_{A}f_{\mu }\bigl(\kappa (\mu ,\mu ^{\prime })\bigr)J(\mu ,\mu
^{\prime })\,d\mu ^{\prime },
\end{equation*}%
where 
\begin{equation*}
f_{\mu }(g):=\mu f_{1}(g)+(1-\mu )f_{2}(g).
\end{equation*}

For $i=1,2$, define the regime-$i$ conditional kernel by 
\begin{equation*}
Q_i(\mu,d\mu^{\prime }) := \frac{ f_i\bigl(\kappa(\mu,\mu^{\prime })\bigr) }{
f_\mu\bigl(\kappa(\mu,\mu^{\prime })\bigr) } Q(\mu,d\mu^{\prime }).
\end{equation*}
Thus, $Q_i(\mu,\cdot)$ is the conditional distribution of $\mu_{t+1}$ given
that $\mu_t=\mu$ and that the economy is in regime $i$ in period $t+1$.

Finally, as we show in the Supplemental Appendix~\ref%
{app:two-state-hidden-regime}, a sufficient condition for $C\in \mathcal{L}%
_{++}^{1}(\delta)$ is 
\begin{equation}  \label{eq:hiddencondition}
\delta\exp\left(\frac12\sigma^2\right) \left[ \lambda_{11} e^{\bar g_1}
+(1-\lambda_{11})e^{\bar g_2} \right]<1.
\end{equation}

We are going to study the recursive utility in this model under three
distinct preference specifications: smooth ambiguity, Hansen--Sargent robust
control, and multiple priors. We focus on the calibration considered by Ju
and Miao~\cite{JuMiao2012}.

\paragraph{Smooth ambiguity and preferences.}

Smooth ambiguity is a popular model of ambiguity (e.g., see \cite%
{KlibanoffMarinacciMukerji2005,KlibanoffMarinacciMukerji2009,HayashiMiao2011}%
). We follow the specification in Ju and Miao~\cite[eq.~(8)]{JuMiao2012}.
For $i=1,2$, define the certainty equivalent conditional on regime $i$ by 
\begin{equation*}
\mathcal{M}_{t}^{i}(C_{t+1})=\left[ \int_{K}\left( {C_{t}e^{\kappa (\mu ,\mu
^{\prime })}}\right) ^{s}Q_{i}(\mu ,d\mu ^{\prime })\right] ^{1/s}.
\end{equation*}%
The smooth-ambiguity certainty equivalent is then 
\begin{equation*}
\mathcal{M}_{t}(C_{t+1})=\left\{ \mu _{t}\left[ \mathcal{M}_{t}^{1}(C_{t+1})%
\right] ^{\eta }+(1-\mu _{t})\left[ \mathcal{M}_{t}^{2}(C_{t+1})\right]
^{\eta }\right\} ^{1/\eta }.
\end{equation*}%
Thus, the evaluation proceeds in two stages. First, risk is evaluated
conditional on each possible regime, with $\mathrm{RRA}=1-s$ reflecting risk
aversion. These regime-specific certainty equivalents are then combined
using the current probabilities $\mu _{t}$ and $1-\mu _{t}$, with $1-\eta $
reflecting the sensitivity to ambiguity about which regime will prevail. In
particular, the case $\eta =s$ reduces to the Epstein--Zin certainty
equivalent. It is immediate to verify that the smooth-ambiguity certainty
equivalent is monotone, positively homogeneous, stationary Markov, $\sigma $%
-order continuous, and when $0\neq \eta ,s\leq 1$ it is also order concave.

For this certainty equivalent (see Appendix \ref%
{app:ju-miao-integral-equation}) the function $h:K\rightarrow \mathbb{R}$
can be computed in closed form and is given by 
\begin{equation*}
h(\mu )=\left\{ \mu \exp \left[ \eta \left( \bar{g}_{1}+\frac{1}{2}s\sigma
^{2}\right) \right] +(1-\mu )\exp \left[ \eta \left( \bar{g}_{2}+\frac{1}{2}%
s\sigma ^{2}\right) \right] \right\} ^{1/\eta }.
\end{equation*}

The calibration of Ju and Miao~\cite{JuMiao2012} is summarized in Table~\ref%
{tab:ju-miao-calibration}. Under this calibration, the belief support is $%
K=[\lambda _{21},\lambda _{11}]=[0.484,0.978]$. 
\begin{table}[h]
\begin{center}
\begin{tabular}{lll}
\hline
Parameter & Value & Description \\ \hline
$\lambda_{11}$ & $0.978$ & Persistence of the high-growth regime \\ 
$\lambda_{22}$ & $0.516$ & Persistence of the low-growth regime \\ 
$\bar g_1$ & $0.02251$ & Mean growth in the high-growth regime \\ 
$\bar g_2$ & $-0.06785$ & Mean growth in the low-growth regime \\ 
$\sigma$ & $0.03127$ & Conditional standard deviation of growth \\ 
$\beta$ & $0.975$ & Continuation weight \\ 
$r$ ($\mathrm{EIS}$) & $1/3$ ($3/2$) & CES parameter (EIS) \\ 
$s$ ($\mathrm{RRA}$) & $-1$ ($2$) & Risk aversion power (RRA) \\ 
$\eta$ ($1-\eta$) & $-7.864$ ($8.864$) & Ambiguity aversion power (ambiguity
aversion) \\ \hline
\end{tabular}%
\end{center}
\caption{Ju and Miao calibration.}
\label{tab:ju-miao-calibration}
\end{table}

For these parameter values, 
\begin{equation*}
\exp\left(\frac12\sigma^2\right) \left[ 0.978e^{\bar g_1}+(1-0.978)e^{\bar
g_2} \right] \simeq1.021.
\end{equation*}
Therefore, from \eqref{eq:hiddencondition}, we obtain $C\in \mathcal{L}%
_{++}^{1}(\delta)$ whenever $0<\delta<1.021^{-1}\simeq 0.9794$.

Let us verify the existence and uniqueness conditions. The function $h$ is
continuous and strictly positive on the compact set $K$. Consequently, it is
bounded above and bounded away from zero, so the existence condition also
implies the additional lower bound required for uniqueness. Since $r=1/3>0$,
the relevant existence condition is 
\begin{equation*}
h(\mu )<\beta ^{-1/r}=0.975^{-3}\simeq 1.079\qquad \text{for every }\mu \in
K.
\end{equation*}

As shown in Supplemental Appendix~\ref{app:ju-miao-integral-equation}, $h$
is increasing, and hence 
\begin{equation*}
h(\mu )\leq h(0.978)\simeq 1.019<\beta ^{-1/r},
\end{equation*}%
so that the existence condition is satisfied. Hence the smooth-ambiguity
recursion has a unique globally attractive fixed point in $\mathcal{K}(C)$.

The non-existence condition for $0<r<1$ is 
\begin{equation*}
h(\mu )\geq \beta ^{-1/r}\simeq 1.079\qquad \text{for every }\mu \in K.
\end{equation*}%
It can be satisfied with higher mean growth and a less persistent low-growth
regime. For instance, suppose that $\lambda _{21}=0.800$, and hence $\lambda
_{22}=0.200$, thereby lowering the probability of remaining in the
low-growth regime from $0.516$ to $0.200$. Raising both mean log-growth
rates from $(\bar{g}_{1},\bar{g}_{2})=(0.02251,-0.06785)$ to $%
(0.10251,0.01215)$ increases gross growth in each regime by about $8.3\%$,
in which case we obtain 
\begin{equation*}
\min_{\mu \in K}h(\mu )\simeq 1.081>\beta ^{-1/r}\simeq 1.079,
\end{equation*}%
implying that with this change in parameters there is no recursive utility
in $\mathcal{L}_{+}$.

Finally, by Proposition~\ref{prop:separable-A-implies-fixed-point} the
unique fixed point in $\mathcal{K}\left( C\right) $ takes the form 
\begin{equation*}
U_{t}=\varphi (\mu _{t})C_{t},
\end{equation*}
where $\varphi:K\to{\mathbb{R}}_{++}$ is the unique positive solution of 
\begin{equation*}
\begin{aligned} \varphi(\mu) &= \Bigg\{ (1-\beta) + \beta \Bigg[ \mu \left(
\int_K \left( e^{\kappa(\mu,\mu')} \varphi(\mu') \right)^s Q_1(\mu,d\mu')
\right)^{\eta/s} \\ &\qquad\qquad\qquad + (1-\mu) \left( \int_K \left(
e^{\kappa(\mu,\mu')} \varphi(\mu') \right)^s Q_2(\mu,d\mu') \right)^{\eta/s}
\Bigg]^{r/\eta} \Bigg\}^{1/r}, \end{aligned}
\end{equation*}
for every $\mu\in K$.


\paragraph{Hansen--Sargent robust-control preferences.}

Let us study the recursive utility under Hansen--Sargent robust-control
preferences (unit EIS), again under the calibration of Ju and Miao~\cite%
{JuMiao2012}. Thus, the only change relative to the previous application is
that we set $r=0$. The CES aggregator is then replaced by $W(x,y)=x^{1-\beta
}y^{\beta }$ and the recursive utility operator is then 
\begin{equation*}
(\mathbf{T}X)_{t}=C_{t}^{1-\beta }\left[ \mathcal{M}_{t}(X_{t+1})\right]
^{\beta }.
\end{equation*}%
Hence a recursive utility is a fixed point $U=\mathbf{T}U$, or, equivalently
by taking logarithms 
\begin{equation*}
\log U_{t}=(1-\beta )\log C_{t}+\beta \log \mathcal{M}_{t}(U_{t+1}).
\end{equation*}%
This is the unit-EIS recursion associated with Hansen--Sargent
robust-control preferences \cite{HansenSargent2001}; see also equation~(9)
in Ju and Miao~\cite{JuMiao2012}.

By Proposition \ref{prop:cobbdouglas}, first note that as before the
condition $C\in L_{++}^{1}(\delta )$ holds whenever $\delta <0.9794$.
Therefore, the condition for existence and uniqueness is $\mathcal{M}(C)\in 
\mathcal{K}(C).$ We have that 
\begin{equation*}
\mathcal{M}_{t}(C_{t+1})=C_{t}h(\mu _{t}),
\end{equation*}%
and 
\begin{equation*}
\min_{\mu }h(\mu )\simeq 0.9681,\qquad \max_{\mu }h(\mu )\simeq 1.019.
\end{equation*}%
Thus, we infer that $\mathcal{M}(C)\in \mathcal{K}(C)$ is true and so
existence and uniqueness hold in $\mathcal{K}\left( C\right) .$

\paragraph{Multiple priors.}

The multiple-priors specification is obtained as the limiting case of smooth
ambiguity when ambiguity aversion becomes infinite. In our parametrization,
this corresponds to $\eta\rightarrow-\infty$ (see equation (11) in Ju and
Miao \cite{JuMiao2012}). We keep all the remaining parameters at the values
reported in Table~\ref{tab:ju-miao-calibration}. Since the consumption
process is unchanged, the sufficient condition for $C\in \mathcal{L}%
_{++}^{1}(\delta)$ is unchanged as well.

The multiple-priors certainty equivalent $\mathcal{M}^{\mathrm{MP}}$ is
given by 
\begin{equation*}
\mathcal{M}_t^{\mathrm{MP}}(C_{t+1}) = \min_{i\in\{1,2\}} \mathcal{M}%
_{t}^i(C_{t+1})=\min_{i\in\{1,2\}} \left[ \int_K
\left(C_te^{\kappa(\mu_t,\mu_{t+1})}\right)^s Q_i(\mu_t,d\mu^{\prime }) %
\right]^{1/s}
\end{equation*}
The agent therefore evaluates continuation utility under the least favorable
regime. For $0 \neq s\leq 1$, this operator is monotone, positively
homogeneous, stationary Markov, $\sigma$-order continuous, and order concave.

Since $\bar{g}_{1}>\bar{g}_{2}$, regime $2$ is the least favorable regime
for next-period consumption. A direct calculation gives 
\begin{equation*}
\mathcal{M}_{t}^{\mathrm{MP}}(C_{t+1})=C_{t}\min_{i\in \{1,2\}}\exp \left( 
\bar{g}_{i}+\frac{1}{2}s\sigma ^{2}\right) \\
=C_{t}\exp \left( \bar{g}_{2}+\frac{1}{2}s\sigma ^{2}\right) =C_{t}h^{%
\mathrm{MP}},
\end{equation*}%
where 
\begin{equation*}
h^{\mathrm{MP}}:=\exp \left( \bar{g}_{2}+\frac{1}{2}s\sigma ^{2}\right) .
\end{equation*}%
The function $h^{\mathrm{MP}}$ does not depend on the state variable. As $%
s=-1$, $\beta =0.975$ and $r=1/3$, we get 
\begin{equation*}
0<h^{\mathrm{MP}}=\exp \left( -0.06785-\frac{1}{2}(0.03127)^{2}\right)
\simeq 0.9339<\beta ^{-1/r}\simeq 1.079.
\end{equation*}%
Hence the sufficient conditions for existence and uniqueness are satisfied,
and so there is a unique globally attractive fixed point in $\mathcal{K}(C)$.

Finally, consider a common increase $d>0$ in the two mean log-growth rates:
\begin{equation*}
(\bar{g}_{1},\bar{g}_{2})\longmapsto (\bar{g}_{1}+d,\bar{g}_{2}+d).
\end{equation*}%
By Proposition~\ref{prop:nonexist-positive}, no recursive utility exists in $%
\mathcal{L}_{+}$ whenever 
\begin{equation*}
e^{d}h^{\mathrm{MP}}\geq \beta ^{-1/r}\iff d\geq \log \left( \frac{\beta
^{-1/r}}{h^{\mathrm{MP}}}\right) \simeq 0.1443,
\end{equation*}%
which corresponds approximately to a $15.5\%$ increase in both growth rates.
Note that, unlike in the smooth-ambiguity case, this threshold does not
depend on the persistence of the low-growth regime, since $h^{\mathrm{MP}}$
is unaffected by the transition probabilities.

\subsection{Epstein--Zin with rare disasters}

\label{sec:bidder-smith-wachter} We next examine a rare-disaster economy
with persistent disaster intensity studied by Bidder and Smith~\cite%
{BidderSmith2018} (see also Wachter~\cite{Wachter2013}), with the difference
that we adopt Epstein--Zin preferences with a non-unit elasticity of
intertemporal substitution. This example illustrates the distinction between
uniqueness of recursive utility, in the general setting, and uniqueness
within the separable class. With non-unit EIS, our existence conditions hold
while the conditions guaranteeing uniqueness of recursive utility fail. Nevertheless, our results imply a unique separable recursive utility in $\mathcal{K}(C)$. By contrast, at unit EIS the same environment admits two
distinct separable recursive utilities.


In this model, consumption growth satisfies 
\begin{equation*}  \label{eq:rare-disaster-growth}
\log\frac{C_{t+1}}{C_t} = \nu_g +\nu_j j_{t+1} +\sigma_j\sqrt{j_{t+1}}%
\,\eta_{t+1} +\sigma\varepsilon_{t+1},
\end{equation*}
where $j_{t+1}$ denotes the number of disasters occurring between times $t$
and $t+1$, while $\eta_{t+1}$ and $\varepsilon_{t+1}$ are independent
standard normal random variables. The parameter $\nu_j<0$ is the average
effect of one disaster, and because $\nu_j<0$, each additional disaster
lowers expected consumption growth.

In particular, the number of disasters at time $t+1$ is distributed
according to a Poisson distribution with parameter $H_{t}$, that is, $%
j_{t+1}\mid H_{t}\sim \text{Poisson}\left( H_{t}\right) $, independently of $%
\eta_{t+1}$ and $\varepsilon _{t+1}$. In words, the number of disasters
depends on the Poisson intensity parameter $H_{t}$, which is the state
variable in this model. We therefore refer to $H_t$ as disaster intensity.

Observe that if we consider an i.i.d. sequence $(u_{t})_{t\geq 0}$ such that 
$u_{t+1}\sim \mathrm{Unif}(0,1)$ independent of $H_{t}$, and define 
\begin{equation*}
J(H,u):=\min \left\{ j\in \{0,1,2,\ldots \}:\sum_{n=0}^{j}e^{-H}\frac{H^{n}}{%
n!}\geq u\right\} ,
\end{equation*}%
then the realized number of disasters can be written as $%
j_{t+1}=J(H_{t},u_{t+1})$.\footnote{%
To see this, fix $H_{t}\geq 0$ and let $F_{H_{t}}(j):=%
\sum_{n=0}^{j}e^{-H_{t}}{H_{t}}^{n}/n!$ be the cumulative distribution
function of a $\text{Poisson}(H_{t})$ random variable. By the definition of $%
J$, the event $\{J(H_{t},u_{t+1})\leq j\}$ is equivalent to $\{u_{t+1}\leq
F_{H_{t}}(j)\}$. Hence, if $u_{t+1}\sim \mathrm{Unif}(0,1)$, then $\mathbb{P}%
(J(H_{t},u_{t+1})\leq j)=F_{H_{t}}(j)$, so $J({H_{t}},u_{t+1})\sim \text{%
Poisson}({H_{t}})$. Thus, we can write $j_{t+1}=J(H_{t},u_{t+1})$.} In this
case the function $\kappa $ is given by 
\begin{equation*}
\kappa (H,H^{\prime },y):=\nu _{g}+\nu _{j}J(H,u)+\sigma _{j}\sqrt{J(H,u)}%
\,\varepsilon _{z}+\sigma \varepsilon _{g},
\end{equation*}%
for every $y=(u,\varepsilon _{z},\varepsilon _{g})\in \mathcal{Y}%
:=(0,1)\times \mathbb{R}\times \mathbb{R}.$

In particular, the disaster intensity $H_{t}$ follows a first-order
autoregressive gamma process (see Gouri\'{e}roux and Jasiak \cite[eq.~(2.1)
and Definition~1]{GourierouxJasiak2006}). Therefore, the density of the
transition kernel $Q$ is 
\begin{equation}
q(H,H^{\prime })=e^{-\phi H/c}\sum_{n=0}^{\infty }\frac{(\phi H/c)^{n}}{n!}%
\frac{(H^{\prime })^{d_{H}+n-1}e^{-H^{\prime }/c}}{c^{d_{H}+n}\Gamma
(d_{H}+n)},\qquad H,H^{\prime }\in \mathbb{R}_{+},
\label{eq:bsw-transition-kernel}
\end{equation}%
where $d_{H}>0$, $c>0$, $0\leq \phi <1$, and $\Gamma (\cdot )$ denotes
Euler's gamma function. Supplemental Appendix~\ref{app:ez-rare-disasters}
verifies that 
\begin{equation*}
\mathbb{E}\left[ H_{t+1}\mid H_{t}=H\right] =cd_{H}+\phi H.
\end{equation*}
Notice that the parameter $\phi $ determines the persistence of disaster
intensity.

As mentioned, we use the Epstein--Zin certainty equivalent 
\begin{equation*}
\mathcal{M}_{t}(X_{t+1})=\left( \mathbb{E}[X_{t+1}^{s}\mid H_{t}]\right)
^{1/s}.
\end{equation*}%
Consequently the function $h:\mathbb{R}_{+}\rightarrow \mathbb{R}$ satisfies 
\begin{equation}
h(H):=\frac{\mathcal{M}_{t}(C_{t+1})}{C_{t}}=\left( \mathbb{E}\left[
e^{sg_{t+1}}\mid H_{t}=H\right] \right) ^{1/s}.
\label{eq:rare-disaster-h-definition}
\end{equation}%
As shown in Supplemental Appendix~\ref{app:ez-rare-disasters}, 
\begin{equation*}
h(H)=B_{s}e^{A_{s}H},  \label{eq:rare-disaster-h}
\end{equation*}%
where 
\begin{equation*}
B_{s}=\exp \left( \nu _{g}+\frac{1}{2}s\sigma ^{2}\right) ,\qquad A_{s}=%
\frac{\exp \left( s\nu _{j}+\frac{1}{2}s^{2}\sigma _{j}^{2}\right) -1}{s}.
\end{equation*}%
In the calibration considered below, $A_{s}<0$, so $h(H)$ decreases with $H$%
: states with greater disaster risk have lower \textquotedblleft
expected\textquotedblright\ growth.

Supplemental Appendix~\ref{app:ez-rare-disasters} also shows that a
sufficient condition for $C\in \mathcal{L}^1_{++}(\delta)$ is 
\begin{equation*}  \label{eq:rare-disaster-admissibility}
\nu_j+\frac{1}{2}\sigma_j^2\leq 0 \qquad\text{and}\qquad 0<\delta <
\min\left\{ 1, \exp\left( -\nu_g-\frac{1}{2}\sigma^2 \right) \right\}.
\end{equation*}

We now verify the existence and uniqueness conditions under the standard
monthly rare-disaster calibration of Bidder and Smith~\cite{BidderSmith2018}%
. The only preference parameter we change is the $\mathrm{EIS}$, which we
set equal to $1.9$ (equivalently $r=9/19$). The parameter values are
reported in Table~\ref{tab:ez-rd-calibration}.

\begin{table}[htbp]
\centering
\begin{tabular}{lll}
\hline
Parameter & Value & Description \\ \hline
$\nu_g$ & $1.50\times 10^{-3}$ & Mean non-disaster growth component \\ 
$\sigma^2$ & $3.30\times 10^{-5}$ & Variance of the Gaussian growth shock \\ 
$\nu_j$ & $-0.30$ & Mean log effect of one disaster \\ 
$\sigma_j$ & $0.10$ & Standard deviation of an individual disaster effect \\ 
$\phi$ & $0.993$ & Persistence of disaster intensity \\ 
$d_H$ & $1.788$ & Shape parameter of the ARG intensity process \\ 
$c$ & $7.766\times 10^{-6}$ & Scale parameter of the ARG intensity process
\\ 
$\beta$ & $0.9992$ & Continuation weight \\ 
$r$ ($\mathrm{EIS}$) & $9/19$ ($19/10$) & CES parameter (EIS) \\ 
$s$ ($\mathrm{RRA}$) & $-2.78$ ($3.78$) & Risk aversion power (RRA) \\ \hline
\end{tabular}%
\caption{Rare-disaster calibration, with economic counterparts in
parentheses.}
\label{tab:ez-rd-calibration}
\end{table}

For these parameters, 
\begin{equation*}
\nu_j+\frac{1}{2}\sigma_j^2 = -0.295<0,
\end{equation*}
so that $C\in\mathcal{L}_{++}^{1}(\delta)$ for every $0<\delta<
\min\left\{1,e^{-\nu_g-\frac12\sigma^2}\right\} \simeq0.9985.$ Moreover, 
\begin{equation*}
B_s=\exp\left(\nu_g+\frac12s\sigma^2\right) \simeq1.00146, \qquad
A_s\simeq-0.5012<0.
\end{equation*}

As $r=9/19\in (0,1)$, the relevant existence condition is 
\begin{equation*}
\sup_{H\geq 0}h(H)<\beta ^{-1/r}.
\end{equation*}%
Since $A_{s}<0$, the function $h$ is decreasing and its supremum is attained
at $H=0$. Therefore, 
\begin{equation*}
\sup_{H\geq 0}h(H)=h(0)=B_{s}\simeq 1.00146<\beta
^{-1/r}=0.9992^{-19/9}\simeq 1.00169.
\end{equation*}%
Thus, the existence condition guarantees at least one recursive utility in $%
\mathcal{K}(C)$.

The uniqueness condition of Proposition~\ref{prop:unique2}, however, fails.
Since $A_s<0$, it follows that $\lim_{H\uparrow \infty}h(H) =0 $ and hence 
\begin{equation*}
\inf_{H\geq0}h(H)=\inf \frac{\mathcal{M}(C)}{C}=0.
\end{equation*}
Economically, the state space contains arbitrarily adverse states: as $H$
increases, the expected number of disasters becomes arbitrarily large and
certainty-equivalent consumption growth approaches zero. Thus, recursive
utility exists, but we are not able to claim that there is a unique
recursive utility $U\in \mathcal{K}(C)$ satisfying $\mathbf{T}U=U$.

However, since all the assumptions of Proposition~\ref{prop:operatorA} hold,
we are able to claim that the operator $\mathbf{A}$ defined by 
\begin{equation*}
(\mathbf{A}\varphi )(H)=\left\{ (1-\beta )+\beta \left[ h(H)^{s}\int_{%
\mathbb{R}_{+}}\varphi ^{s}(H^{\prime })Q(H,dH^{\prime })\right]
^{r/s}\right\} ^{1/r},  \label{eq:rare-disaster-multiplier}
\end{equation*}
has a unique fixed point, and therefore there is a unique separable
recursive utility. Formally, if $\varphi \in \mathcal{K}(\mathbf{1})$ is the
unique fixed point of $\mathbf{A}$, then by Proposition \ref%
{prop:separable-A-implies-fixed-point} the process defined by $U_{t}=\varphi
(H_{t})C_{t}\in \mathcal{K}(C)$ is a recursive utility.

\paragraph{Comparison with the unit-EIS case.}

It is useful to contrast this uniqueness result with the limit case $r=0$.
Since $h(H)\rightarrow 0$ as $H\rightarrow \infty $, we have that $h\notin 
\mathcal{K}(\mathbf{1})$ and consequently $\mathcal{M}_{t}\left(
C_{t+1}\right) =h_{t}C_{t}\notin \mathcal{K}(C)$. Proposition~\ref%
{prop:cobbdouglas} implies that no fixed point belongs to $\mathcal{K}(C)$.
Nevertheless, Proposition~\ref{prop:FFP} allows us to identify other
admissible recursive utilities. As shown in Appendix \ref%
{app:ez-rare-disasters}, under the calibration of Table~\ref%
{tab:ez-rd-calibration}, equation~\eqref{eq:DDD} admits the two solutions 
\begin{equation*}
Z_{t}^{(1)}\simeq C_{t}e^{-83.39H_{t}},\quad Z_{t}^{(2)}\simeq
C_{t}e^{-278.1H_{t}}.
\end{equation*}%
For both solutions, substitution into~\eqref{eq:DDD} yields a constant $%
B_{t}^{(i)}\equiv \bar{B}_{i}\in \mathcal{K}(\mathbf{1})$, so that, thanks to %
\eqref{eq:LSX}, we obtain the unique fixed point in each $\mathcal{K}%
(Z^{(i)}),$ given by 
\begin{equation*}
U_{t}^{(1)}\simeq C_{t}e^{\,0.3686-83.39H_{t}},\quad U_{t}^{(2)}\simeq
C_{t}e^{-3.022-278.1H_{t}}.
\end{equation*}%
These coincide with the two affine solutions found by Christensen~\cite[%
Proposition~2.2]{Christensen2022}. Observe that $\mathcal{K}(Z^{(1)})\cap 
\mathcal{K}(Z^{(2)})=\varnothing $. Thus, surprisingly, with non-unit EIS
there is a unique separable recursive utility in $\mathcal{K}(C)$, whereas
for unit EIS two separable recursive utilities exist in distinct cones.



\subsection{Long-run risk and multiplicity at unit EIS}

\label{sec:long-run-risk-multiplicity}

Finally, we apply the results of Sections~\ref{subsection Unit} and \ref%
{subsec:cobb} to establish multiplicity of separable recursive utilities
under the constant-volatility specification of Bansal and
Yaron~\cite{BansalYaron2004}. The proofs are provided in the Supplemental
Appendix~\ref{app:multiplicity}. Consider the consumption-growth model 
\begin{equation}
\log \frac{C_{t+1}}{C_{t}}=a+X_{t}+\sigma _{t}Y_{t+1}\text{ and }%
X_{t+1}=\rho X_{t}+\chi\sigma _{t}\varepsilon _{t+1}^{x},
\label{eq:auto}
\end{equation}%
where $|\rho |<1$, $\chi>0$, and $(\sigma _{t})$ is a
time-homogeneous Markov process taking values in $(0,\infty )$. Here, the
sequences $(\varepsilon _{t}^{x})$ and $(Y_{t})$ are each i.i.d., mutually
independent, and independent of $(\sigma _{t})$ as well. We assume that $%
\varepsilon _{1}^{x}$ has full support on $\mathbb{R}$, so that $X_{t}$ has
full support on $\mathbb{R}$ for every $t\geq 1$.

Set $r=0$, $\alpha +\beta =1$ and take the Epstein--Zin certainty equivalent with 
$s<0$. Suppose that $\mathbb{E}[e^{s\sigma Y_{1}}]<\infty $ holds for every $%
\sigma >0$. Under these assumptions the function $h$ is given by 
\begin{equation*}
h(x,\sigma )=e^{a+x}\bigl(\mathbb{E}[e^{s\sigma Y_{1}}]\bigr)^{1/s}.
\end{equation*}%
Since $h(x,\sigma )\rightarrow \infty $ as $x\rightarrow \infty $, it
follows that $h$ is unbounded above. Consequently the condition $h\in 
\mathcal{K}(\mathbf{1})$ of Section~\ref{subsec:cobb} fails. Whenever $C\in \mathcal{L%
}_{++}^{1}(\delta )$, Proposition~\ref{prop:cobbdouglas} therefore rules out
fixed points in the cone $\mathcal{K}(C)$.

Using Proposition~\ref{prop:FFP}, we find two separable recursive utilities for the
constant-volatility specification of Bansal and Yaron~\cite{BansalYaron2004}. Formally, let $\sigma _{t}:=\bar \sigma >0$, $0<\rho <1$, and $%
\varepsilon _{t}^{x},Y_{t}\sim N(0,1)$. The two separable recursive
utilities are 
\begin{equation*}
U_{t}^{(i)}=C_{t}\exp \!\left( A_{i}+\frac{q_{i}X_{t}^{2}+b_{i}X_{t}}{s}%
\right) ,\qquad i=0,1,
\end{equation*}%
where $q_{0}=0$ and $q_{1}>0$. For the Bansal--Yaron calibration, 
\begin{equation}
\beta =0.998,\quad \rho =0.979,\quad a=0.0015,\quad \bar{\sigma}%
=0.0078,\quad \chi=0.044,\quad s=-9
\label{eq:unit-eis-BY-common-calibration}
\end{equation}%
the separable recursive utilities turn out to be 
\begin{equation*}
\begin{aligned} U_t^{(0)}&\simeq C_t\exp(0.1121+43.47X_t),\\
U_t^{(1)}&\simeq C_t\exp\!\left(
-1.219-46.53X_t-(2.051\times10^4)X_t^2\right). \end{aligned}
\end{equation*}%
Relative to consumption, the higher utility is log-linear in the state
variable, while the lower one is log-quadratic. The existence of a
log-linear solution at unit EIS is known (see, e.g., \cite{PohlSchmeddersWilms2024}), but our results also identify a second, log-quadratic
solution.

Observe that long-run risk models typically assume that the EIS is greater
than one (i.e., $r>0$). In that case, using Proposition \ref%
{prop:nonexistweak}, the unboundedness of $h$ also rules
out a recursive utility in $\mathcal{K}(C)$. We do not expatiate further on this point.

\section{Concluding remarks}

\label{sec:conc} Uniqueness of recursive utility is a question of central
importance, since without it a given set of preference parameters (such as
risk aversion, the discount factor, and the elasticity of intertemporal
substitution) need not pin down a unique utility, and may therefore imply
indeterminate welfare or asset-pricing predictions.

The multiplicity that can arise here resembles a familiar feature of
forward-looking equations in economics. For example, in standard monetary
models, a weak policy response to inflation can allow multiple
self-fulfilling paths for inflation. A similar issue arises here because
current utility depends on continuation utility. When the state variable is
persistent, very good or very bad states can last far into the future.
Different levels of continuation utility in those distant states can then
feed back into utility today, allowing different utility processes to
satisfy the same forward-looking equation.

In Markov environments, we have identified a key distinction between unit
and non-unit EIS: multiplicity can arise with unit EIS, whereas with
non-unit EIS we have established uniqueness within the class of separable
recursive utilities. In the long-run-risk model of Bansal and Yaron~\cite%
{BansalYaron2004}, the state variable ranges over the entire real line, so
the bounds required by our unit-EIS uniqueness condition fail. For the
constant-volatility specification with unit EIS, we establish multiplicity.
Whether such a multiplicity may also arise in this model with
non-unit EIS remains an open question.

\appendix

\section{Appendix\label{append}}

\subsection{Proofs}

\subsubsection{Proof of Lemma~\protect\ref{lemRT}}

We have seen that $\mathcal{K}\left( C\right) \subset L_{++}$. Since%
\begin{equation*}
\varepsilon \mathcal{M}\left( C\right) \leq \mathcal{M}\left( \varepsilon
C\right) \leq \mathcal{M}\left( C\right)
\end{equation*}%
for $0<\varepsilon <1$, it follows that $\mathcal{M}\left( \varepsilon
C\right) \in L_{+}$ and so $\mathbf{T}\left( \varepsilon C\right) =W\left( C,%
\mathcal{M}\left( \varepsilon C\right) \right) $ is well defined.

As $\mathcal{M}$ is subhomogeneous, we can write 
\begin{equation*}
\frac{\mathbf{T}\left( \varepsilon C\right) }{\varepsilon C}=\frac{W\left( C,%
\mathcal{M}\left( \varepsilon C\right) \right) }{\varepsilon C}\geq \frac{%
W\left( C,\varepsilon \mathcal{M}\left( C\right) \right) }{\varepsilon C}
\end{equation*}

Moreover, $W$ is positively homogeneous, so 
\begin{equation*}
\frac{\mathbf{T}\left( \varepsilon C\right) }{\varepsilon C}\geq W\left( 
\frac{1}{\varepsilon },\frac{\mathcal{M}\left( C\right) }{C}\right)
\end{equation*}%
Therefore, if we prove that 
\begin{equation}
W\left( \frac{1}{\varepsilon },\frac{\mathcal{M}\left( C\right) }{C}\right)
\gg 1  \label{eq:DDL}
\end{equation}%
we have the desired result. On the other hand, 
\begin{equation*}
W\left( \frac{1}{\varepsilon },\frac{\mathcal{M}\left( C\right) }{C}\right) =%
\left[ (1-\beta )\varepsilon ^{\left\vert r\right\vert }+\beta \left( \frac{%
\mathcal{M}\left( C\right) }{C}\right) ^{r}\right] ^{1/r}
\end{equation*}%
Since $1/r<0,$ (\ref{eq:DDL}) is equivalent to%
\begin{equation*}
(1-\beta )\varepsilon ^{\left\vert r\right\vert }+\beta \left( \frac{%
\mathcal{M}\left( C\right) }{C}\right) ^{r}\ll 1.
\end{equation*}%
After some algebra, we easily obtain the equivalent condition 
\begin{equation}
\frac{\mathcal{M}(C)}{C}\gg \left( \frac{1-(1-\beta )\varepsilon
^{\left\vert r\right\vert }}{\beta }\right) ^{1/r}  \label{eq:GGE}
\end{equation}%
Observe that 
\begin{equation}
\left( \frac{1-(1-\beta )\varepsilon ^{\left\vert r\right\vert }}{\beta }%
\right) ^{1/r}\downarrow \beta ^{\frac{1}{\left\vert r\right\vert }}\text{ \
as \ }\varepsilon \downarrow 0  \label{eq:BVZ}
\end{equation}%
From (b) it follows 
\begin{equation*}
\frac{\mathcal{M}\left( C\right) }{C}\geq \beta ^{\frac{1}{\left\vert
r\right\vert }}+\eta
\end{equation*}%
for some $\eta >0$, and from (\ref{eq:BVZ}) we have 
\begin{equation*}
\left( \frac{1-(1-\beta )\varepsilon ^{\left\vert r\right\vert }}{\beta }%
\right) ^{1/r}\leq \beta ^{\frac{1}{\left\vert r\right\vert }}+\eta /2
\end{equation*}%
for $\varepsilon >0$ small enough. At last we get%
\begin{equation*}
\frac{\mathcal{M}\left( C\right) }{C}\geq \beta ^{\frac{1}{\left\vert
r\right\vert }}+\eta /2+\eta /2\geq \left( \frac{1-(1-\beta )\varepsilon
^{\left\vert r\right\vert }}{\beta }\right) ^{1/r}+\eta /2
\end{equation*}%
which is (\ref{eq:GGE}) and so the claim is proved.\hfill $\blacksquare $

\subsubsection{Proof of Proposition~\protect\ref{prop:exist-neg}}

Let $X\in \mathcal{K}\left( C\right) $. Then $\mathcal{M}(X)\in L_{+}$ (see (%
\ref{eq:HHH})). In view of (\ref{eq:LLL}) we have $\mathbf{T}\left( X\right)
\leq (1-\beta) ^{1/r}C$. On the other hand, $X\geq \varepsilon C$ holds for $%
\varepsilon $ small enough. By Lemma \ref{lemRT} we infer $\mathbf{T}X\geq 
\mathbf{T}\left( \varepsilon C\right) \geq \varepsilon C$ and so we have
proved that $\mathbf{T}\left( X\right) \in \mathcal{K}\left( C\right) $ as
well as the relation (\ref{eq:YUP})

Fix now the scalar $\varepsilon >0$ for which $\mathbf{T}$ maps $\left[
\varepsilon C,(1-\beta )^{1/r}C\right] $ into itself, and set $%
X^{0}=(1-\beta )^{1/r}C$. Then $X^{1}=\mathbf{T}X^{0}\leq X^{0}$. As $%
\mathbf{T}$ is monotone, $\mathbf{T}X^{1}=\mathbf{T}^{2}X^{0}\leq \mathbf{T}%
X^{0}\leq X^{0}$. Recursively, our processes form a decreasing sequence 
\begin{equation*}
X^{0}\geq \cdot \cdot \cdot \geq \mathbf{T}^{n-1}X^{0}\geq \mathbf{T}%
^{n}X^{0}\geq \cdot \cdot \cdot \geq \varepsilon C,
\end{equation*}%
which is bounded from below by the element $\varepsilon C\in
L_{++}^{1}\left( \delta \right) .$ By the Dedekind $\sigma $-completeness we
have the point-wise convergence 
\begin{equation*}
\mathbf{T}^{n}X^{0}\downarrow U^{\ast }=\inf_{n}\mathbf{T}^{n}X^{0}.
\end{equation*}%
As $\mathbf{T}$ is $\sigma $-order continuous, $\mathbf{T}U^{\ast }=U^{\ast
} $ and $U^{\ast }$ is a fixed point of $\mathbf{T}$ belonging to $%
L_{++}^{1}\left( \delta \right) $.

Let $U\in \mathcal{K}\left( C\right) $ be any fixed point of $\mathbf{T.}$
Since $U=\mathbf{T}U$ we have $U\leq (1-\beta) ^{1/r}C$ and, by the
monotonicity, $U=\mathbf{T}^{n}U\leq \mathbf{T}^{n}\left( (1-\beta)
^{1/r}C\right) $. Therefore, $U\leq U^{\ast }$ as $n\rightarrow \infty $ and
this proves that $U^{\ast }$ is the greatest fixed point.

Observe that by using the increasing sequence 
\begin{equation*}
X^{0}\leq \cdot \cdot \cdot \leq \mathbf{T}^{n-1}X^{0}\leq \mathbf{T}%
^{n}X^{0}\leq \cdot \cdot \leq (1-\beta) ^{1/r}C,
\end{equation*}%
starting at $X^{0}=\varepsilon C$, we deduce, by a similar argument, the
existence of a fixed point $U_{\ast }=\sup_{n}\mathbf{T}^{n}X^{0}$. The
point $U_{\ast }$ is the least fixed point in the interval $\left[
\varepsilon C,(1-\beta) ^{1/r}C\right] .$\hfill $\blacksquare $

\subsubsection{Proof of Proposition~\protect\ref{prop:unique1}}

Our aim is to invoke Theorem 1 of \cite{marmont1}. Fix a scalar $\varepsilon 
$ for which (\ref{eq:HAZ}) holds. Let us prove that no fixed point of $%
\mathbf{T}$ lies on the lower perimeter $\partial ^{\diamond }\left[
\varepsilon C,(1-\beta) ^{1/r}C\right] .$

Suppose by contradiction that there is a fixed point $\overline{U}=\mathbf{T}%
\overline{U}\in \partial _{\diamond }\left[ \varepsilon C,(1-\beta )^{1/r}C%
\right] $. Hence $\overline{U}\geq U_{\ast },$ where $U_{\ast }$ is the
least fixed point of $\mathbf{T}$ in the interval $\left[ \varepsilon
C,(1-\beta )^{1/r}C\right] $ (see the previous proof). Consequently all the
elements of the increasing sequence $\mathbf{T}^{n}\left( \varepsilon
C\right) \uparrow U_{\ast }$ are contained in the lower perimeter. On the
other hand, in view of (\ref{eq:HAZ}) and (\ref{eq:upper}), $\mathbf{T}%
\left( \varepsilon C\right) \notin \partial ^{\diamond }\left[ \varepsilon
C,(1-\beta )^{1/r}C\right] .$ This is a contradiction, so Theorem 1 of \cite%
{marmont1} implies the existence of a unique fixed point in $\left[
\varepsilon C,(1-\beta )^{1/r}C\right] $. Since $\varepsilon $ can be taken
arbitrarily small, we deduce the uniqueness of the fixed point in the cone $%
\mathcal{K}\left( C\right) $.\hfill $\blacksquare $

\subsubsection{Proof of Proposition~\protect\ref{prop:unique2}}

Let us only show that the relation $\mathbf{T}\left( (1-\beta
)^{1/r}C\right) /C\gg (1-\beta )^{1/r}$ is equivalent to $\mathcal{M}\left(
C\right) /C\gg 0.$ By subhomogeneity, 
\begin{equation*}
\mathcal{M}((1-\beta) ^{1/r}C)\geq (1-\beta) ^{1/r}\mathcal{M}(C).
\end{equation*}%
Since $\mathcal{M}(C)/C\gg 0$, there exists $\eta >0$ such that $\mathcal{M}%
(C)\geq \eta C$. Hence 
\begin{equation*}
\mathcal{M}((1-\beta) ^{1/r}C)\geq (1-\beta) ^{1/r}\eta C.
\end{equation*}%
Therefore, 
\begin{align*}
\frac{\mathbf{T}\bigl((1-\beta)^{1/r}C\bigr)}{C} &= \left[ (1-\beta) +\beta
\left( \frac{\mathcal{M}\bigl((1-\beta)^{1/r}C\bigr)}{C} \right)^r \right]%
^{1/r} \\
&\geq \left[ (1-\beta) +\beta(1-\beta)\eta^r \right]^{1/r} \\
&> (1-\beta)^{1/r},
\end{align*}
as desired. The converse is straightforward. The rest of the proof follows the same lines as the proof of
Proposition \ref{prop:unique1}. \hfill $\blacksquare $

\subsubsection{Proof of Proposition~\protect\ref{prop:cobbdouglas}}

Recall that $\mathbf{T}$ is $\beta $-subhomogeneous. First, we point out the
following equivalence 
\begin{equation}
\mathcal{M}(C)\in \mathcal{K}\left( C^{p}\right) \Longleftrightarrow \mathbf{%
T}\left( C\right) \in \mathcal{K}\left( C\right)  \label{eq:RTM}
\end{equation}

Actually, $\mathcal{M}(C)\in \mathcal{K}\left( C^{p}\right) $ means 
\begin{equation*}
\lambda C^{p}\leq \mathcal{M}(C)\leq \mu C^{p}\text{ \ with }\lambda ,\mu >0
\end{equation*}%
Straightforward algebra gives 
\begin{equation*}
\lambda ^{\beta }C^{\alpha +p\beta }\leq \mathbf{T}(C)\leq \mu ^{\beta
}C^{\alpha +p\beta }\text{ }
\end{equation*}%
As $\alpha +p\beta =1$ we deduce that $\mathbf{T}\left( C\right) \in 
\mathcal{K}\left( C\right) $. The converse implication is proved similarly.

Suppose now that $\mathcal{M}(C)\in \mathcal{K}\left( C^{p}\right) $. Hence $%
\mathbf{T}\left( C\right) \in \mathcal{K}\left( C\right) $ and so 
\begin{equation*}
mC\leq \mathbf{T}(C)\leq MC
\end{equation*}%
for some $0<m\leq M$. Pick $0<\varepsilon <1$ and $\lambda >1$ such that $%
\varepsilon ^{1-\beta }\leq m$ and $\lambda ^{1-\beta }\geq M.$

Consequently, 
\begin{eqnarray*}
\mathbf{T}(\varepsilon C) &\geq &\varepsilon ^{\beta }\mathbf{T}(C)\geq
\varepsilon ^{\beta }mC\geq \varepsilon C \\
\mathbf{T}(\lambda C) &\leq &\lambda ^{\beta }\mathbf{T}(C)\leq \lambda
^{\beta }MC\leq \lambda C
\end{eqnarray*}%
Hence $\mathbf{T}$ maps the interval $\left[ \varepsilon C,\lambda C\right] $
into itself. Since $\varepsilon $ and $\lambda $ can be taken arbitrarily
small and large, respectively, $\mathbf{T}$ maps $\mathcal{K}\left( C\right) 
$ into itself. By the usual monotone iteration we deduce the existence of a
fixed point in the interval $\left[ \varepsilon C,\lambda C\right] $.

Let $U$ and $V$ be two fixed points in $\mathcal{K}\left( C\right) $. We
have thus 
\begin{equation*}
q^{-1}U\leq V\leq qU
\end{equation*}%
for some $q>1.$ Monotonicity and $\beta $-subhomogeneity give 
\begin{equation*}
q^{-\beta }\mathbf{T}U\leq \mathbf{T}\left( q^{-1}U\right) \leq \mathbf{T}%
V\leq \mathbf{T}\left( qU\right) \leq q^{\beta }\mathbf{T}U
\end{equation*}%
Using $U=\mathbf{T}U$ and $V=\mathbf{T}V$ yields%
\begin{equation*}
q^{-\beta }U\leq V\leq q^{\beta }U
\end{equation*}%
Iterating 
\begin{equation*}
q^{-\beta ^{n}}U\leq V\leq q^{\beta ^{n}}U\text{ \ for every }n
\end{equation*}%
Letting $n\rightarrow \infty $ gives $U=V$ and so the fixed point in $%
\mathcal{K}\left( C\right) $ is unique.

Suppose now that $\mathcal{M}(C)\notin \mathcal{K}\left( C^{p}\right) $ and
that, by contradiction, there is a fixed point $U\in \mathcal{K}\left(
C\right) $. Hence $q^{-1}C\leq U\leq qC$ for some $q>1.$ Monotonicity and
subhomogeneity give

\begin{equation*}
q^{-1}\mathbf{T}C\leq U\leq q\mathbf{T}C
\end{equation*}%
Namely $U\in \mathcal{K}\left( \mathbf{T}C\right) \Longleftrightarrow 
\mathbf{T}C\in \mathcal{K}\left( U\right) =\mathcal{K}\left( C\right) $
which contradicts the assumption $\mathcal{M}(C)\notin \mathcal{K}\left(
C^{p}\right) $. Consequently, no fixed point lies in $\mathcal{K}\left(
C\right) $.

\subsubsection{Proof of Proposition \protect\ref{prop:FFP}}

The analogue of (\ref{eq:RTM}) is the equivalence: 
\begin{equation*}
\mathcal{M}\left( Z\right) \in \mathcal{K}\left( C^{-\alpha /\beta
}Z^{1/\beta }\right) \Longleftrightarrow \mathbf{T}Z\in \mathcal{K}\left(
Z\right)
\end{equation*}%
which extends (\ref{eq:RTM}). Indeed, if $\mathcal{M}\left( Z\right) \in 
\mathcal{K}\left( C^{-\alpha /\beta }Z^{1/\beta }\right) $ then $\mathcal{M}%
_{t}\left( Z_{t+1}\right) =C_{t}^{-\alpha /\beta }Z_{t}^{1/\beta }B_{t}$,
where $B\in \mathcal{K}\left( 1\right) .$ Therefore%
\begin{equation*}
\left( \mathbf{T}Z\right) _{t}=C_{t}^{\alpha }\left[ \mathcal{M}_{t}\left(
Z_{t+1}\right) \right] ^{\beta }=B_{t}^{\beta }Z_{t}
\end{equation*}%
Namely, $\mathbf{T}Z\in \mathcal{K}\left( Z\right) $. The converse is
similar. At this point the proof goes on exactly as the proof of Proposition~%
\ref{prop:cobbdouglas}, by replacing $C$ with $Z$.

Finally, in the Epstein--Zin case, $\mathcal{M}_{t}(Z_{t+1})=(\mathbb{E}%
_{t}Z_{t+1}^{s})^{1/s}$. Hence condition~\eqref{eq:FFP} is equivalent to 
\begin{equation*}
(\mathbb{E}_{t}Z_{t+1}^{s})^{1/s}=C_{t}^{-\alpha /\beta }Z_{t}^{1/\beta }%
\widetilde{B}_{t}
\end{equation*}%
for some $\widetilde{B}\in \mathcal{K}(\mathbf{1})$. Raising to the power $s$ and
setting $B_{t}:=\widetilde{B}_{t}^{\,s}\in \mathcal{K}(\mathbf{1})$ gives %
\eqref{eq:DDD}.

If $B_{t}\equiv \overline{B}>0,$ then $\left( \mathbf{T}Z\right)
_{t}=C_{t}^{\alpha }\left[ \mathcal{M}_{t}\left( Z_{t+1}\right) \right]
^{\beta }=C_{t}^{\alpha }C_{t}^{-\alpha }Z_{t}\overline{B}^{\beta /s}=%
\overline{B}^{\beta /s}Z_{t}$. Hence, $\mathbf{T}Z=\overline{B}^{\beta /s}Z.$
Setting now $U=\overline{B}^{\frac{\beta}{\alpha s}}Z$, we get 
\begin{equation*}
\mathbf{T}U =\mathbf{T}\left(\overline{B}^{\frac{\beta}{\alpha s}}Z\right) =%
\overline{B}^{\frac{\beta^{2}}{\alpha s}}\mathbf{T}(Z) =\overline{B}^{\frac{%
\beta^{2}}{\alpha s}}\overline{B}^{\beta/s}Z =\overline{B}^{\frac{\beta^{2}}{%
\alpha s}} \overline{B}^{\beta/s}\overline{B}^{-\frac{\beta}{\alpha s}}U =U
\end{equation*}
as desired.\hfill $\blacksquare $

\subsubsection{Proof of Proposition~\protect\ref{prop:convergence}}

Suppose $r<0$; the other case is similar. Let $\overline{X}$ be any 
process in $\mathcal{K}\left( C\right) $ and $\overline{U}$ be the unique
fixed point in $\mathcal{K}\left( C\right) .$ Consequently there exist a
sufficiently small $\varepsilon \in (0,1)$ and a sufficiently large $\lambda
>1$ such that $\overline{X},\overline{U}\in \left[ \varepsilon C,\lambda C%
\right] $ and with $\mathbf{T}(\varepsilon C)\geq \varepsilon C$, $\mathbf{T}%
(\lambda C)\leq \lambda C$. Iterating the relation 
\begin{equation*}
\varepsilon C\leq \overline{X}\leq \lambda C
\end{equation*}%
we get $\mathbf{T}^{n}\left( \varepsilon C\right) \leq \mathbf{T}^{n}\left( 
\overline{X}\right) \leq \mathbf{T}^{n}\left( \lambda C\right) $ and so 
\begin{equation*}
\mathbf{T}^{n}\left( \varepsilon C\right) -\overline{U}\leq \mathbf{T}^{n}%
\overline{X}-\overline{U}\leq \mathbf{T}^{n}\left( \lambda C\right) -%
\overline{U}
\end{equation*}%
Since $\mathbf{T}^{n}\left( \varepsilon C\right) \uparrow \overline{U}$ and
$\mathbf{T}^{n}\left( \lambda C\right) \downarrow \overline{U}$ it follows
that $\mathbf{T}^{n}\left( \varepsilon C\right) -\overline{U}\leq 0$ and
$\mathbf{T}^{n}\left( \lambda C\right) -\overline{U}\geq 0$. This
implies 
\begin{equation*}
\left\vert \mathbf{T}^{n}\left( \overline{X}\right) -\overline{U}\right\vert
\leq Z_{n}
\end{equation*}%
where $Z_{n}=\left( \mathbf{T}^{n}\left( \lambda C\right) -\overline{U}%
\right) \vee \left( \overline{U}-\mathbf{T}^{n}\left( \varepsilon C\right)
\right) $. Clearly $Z_{n}\geq 0$ and $Z_{n+1}\downarrow 0$ and this proves
the first claim. On the other hand 
\begin{equation*}
\left\Vert \overline{U}-\mathbf{T}^{n}\left( \overline{X}\right) \right\Vert
=\left\Vert \text{ }\left\vert \overline{U}-\mathbf{T}^{n}\left( \overline{X}%
\right) \right\vert \text{ }\right\Vert \leq \left\Vert Z_{n}\right\Vert
\downarrow 0
\end{equation*}%
by the monotone convergence theorem.\hfill $\blacksquare $

\subsubsection{Proof of Proposition~\protect\ref{prop:nonexist-positive}}

Let $r\in (0,1)$ and $U=$ $\mathbf{T(}U)$ be a fixed point. In view of (\ref%
{eq:ZAZ}) $U\geq \left( 1-\beta \right) ^{1/r}C.$ Let $\overline{\lambda}$ be the greatest scalar such that
$U \geq \overline{\lambda} C$ (see Section \ref{sect. intro}). Therefore,%
\begin{align*}
U& =\mathbf{T}\left( U\right) \geq \mathbf{T}\left( \overline{\lambda }%
C\right) =W\left( C,\mathcal{M}\left( \overline{\lambda }C\right) \right)
=W\left( C,\overline{\lambda }\mathcal{M}\left( C\right) \right) \\
& =CW\left( 1,\overline{\lambda }\frac{\mathcal{M}\left( C\right) }{C}%
\right) =C\left[ (1-\beta )+\beta \overline{\lambda }^{r}\left[ \frac{%
\mathcal{M}(C)}{C}\right] ^{r}\right] ^{1/r}
\end{align*}%
Inequality (\ref{eq:non2}) implies%
\begin{equation*}
U\geq C\left[ (1-\beta )+\overline{\lambda }^{r}\right] ^{1/r}
\end{equation*}%
Thanks to the property of the scalar $\overline{\lambda }$, we would have
\begin{equation*}
\overline{\lambda }\geq \left[ (1-\beta )+\overline{\lambda }^{r}\right]
^{1/r}\Longleftrightarrow \overline{\lambda }^{r}\geq (1-\beta )+\overline{%
\lambda }^{r}
\end{equation*}%
which is absurd. Hence, there is no fixed point$.$

The case $r<0$ is similar by taking into account that $U\leq \left( 1-\beta
\right) ^{1/r}C$ holds for any fixed point of $\mathbf{T}$.\hfill $%
\blacksquare $

\subsubsection{Proof of Proposition \ref{prop:nonexistweak}}

Let $r\in (0,1)$. Suppose, by contradiction, that a fixed point $U$ belongs
to $\mathcal{K}(C)$. There is then a constant $q>1$ such that $q^{-1}C\leq
U\leq qC$. Monotonicity and subhomogeneity of the certainty equivalent imply 
\begin{equation*}
\mathcal{M}_{t}(U_{t+1})\geq \mathcal{M}_{t}(q^{-1}C_{t+1})\geq q^{-1}%
\mathcal{M}_{t}(C_{t+1}).
\end{equation*}%
Dividing the recursive-utility equation by $C_{t}$ and using $r>0$ therefore
gives 
\begin{eqnarray*}
q &\geq &\frac{U_{t}}{C_{t}}=\left[ 1-\beta +\beta \left( \frac{\mathcal{M}%
_{t}(U_{t+1})}{C_{t}}\right) ^{r}\right] ^{1/r}\geq \\
&&\beta ^{1/r}\frac{\mathcal{M}_{t}(U_{t+1})}{C_{t}}\geq \beta ^{1/r}q^{-1}%
\frac{\mathcal{M}_{t}(C_{t+1})}{C_{t}}
\end{eqnarray*}%
It follows that $\sup \frac{\mathcal{M}(C)}{C}<\infty $. A contradiction.

Analogously, suppose $r<0$ and $q^{-1}C\leq U\leq qC$ where $U$ is a fixed
point of $\mathbf{T}$. Then we can write%
\begin{equation*}
\mathcal{M}_{t}(U_{t+1})\leq \mathcal{M}_{t}(qC_{t+1})\leq q\mathcal{M}%
_{t}(C_{t+1}).
\end{equation*}%
In this case%
\begin{eqnarray*}
q^{-1} &\leq &\frac{U_{t}}{C_{t}}=\left[ 1-\beta +\beta \left( \frac{%
\mathcal{M}_{t}(U_{t+1})}{C_{t}}\right) ^{r}\right] ^{1/r} \\
&\leq &\beta ^{1/r}\frac{\mathcal{M}_{t}(U_{t+1})}{C_{t}}\leq \beta ^{1/r}q%
\frac{\mathcal{M}_{t}(C_{t+1})}{C_{t}}
\end{eqnarray*}%
and so $\frac{\mathcal{M}(C)}{C}\gg 0$ which contradicts our
hypothesis.\hfill $\blacksquare $

\subsubsection{Proof of Proposition~\protect\ref{prop:operatorA}}

Suppose first that $r\in(0,1)$ and let
$(1-\beta)^{1/r}\leq\varphi\leq b$. By monotonicity, in view of
(\ref{eq:GSL}),
\begin{equation*}
(1-\beta)^{1/r}\leq\mathbf{A}\varphi\leq\mathbf{A}b
=\Gamma_r\bigl(bh(\xi)\bigr)
\leq\Gamma_r\bigl(b\overline h\bigr)
=\left((1-\beta)+\beta b^r\overline h^{\,r}\right)^{1/r}.
\end{equation*}
Hence $\mathbf{A}$ maps $[(1-\beta)^{1/r},b]$ into itself whenever
\[
b\geq
\left(\frac{1-\beta}{1-\beta\overline h^{\,r}}\right)^{1/r},
\]
which is well defined under~\eqref{eq;MMM}.

For $r<0$, an analogous argument, using $\underline h=\inf h$, shows that
$\mathbf{A}$ maps $[b,(1-\beta)^{1/r}]$ into itself whenever
\[
0<b\leq
\left(\frac{1-\beta}{1-\beta\underline h^{\,r}}\right)^{1/r}.
\]

In either case, monotonicity and $\sigma$-order continuity imply that the
iterates from the endpoints of the invariant interval converge monotonically
to its least and greatest fixed points. When $\mathcal{M}$ is order concave,
$\mathbf{A}$ is order concave, and uniqueness follows from the same
lower-perimeter arguments as in the proofs of
Propositions~\ref{prop:unique1} and~\ref{prop:unique2}. Finally, for
$r\in(0,1)$ every fixed point belongs to $\mathcal{K}(\mathbf{1})$.
\hfill $\blacksquare$

\clearpage

\setcounter{subsection}{0} \setcounter{section}{0} \renewcommand{%
\thesection}{S} \renewcommand{\thesubsection}{S.\arabic{subsection}} %
\renewcommand{\thesubsubsection}{\thesubsection.\arabic{subsubsection}}

\section{Supplemental Appendix}

\label{suppapend}

\subsection{Concavity of the certainty equivalents\label{sec:concave}}

We have frequently claimed in the present paper the concavity of certainty
equivalents by resorting to a well-established property, known in the finite-dimensional setting: a positively homogeneous and quasi-concave function is
concave on the positive cone.

We intend to point out an extension of this basic result to the class of
quasi-arithmetic operators 
\begin{equation*}
\Phi (X)=\left( \mathbb{E}\left[ \varphi (X)^{s}\mid \mathcal{F}\right]
\right) ^{1/s}
\end{equation*}

Clearly the functional $\Phi $ takes values in an ordered vector space and
therefore the concepts of concavity as well as quasi-concavity must be
adapted. The translation to a Markovian setting is straightforward.

\bigskip

Let $\varphi :{B}_{s}\rightarrow L_{++}^{0}(\Omega ,\Sigma ,\mathbb{P})$ be
a (conditionally) positively homogeneous and concave function defined on a (conditionally) convex set ${B}%
_{s}$ of $L_{+}^{0}(\Omega ,\Sigma ,\mathbb{P})$ such that $\mathbb{E}%
[\varphi (X)^{s}\mid \mathcal{F}]$ is finite for every $X\in {B}_{s}$ and
where $\mathcal{F}$ is a sub-$\sigma $-algebra of $\Sigma $.

\begin{proposition}
\label{prop:conditional-power-mean-concavity} The map 
\begin{equation*}
\Phi (X)=\left( \mathbb{E}\left[ \varphi (X)^{s}\mid \mathcal{F}\right]
\right) ^{1/s}
\end{equation*}%
from $B_{s}$ to $L_{++}^{0}(\Omega ,\mathcal{F},\mathbb{P})$ is positively
homogeneous and concave for every $s\in (-\infty ,0)\cup (0,1)$.
\end{proposition}

A similar result holds for $s>1$, provided $\varphi $ is positively
homogeneous and convex. In this case $\Phi (X)$ is convex.

Notice that, by virtue of the property of the expectation, a stronger
property of homogeneity will be useful. Namely, 
\begin{equation}
\Phi (YX)=Y\Phi (X)\text{ for }X\in B_{s}\text{ and }Y\in L_{++}^{0}(\Omega ,%
\mathcal{F},\mathbb{P}).  \label{eq:JHZ}
\end{equation}
Earlier, we called this property conditional positive homogeneity in
subsection \ref{sec:markovstationary}.

\bigskip

\textbf{Proof. \ }The map $\Phi $ is well defined for every $X\in B_{s}$,
and $\Phi \left( X\right) $ is a finite and strictly positive $\mathcal{F}$%
-measurable random variable for every $X\in B_{s}$.

Let us first assume $s\in\left( 0,1\right) .$

It is trivial to check that the map $X\longmapsto \varphi \left( X\right)
^{s}$ from $B_{s}$ into $L_{++}^{0}$ turns out to be concave for $s\in
\left( 0,1\right) .$ Consider first the map 
\begin{equation*}
\rho \left( X\right) =\mathbb{E}\left[ \varphi \left( X\right)
^{s}\left\vert \mathcal{F}\right. \right]
\end{equation*}%
defined on $B_{s}$. From the properties of the conditional expectation (see
for instance \cite[Sect. IV.3]{neveu}) it follows that the concavity is
preserved. Hence,%
\begin{equation*}
\rho \left( \beta X_{1}+\overline{\beta }X_{2}\right) \geq \beta \rho \left(
X_{1}\right) +\overline{\beta }\rho \left( X_{2}\right)
\end{equation*}%
for all $X_{1},X_{2}\in B_{s}$ and $\overline{\beta }=1-\beta $.

We show that $\rho $ is quasi-concave in the following sense. Given any $%
Y\in L_{++}^{0}\left( \Omega ,\mathcal{F},\mathbb{P}\right) $, define the
upper-level set 
\begin{equation*}
U_{\rho }\left( Y\right) =\left\{ X\in B_{s}:\rho \left( X\right) \geq
Y\right\}
\end{equation*}%
It is easy to check that $U_{\rho }\left( Y\right) $ is  convex whenever it
is not empty.

Since $\Phi \left( X\right) =\left[ \rho \left( X\right) \right] ^{1/s}$, we
have $U_{\Phi }\left( Y\right) =U_{\rho }\left( Y^{s}\right) $ because $%
X\longmapsto X^{s}$ is strictly increasing. Hence $\Phi \left( X\right)
=\left( \mathbb{E}\left[ \varphi \left( X\right) ^{s}\left\vert \mathcal{F}%
\right. \right] \right) ^{1/s}$ is quasi-concave as well.

Now let $Y_{i}=\Phi (X_{i})$ for $X_{i}\in B_{s}$ and $i=1,2$. In view of (%
\ref{eq:JHZ}), we have $\Phi (X_{i}/Y_{i})=\Phi (X_{i})/Y_{i}=1.$
Consequently, $X_{i}/Y_{i}\in U_{\Phi }\left( 1\right) $. But $U_{\Phi
}\left( 1\right) $ is convex. Hence 
\begin{equation*}
\Phi \left( \frac{Y_{1}}{Y_{1}+Y_{2}}\cdot \frac{X_{1}}{Y_{1}}+\frac{Y_{2}}{%
Y_{1}+Y_{2}}\cdot \frac{X_{2}}{Y_{2}}\right) =\Phi \left( \frac{X_{1}+X_{2}}{%
Y_{1}+Y_{2}}\right) \geq 1
\end{equation*}%
and hence $\Phi (X_{1}+X_{2})\geq \Phi (X_{1})+\Phi (X_{2})$. This is
superadditivity; combined with positive homogeneity, it implies concavity.

When $s<0$, the map $X\mapsto \varphi (X)^{s}$ is convex, so $\rho $ is convex. Because $t\mapsto t^{1/s}$ is decreasing, $\Phi $ is again
quasi-concave and the proof proceeds in the same way. \hfill $\blacksquare $

\subsection{Calculations for the hidden-regime model}

\label{app:two-state-hidden-regime}

Here, we formally recast the two-state hidden-regime model from Ju and Miao~%
\cite{JuMiao2012} into our Markov economy notation and then verify the
calculations reported in Section~\ref{subsec:two-state-hidden-regime} (see
also Section 3 in \cite{JuMiao2012}).

\paragraph{Primitive hidden-regime economy.}

Let $S_{t}\in \{1,2\}$ be the unobserved regime, with transition
probabilities 
\begin{equation*}
\mathbb{P}(S_{t+1}=1\mid S_{t}=1)=\lambda _{11},\qquad \mathbb{P}%
(S_{t+1}=1\mid S_{t}=2)=\lambda _{21},\qquad 0<\lambda _{21}<\lambda _{11}<1.
\end{equation*}%
At date $t$, the agent has observed the growth history $\mathcal{Y}%
_{t}=\sigma (g_{1},\ldots ,g_{t})$, but not the underlying regime. The
agent's belief at time $t$ that the economy will be in the high-growth
regime in the next period is 
\begin{equation*}
\mu _{t}:=\mathbb{P}(S_{t+1}=1\mid \mathcal{Y}_{t}).
\end{equation*}%
Conditional on $S_{t+1}=i$, the density of $g_{t+1}$ is $f_{i}$. The
conditional density of $g_{t+1}$ given the information available at date $t$
is therefore obtained from the law of total probability: 
\begin{equation*}
\begin{aligned} f_{g_{t+1}\mid\mathcal{Y}_t}(g) &= \sum_{i=1}^{2}
\mathbb{P}(S_{t+1}=i\mid\mathcal{Y}_t) f_{g_{t+1}\mid
S_{t+1}=i,\mathcal{Y}_t}(g) \\ &= \mu_t
f_1(g)+(1-\mu_t)f_2(g):=f_{\mu_t}(g). \end{aligned}
\end{equation*}

After observing $g_{t+1}=g$, the agent updates the belief $p(\mu_t,g)$ that
the regime realized at date $t+1$ was the high-growth regime. Since $%
\mathcal{Y}_{t+1}=\sigma(\mathcal{Y}_t,g_{t+1})$, Bayes' rule gives 
\begin{align}
p(\mu_t,g) &:=
\mathbb{P}(S_{t+1}=1\mid\mathcal{Y}_t,g_{t+1}=g)\notag \\ &= \frac{
\mathbb{P}(S_{t+1}=1\mid\mathcal{Y}_t) f_{g_{t+1}\mid
S_{t+1}=1,\mathcal{Y}_t}(g) }{ f_{g_{t+1}\mid\mathcal{Y}_t}(g) }\notag \\ &=
\frac{\mu_t f_1(g)} {\mu_t f_1(g)+(1-\mu_t)f_2(g)} = \frac{\mu_t
f_1(g)}{f_{\mu_t}(g)}. \label{eq:app-posterior-two-state}
\end{align}
For $\mu_{t+1}$ we obtain 
\begin{align}
\mu_{t+1} &= \sum_{i=1}^{2} \mathbb{P}(S_{t+2}=1\mid S_{t+1}=i,\mathcal{Y}_{t+1}) \mathbb{P}(S_{t+1}=i\mid\mathcal{Y}_{t+1}) \notag \\
&= \lambda_{11}p(\mu_t,g_{t+1}) +\lambda_{21}\bigl[1-p(\mu_t,g_{t+1})\bigr] \label{eq:app-belief-update-two-state} \\
&= \lambda_{21} +(\lambda_{11}-\lambda_{21})p(\mu_t,g_{t+1}):=B(\mu_t,g_{t+1}). \notag
\end{align}
Now, because $0<\mu<1$ and the Gaussian densities $f_1(g)$ and $f_2(g)$ are
strictly positive for every finite $g$, Bayes' formula implies $p(\mu,g)\in
(0,1)$, and since $0<\lambda_{21}<\lambda_{11}<1$, it follows that 
\begin{equation*}
\lambda_{21} < B(\mu,g) < \lambda_{11}.
\end{equation*}
Hence, we can conclude that $\mu_t\in (\lambda_{21},\lambda_{11})$ for every 
$t\geq 1$. Therefore, one can take its closure $\mathcal{Z}%
=K:=[\lambda_{21},\lambda_{11}]$ as the state space.

\paragraph{Function $\protect\kappa $.}

Note that equation \eqref{eq:app-belief-update-two-state} is equivalent to
\[
\frac{p(\mu_t,g_{t+1})}{1-p(\mu_t,g_{t+1})}
=
\frac{\mu_{t+1}-\lambda_{21}}
{\lambda_{11}-\mu_{t+1}},
\]
which using \eqref{eq:app-posterior-two-state} yields
\[
\frac{f_{1}(g_{t+1})}{f_{2}(g_{t+1})}
=
\frac{1-\mu_t}{\mu_t}
\frac{\mu_{t+1}-\lambda_{21}}
{\lambda_{11}-\mu_{t+1}}.
\]
Using the fact that $f_{1}$ and $f_{2}$ are normal densities with a common
variance, we obtain 
\begin{equation*}
\log \frac{f_{1}(g_{t+1})}{f_{2}(g_{t+1})}=\frac{\bar{g}_{1}-\bar{g}_{2}}{%
\sigma ^{2}}\left( g_{t+1}-\frac{\bar{g}_{1}+\bar{g}_{2}}{2}\right) .
\end{equation*}%
Therefore, combining these two identities gives 
\begin{equation*}
\log \left( \frac{1-\mu _{t}}{\mu _{t}}\right) +\log \left( \frac{\mu
_{t+1}-\lambda _{21}}{\lambda _{11}-\mu _{t+1}}\right) =\frac{\bar{g}_{1}-%
\bar{g}_{2}}{\sigma ^{2}}\left( g_{t+1}-\frac{\bar{g}_{1}+\bar{g}_{2}}{2}%
\right) ,
\end{equation*}%
which, when solved for $g_{t+1}=\kappa (\mu _{t},\mu _{t+1})$, yields 
\begin{equation*}
\kappa (\mu _{t},\mu _{t+1})=\frac{\bar{g}_{1}+\bar{g}_{2}}{2}+\frac{\sigma
^{2}}{\bar{g}_{1}-\bar{g}_{2}}\left[ \log \frac{\mu _{t+1}-\lambda _{21}}{%
\lambda _{11}-\mu _{t+1}}-\log \frac{\mu _{t}}{1-\mu _{t}}\right] .
\end{equation*}%
Thus, conditional on the current belief, the next belief and realized growth
contain the same information. This is the function used in %
\eqref{eq:kappa-two-state-explicit}.

\paragraph{Transition kernel $Q$.}

We want to find $\mathbb{P}(\mu _{t+1}\in A\mid \mu _{0},\ldots ,\mu _{t}).$
Using the fact that $\mu _{t+1}=B(\mu _{t},g_{t+1})$ and the law of iterated
expectations, we can write 
\begin{align*}
\mathbb{P}(\mu _{t+1}\in A\mid \mu _{0},\ldots ,\mu _{t})& =\mathbb{E}\left[ 
\mathbb{P}(\mu _{t+1}\in A\mid \mathcal{Y}_{t})\mid \mu _{0},\ldots ,\mu _{t}%
\right] \\
& =\mathbb{E}\left[ \int_{\mathbb{R}}\mathbf{1}_{A}\bigl(B(\mu _{t},g)\bigr)%
f_{\mu _{t}}(g)\,dg\mid \mu _{0},\ldots ,\mu _{t}\right] \\
& =\int_{\mathbb{R}}\mathbf{1}_{A}\bigl(B(\mu _{t},g)\bigr)f_{\mu
_{t}}(g)\,dg=\mathbb{P}(\mu _{t+1}\in A\mid \mu _{t}).
\end{align*}%
Therefore, $(\mu _{t})_{t\geq 0}$ is a time-homogeneous Markov process with
transition kernel $Q$ given by 
\begin{equation*}
Q(\mu _{t},A):=\mathbb{P}(\mu _{t+1}\in A\mid \mu _{t})=\int_{\mathbb{R}}%
\mathbf{1}_{A}\bigl(B(\mu _{t},g)\bigr)f_{\mu _{t}}(g)\,dg.
\end{equation*}%
For each fixed $\mu $, the map $g\mapsto B(\mu ,g)$ is a continuously
differentiable and strictly increasing bijection from $\mathbb{R}$ onto $%
(\lambda _{21},\lambda _{11})$. Hence, it is immediate to verify that $Q(\mu
,\cdot )$ is absolutely continuous with respect to the Lebesgue measure and
its Radon--Nikodym derivative is 
\begin{equation*}
\frac{dQ(\mu ,\cdot )}{d\mu ^{\prime }}(\mu ^{\prime })=f_{\mu }\bigl(\kappa
(\mu ,\mu ^{\prime })\bigr)\frac{\partial \kappa (\mu ,\mu ^{\prime })}{%
\partial \mu ^{\prime }}=f_{\mu }\bigl(\kappa (\mu ,\mu ^{\prime })\bigr)%
J(\mu ,\mu ^{\prime }),
\end{equation*}%
for every $\mu ^{\prime }\in (\lambda _{21},\lambda _{11}).$ Hence 
\begin{equation*}
Q(\mu ,d\mu ^{\prime })=f_{\mu }\bigl(\kappa (\mu ,\mu ^{\prime })\bigr)%
J(\mu ,\mu ^{\prime })\,d\mu ^{\prime },
\end{equation*}%
or, equivalently, for every Borel set $A\subseteq K$, 
\begin{equation*}
Q(\mu ,A)=\int_{A}f_{\mu }\bigl(\kappa (\mu ,\mu ^{\prime })\bigr)J(\mu ,\mu
^{\prime })\,d\mu ^{\prime },
\end{equation*}%
as desired.

\paragraph{Integrability of consumption.}

Since $C_{t}=C_{0}\exp \left( \sum_{j=1}^{t}g_{j}\right) ,$ $C_{t}$ is $%
\mathcal{Y}_{t}$-measurable. In particular, we have that 
\begin{equation*}
\mathbb{E}[C_{t+1}\mid \mathcal{Y}_{t}]=C_{t}\,\mathbb{E}[e^{g_{t+1}}\mid 
\mathcal{Y}_{t}]=C_{t}\left\{ \mu _{t}\mathbb{E}[e^{\bar{g}_{1}+\sigma
\varepsilon _{t+1}}]+(1-\mu _{t})\mathbb{E}[e^{\bar{g}_{2}+\sigma
\varepsilon _{t+1}}]\right\} .
\end{equation*}%
As $\varepsilon _{t+1}\sim N(0,1)$, we get 
\begin{equation*}
\mathbb{E}[e^{\bar{g}_{i}+\sigma \varepsilon _{t+1}}]=\exp \left( \bar{g}%
_{i}+\frac{1}{2}\sigma ^{2}\right) ,\qquad i=1,2.
\end{equation*}%
Thus 
\begin{equation}
\mathbb{E}[C_{t+1}\mid \mathcal{Y}_{t}]=C_{t}\exp \left( \frac{1}{2}\sigma
^{2}\right) \left[ \mu _{t}e^{\bar{g}_{1}}+(1-\mu _{t})e^{\bar{g}_{2}}\right]
.  \label{eq:app-conditional-consumption-bound-start}
\end{equation}%
Since $\bar{g}_{1}>\bar{g}_{2}$ and $\mu _{t}\leq \lambda _{11}$, equation %
\eqref{eq:app-conditional-consumption-bound-start} gives 
\begin{equation*}
\mathbb{E}[C_{t+1}\mid \mathcal{Y}_{t}]\leq C_{t}\exp \left( \frac{1}{2}%
\sigma ^{2}\right) \left[ \lambda _{11}e^{\bar{g}_{1}}+(1-\lambda _{11})e^{%
\bar{g}_{2}}\right] .
\end{equation*}%
Taking expectations and using the law of iterated expectations, 
\begin{equation*}
\mathbb{E}[C_{t+1}]=\mathbb{E}\!\left[ \mathbb{E}[C_{t+1}\mid \mathcal{Y}%
_{t}]\right] \leq \exp \left( \frac{1}{2}\sigma ^{2}\right) \left[ \lambda
_{11}e^{\bar{g}_{1}}+(1-\lambda _{11})e^{\bar{g}_{2}}\right] \mathbb{E}%
[C_{t}].
\end{equation*}%
Iterating this inequality yields 
\begin{equation*}
\mathbb{E}[C_{t}]\leq C_{0}\left( \exp \left( \frac{1}{2}\sigma ^{2}\right) %
\left[ \lambda _{11}e^{\bar{g}_{1}}+(1-\lambda _{11})e^{\bar{g}_{2}}\right]
\right) ^{t}.
\end{equation*}%
Consequently, 
\begin{equation*}
\sum_{t\geq 0}\delta ^{t}\mathbb{E}[C_{t}]<\infty \qquad \text{whenever}%
\qquad \delta \exp \left( \frac{1}{2}\sigma ^{2}\right) \left[ \lambda
_{11}e^{\bar{g}_{1}}+(1-\lambda _{11})e^{\bar{g}_{2}}\right] <1,
\end{equation*}%
as desired.

\subsubsection{Smooth ambiguity model}

\label{app:ju-miao-integral-equation}

\paragraph{Growth factor.}

For $i=1,2$, the regime-conditional certainty equivalent of consumption is 
\begin{equation*}
\mathcal{M}_{t}^{i}(C_{t+1})=C_{t}\left[ \int_{K}e^{s\kappa (\mu _{t},\mu
^{\prime })}Q_{i}(\mu _{t},d\mu ^{\prime })\right] ^{1/s}.
\end{equation*}%
By the density representation of $Q_{i}$ and the change of variables $%
g=\kappa (\mu ,\mu ^{\prime })$, 
\begin{equation*}
\int_{K}e^{s\kappa (\mu ,\mu ^{\prime })}Q_{i}(\mu ,d\mu ^{\prime })=\int_{%
\mathbb{R}}e^{sg}f_{i}(g)\,dg=\exp \left( s\bar{g}_{i}+\frac{1}{2}%
s^{2}\sigma ^{2}\right) .
\end{equation*}%
Hence, with 
\begin{equation*}
a_{i}:=\bar{g}_{i}+\frac{1}{2}s\sigma ^{2},
\end{equation*}%
we have $\mathcal{M}_{t}^{i}(C_{t+1})=C_{t}e^{a_{i}}$. \ It follows that the
smooth-ambiguity certainty equivalent satisfies 
\begin{equation*}
\frac{\mathcal{M}_{t}(C_{t+1})}{C_{t}}=h(\mu _{t}),
\end{equation*}%
with 
\begin{equation*}
\begin{aligned} h(\mu) &=\left\{ \mu e^{\eta a_1}+(1-\mu)e^{\eta a_2}
\right\}^{1/\eta}\\ &=\left\{ \mu\exp\left[\eta\left(\bar
g_1+\frac12s\sigma^2\right)\right] +(1-\mu)\exp\left[\eta\left(\bar
g_2+\frac12s\sigma^2\right)\right] \right\}^{1/\eta}. \end{aligned}
\end{equation*}%
Moreover, 
\begin{equation*}
h^{\prime }(\mu )=\frac{e^{\eta a_{1}}-e^{\eta a_{2}}}{\eta }\left[ \mu
e^{\eta a_{1}}+(1-\mu )e^{\eta a_{2}}\right] ^{1/\eta -1}>0,
\end{equation*}%
as $a_{1}>a_{2}$. Thus $h$ is increasing for every $\eta \neq 0$, and its
minimum and maximum on $K$ are attained at $\lambda _{21}$ and $\lambda
_{11} $, respectively.

\subsection{Rare-disaster economy}

\label{app:ez-rare-disasters}

This appendix checks the formulas of Section~\ref{sec:bidder-smith-wachter}. Recall that for every $n\geq 0$,
\[
\int_0^\infty
H'
\frac{(H')^{d_H+n-1}e^{-H'/c}}
{c^{d_H+n}\Gamma(d_H+n)}
\,dH'
=
c(d_H+n),
\]
where for every $a> 0$ \[
\Gamma(a):=\int_0^\infty x^{a-1}e^{-x}\,dx.
\]
\paragraph{Disaster-intensity process.}

If $q(H,H^{\prime })$ denotes the conditional density defined in %
\eqref{eq:bsw-transition-kernel}, the associated transition kernel is 
\begin{equation*}
Q(H,A)=\int_{A}q(H,H^{\prime })\,dH^{\prime },\qquad A\in \mathcal{B}(%
\mathbb{R}_{+}).
\end{equation*}%
Consequently, 
\begin{equation*}
\begin{aligned} \mathbb{E}[H_{t+1}\mid H_t=H] &= \int_{\mathbb{R}_+}
H'\,Q(H,dH') \\ &= \int_0^\infty H' q(H,H')\,dH' \\ &= \int_0^\infty H'
e^{-\phi H/c} \sum_{n=0}^{\infty} \frac{(\phi H/c)^n}{n!} \frac{
(H')^{d_H+n-1}e^{-H'/c} }{ c^{d_H+n}\Gamma(d_H+n) } \,dH' \\ &= e^{-\phi
H/c} \sum_{n=0}^{\infty} \frac{(\phi H/c)^n}{n!} \int_0^\infty H' \frac{
(H')^{d_H+n-1}e^{-H'/c} }{ c^{d_H+n}\Gamma(d_H+n) } \,dH' \\ &= e^{-\phi
H/c} \sum_{n=0}^{\infty} \frac{(\phi H/c)^n}{n!}\, c(d_H+n). \end{aligned}
\end{equation*}%
Given this, we obtain 
\begin{equation*}
\begin{aligned} \mathbb{E}[H_{t+1}\mid H_t=H] &= c e^{-\phi H/c}
\sum_{n=0}^{\infty} \frac{(\phi H/c)^n}{n!}(d_H+n) \\ &= cd_H + c e^{-\phi
H/c} \sum_{n=0}^{\infty} n\frac{(\phi H/c)^n}{n!} = cd_H+\phi H,
\end{aligned}
\end{equation*}%
where we are using $e^{-\lambda }\sum_{n=0}^{\infty }n\lambda
^{n}/n!=\lambda $. Therefore, 
\begin{equation*}
\mathbb{E}[H_{t+1}\mid H_{t}]=cd_{H}+\phi H_{t},
\end{equation*}%
as desired. Moreover, one  can also show that (see \cite[Appendix 1, Lemma~1]{GourierouxJasiak2006})
\begin{equation}
\mathbb{E}\left[e^{bH_{t+1}}\mid H_t=H\right]
=
(1-cb)^{-d_H}
\exp\left(\frac{\phi b}{1-cb}H\right),
\qquad cb<1.
\label{eq:arg-exponential-moment}
\end{equation}

\paragraph{Integrability of consumption.}

Using the law of iterated expectations and conditioning first on $j_{t+1}$,
we obtain 
\begin{equation*}
\begin{aligned} \mathbb{E} \left[ e^{g_{t+1}} \mid H_t=H \right] &=
e^{\nu_g} \mathbb{E} \left[ e^{\nu_jj_{t+1}} \mathbb{E} \left[ e^{
\sigma_j\sqrt{j_{t+1}}\eta_{t+1} + \sigma\varepsilon_{t+1} } \biggm|
j_{t+1},H_t=H \right] \biggm| H_t=H \right] \\ &= \exp\left(
\nu_g+\frac{1}{2}\sigma^2 \right) \mathbb{E} \left[ \exp\left\{ \left(
\nu_j+\frac{1}{2}\sigma_j^2 \right)j_{t+1} \right\} \biggm| H_t=H \right].
\end{aligned}
\end{equation*}
The second equality uses the independence and standard normality of $%
\eta_{t+1}$ and $\varepsilon_{t+1}$, together with $\mathbb{E}%
[e^{aX}]=e^{a^2/2}$ for $X\sim N(0,1)$.

Since $j_{t+1}\mid H_t=H\sim\text{Poisson}(H)$, 
\begin{equation*}
\begin{aligned} &\mathbb{E} \left[ \exp\left\{ \left(
\nu_j+\frac{1}{2}\sigma_j^2 \right)j_{t+1} \right\} \biggm| H_t=H \right] \\
&\qquad= e^{-H} \sum_{k=0}^{\infty} \frac{ \left[ H\exp\left(
\nu_j+\frac{1}{2}\sigma_j^2 \right) \right]^k }{k!} \\ &\qquad= \exp\left\{
H\left[ \exp\left( \nu_j+\frac{1}{2}\sigma_j^2 \right)-1 \right] \right\}.
\end{aligned}
\end{equation*}
Therefore, 
\begin{equation}  \label{eq:bsw-first-growth-moment}
\mathbb{E} \left[ e^{g_{t+1}} \mid H_t=H \right] = \exp\left( \nu_g+\frac{1}{%
2}\sigma^2 \right) \exp\left\{ H\left[ \exp\left( \nu_j+\frac{1}{2}%
\sigma_j^2 \right)-1 \right] \right\}.
\end{equation}
If 
\begin{equation*}
\nu_j+\frac{1}{2}\sigma_j^2\leq0,
\end{equation*}
then 
\begin{equation*}
\exp\left( \nu_j+\frac{1}{2}\sigma_j^2 \right)-1 \leq0,
\end{equation*}
and therefore 
\begin{equation}  \label{eq:bsw-growth-bound}
\mathbb{E} \left[ e^{g_{t+1}} \mid H_t \right] \leq \exp\left( \nu_g+\frac{1%
}{2}\sigma^2 \right).
\end{equation}

Since $C_{t+1}=C_{t}e^{g_{t+1}}$, the law of iterated expectations and %
\eqref{eq:bsw-growth-bound} imply 
\begin{equation*}
\begin{aligned} \mathbb{E}[C_{t+1}\mid H_0] &= \mathbb{E} \left[ C_t
\mathbb{E} \left[ e^{g_{t+1}} \mid H_t \right] \biggm| H_0 \right] \\ &\leq
\exp\left( \nu_g+\frac{1}{2}\sigma^2 \right) \mathbb{E}[C_t\mid H_0].
\end{aligned}
\end{equation*}%
Iterating this inequality yields 
\begin{equation*}
\mathbb{E}[C_{t}\mid H_{0}]\leq C_{0}\exp \left\{ t\left( \nu _{g}+\frac{1}{2%
}\sigma ^{2}\right) \right\} .
\end{equation*}%
Consequently, 
\begin{equation*}
\begin{aligned} \sum_{t=0}^{\infty} \delta^t\mathbb{E}[C_t\mid H_0] &\leq
C_0 \sum_{t=0}^{\infty} \left[ \delta \exp\left( \nu_g+\frac{1}{2}\sigma^2
\right) \right]^t \\ &< \infty \end{aligned}
\end{equation*}%
whenever 
\begin{equation*}
\delta \exp \left( \nu _{g}+\frac{1}{2}\sigma ^{2}\right) <1.
\end{equation*}%
A sufficient condition for $C\in \mathcal{L}_{++}^{1}(\delta )$ is 
\begin{equation*}
\nu _{j}+\frac{1}{2}\sigma _{j}^{2}\leq 0\qquad \text{and}\qquad 0<\delta
<\min \left\{ 1,\exp \left( -\nu _{g}-\frac{1}{2}\sigma ^{2}\right) \right\}
,
\end{equation*}%
as desired.

\paragraph{Function $h$.}

Using \eqref{eq:rare-disaster-h-definition} and the same calculation as in %
\eqref{eq:bsw-first-growth-moment} we have 
\begin{equation*}
\begin{aligned} {h(H)}^s=\mathbb{E} \left[ e^{s g_{t+1}} \mid H_t=H \right]
&= \exp\left( s\nu_g+\frac{1}{2}s^2\sigma^2 \right) \\ &\quad\times
\exp\left\{ H \left[ \exp\left( s\nu_j+\frac{1}{2}s^2\sigma_j^2 \right)-1
\right] \right\}. \end{aligned}
\end{equation*}%
Namely, 
\begin{equation*}
h(H)=\exp \left\{ \nu _{g}+\frac{1}{2}s\sigma ^{2}+\frac{H}{s}\left[ \exp
\left( s\nu _{j}+\frac{1}{2}s^{2}\sigma _{j}^{2}\right) -1\right] \right\} .
\end{equation*}%
Defining 
\begin{equation*}
B_{s}:=\exp \left( \nu _{g}+\frac{1}{2}s\sigma ^{2}\right) ,\qquad A_{s}:=%
\frac{\exp \left( s\nu _{j}+\frac{1}{2}s^{2}\sigma _{j}^{2}\right) -1}{s},
\end{equation*}%
we get 
\begin{equation*}
h(H)=B_{s}e^{A_{s}H},
\end{equation*}%
as desired.

\paragraph{Fixed points.}

In this case the operator $\mathbf{A}$  of Section~\ref%
{sec:markovstationary} leads to 
\begin{equation*}
\begin{aligned} (\mathbf{A}\varphi)(H) &= \left\{ (1-\beta) + \beta \left[
\int_{\mathbb{R}_+} \int_{\mathcal{Y}} e^{s\kappa(H,y)} \varphi(H')^s
\,\nu(dy)\, Q(H,dH') \right]^{r/s} \right\}^{1/r} \\ &= \left\{ (1-\beta) +
\beta \left[ \left( \int_{\mathcal{Y}} e^{s\kappa(H,y)} \,\nu(dy) \right)
\left( \int_{\mathbb{R}_+} \varphi(H')^s Q(H,dH') \right) \right]^{r/s}
\right\}^{1/r} \\ &= \left\{ (1-\beta) + \beta \left[ h(H)^s
\int_{\mathbb{R}_+}\varphi(H^{\prime})^s Q(H,dH^{\prime}) \right]^{r/s}
\right\}^{1/r}. \end{aligned}
\end{equation*}%
Hence $\varphi $ is a fixed point of $\mathbf{A}$ if 
\begin{equation*}
\varphi (H)=\left\{ (1-\beta )+\beta \left[ h(H)^{s}\int_{\mathbb{R}%
_{+}}\varphi (H^{\prime })^{s}Q(H,dH^{\prime })\right] ^{r/s}\right\} ^{1/r}.
\end{equation*}%
By Proposition~\ref{prop:separable-A-implies-fixed-point}, every such
positive fixed point $\varphi $ generates a separable recursive utility 
\begin{equation*}
U_{t}=\varphi (H_{t})C_{t}.
\end{equation*}%
As $\lim_{H\rightarrow \infty }h(H)=0$, it follows that any fixed point $%
\varphi $ of $\mathbf{A}$ must satisfy $\lim_{H\rightarrow \infty }\varphi
(H)=(1-\beta )^{1/r}$, which is consistent with the fact that $\varphi \in 
\mathcal{K}(\mathbf{1})$.

\paragraph{Unit-EIS solutions.}Consider the ansatz 
\begin{equation*}
Z_{t}=C_{t}\exp \left( \frac{b}{s}H_{t}\right),
\end{equation*}%
where $b$ is an unknown parameter to be determined.
Since, conditional on $H_t$, the shocks determining consumption growth are independent of $H_{t+1}$,
\[
\mathbb{E}_t[Z_{t+1}^s]
=
C_t^s\mathbb{E}_t[e^{sg_{t+1}}]
\mathbb{E}_t[e^{bH_{t+1}}].
\]
Using $h(H)=\left(\mathbb{E}\left[e^{s g_{t+1}} \mid H_t=H\right]\right)^{1 / s}=B_se^{A_sH}$ and \eqref{eq:arg-exponential-moment}, we obtain
\[
\mathbb{E}_{t}[Z_{t+1}^{s}]
=
C_{t}^{s}B_{s}^{s}(1-cb)^{-d_{H}}
\exp \left\{
\left( sA_{s}+\frac{\phi b}{1-cb}\right) H_{t}
\right\},
\]
for $cb<1$.
 The right-hand side of~\eqref{eq:DDD} is 
\begin{equation*}
C_{t}^{-\alpha s/\beta }Z_{t}^{s/\beta }B_{t}=C_{t}^{s}\exp \left( \frac{b}{%
\beta }H_{t}\right) B_{t}.
\end{equation*}%
It follows that
\[
B_t
=
B_s^s(1-cb)^{-d_H}
\exp\left\{
\left(
sA_s+\frac{\phi b}{1-cb}-\frac{b}{\beta}
\right)H_t
\right\}.
\]
Since $H_t$ has unbounded support on $\mathbb{R}_+$, we have $B\in\mathcal{K}(\mathbf{1})$ if and only if the coefficient of $H_t$ in the exponent is zero, that is,
\[
sA_s+\frac{\phi b}{1-cb}=\frac{b}{\beta},
\]
or, equivalently,
\begin{equation*}
cb^{2}+\left( \beta \phi -1-\beta sA_{s}c\right) b+\beta sA_{s}=0.
\end{equation*}%
Under the calibration of Table~\ref{tab:ez-rd-calibration}, the two roots
are 
\begin{equation*}
b_{1}\simeq 231.8,\qquad b_{2}\simeq 773.2,
\end{equation*}%
and therefore 
\begin{equation*}
Z_{t}^{(1)}\simeq C_{t}e^{-83.39H_{t}},\quad Z_{t}^{(2)}\simeq
C_{t}e^{-278.1H_{t}}.
\end{equation*}%
For each root, the corresponding process in~\eqref{eq:DDD} is constant, 
\begin{equation*}
B_{t}^{(i)}\equiv \bar{B}_{i}=B_{s}^{s}(1-cb_{i})^{-d_{H}}.
\end{equation*}%
Hence, in view of~\eqref{eq:LSX}, the unique recursive utilities in the two
corresponding cones are 
\begin{equation*}
U^{(i)}=\bar{B}_{i}^{\beta /[s(1-\beta )]}Z^{(i)},\qquad i=1,2,
\end{equation*}%
which numerically yield the two utilities 
\begin{equation*}
U_{t}^{(1)}\simeq C_{t}e^{\,0.3686-83.39H_{t}},\quad U_{t}^{(2)}\simeq
C_{t}e^{-3.022-278.1H_{t}}.
\end{equation*}

\subsection{Calculations for the long-run risk model}

\label{app:multiplicity}

Here we prove the claims in Section~\ref{sec:long-run-risk-multiplicity}. In
the constant-volatility specification we write $\sigma_t\equiv\bar\sigma>0$,
consistently with the calibration in %
\eqref{eq:unit-eis-BY-common-calibration}. Throughout, $C_0>0$ and $X_0\in%
\mathbb{R}$ are assumed to be deterministic.

\paragraph{Failure of uniqueness.}

Let us first describe the general model \eqref{eq:auto}. Let $Q_{\text{vol}}$
be the transition kernel of the Markov process $(\sigma _{t})$ and $\eta $
the distribution of $\varepsilon _{1}^{x}$. The Markov state $(X_{t},\sigma
_{t})$ has transition kernel 
\begin{equation*}
Q\big((x,\sigma ),A\times B\big)=Q_{\text{vol}}(\sigma ,B)\int_{\mathbb{R}}%
\mathbf{1}_{A}(\rho x+\chi\sigma e)\,\eta (de),
\end{equation*}%
for Borel sets $A\subseteq \mathbb{R}$ and $B\subseteq (0,\infty )$. The
growth function $\kappa $ in \eqref{eq:XXX} is 
\begin{equation*}
\kappa \big((x,\sigma ),(x^{\prime },\sigma ^{\prime }),y^{\prime }\big)=\mu
_{c}+x+\sigma y^{\prime }.
\end{equation*}%
Under our assumptions, we have that 
\begin{equation*}
\mathbb{E}_{t}[C_{t+1}^{s}]=C_{t}^{s}e^{s(a+X_{t})}\mathbb{E}%
_{t}[e^{s\sigma _{t}Y_{t+1}}],
\end{equation*}%
which implies 
\begin{equation*}
\frac{\mathcal{M}_{t}(C_{t+1})}{C_{t}}=e^{a+X_{t}}\bigl(\mathbb{E}%
_{t}[e^{s\sigma _{t}Y_{t+1}}]\bigr)^{1/s}=h(X_{t},\sigma _{t}).
\end{equation*}%
Therefore, 
\begin{equation*}
h(x,\sigma )=e^{a+x}\bigl(\mathbb{E}[e^{s\sigma Y_{1}}]\bigr)^{1/s}.
\end{equation*}%
Observe that $h(x,\sigma )\rightarrow \infty $ as $x\rightarrow \infty $.

\paragraph{Constant-volatility model.}

For the rest of the construction, let $0<\rho <1$ and $\sigma
_{t}:= \bar{\sigma}>0$, and take both innovation sequences to be
standard normal. Write $\sigma _{x}=\chi\bar{\sigma}>0$. The dynamic
model becomes 
\begin{equation*}
\log \frac{C_{t+1}}{C_{t}}=a+X_{t}+\bar{\sigma}Y_{t+1},\quad
X_{t+1}=\rho X_{t}+\sigma _{x}\varepsilon _{t+1}^{x}.
\end{equation*}%
Given the ansatz 
\begin{equation*}
Z_{t}=C_{t}\exp \!\left( \frac{qX_{t}^{2}+bX_{t}}{s}\right)
\end{equation*}%
the factor multiplying $B_{t}$ on the right-hand side of \eqref{eq:DDD} is
given by 
\begin{equation*}
\begin{aligned} C_t^{-(1-\beta)s/\beta}Z_t^{s/\beta}
&=C_t^{-(1-\beta)s/\beta+s/\beta}
\exp\!\left(\frac{qX_t^2+bX_t}{\beta}\right)\\
&=C_t^s\exp\!\left(\frac{qX_t^2+bX_t}{\beta}\right). \end{aligned}
\end{equation*}%
The left-hand side of equation \eqref{eq:DDD} is 
\begin{equation*}
\mathbb{E}_{t}[Z_{t+1}^{s}]=C_{t}^{s}e^{s(a+X_{t})}\mathbb{E}[e^{s%
\bar{\sigma}Y_{1}}]\mathbb{E}_{t}[e^{qX_{t+1}^{2}+bX_{t+1}}].
\end{equation*}%
Observe that we have 
\begin{equation*}
\begin{aligned} \mathbb E[e^{s\bar\sigma Y_1}] &=\frac{1}{\sqrt{2\pi}}
\int_{\mathbb R}e^{s\bar\sigma y-y^2/2}\,dy\\ &=e^{s^2\bar\sigma^2/2}
\frac{1}{\sqrt{2\pi}} \int_{\mathbb R}e^{-(y-s\bar\sigma)^2/2}\,dy
=e^{s^2\bar\sigma^2/2}. \end{aligned}
\end{equation*}%
For the state innovation, put $d=1-2q\sigma _{x}^{2}$ and assume $d>0$.
Conditional on $X_{t}=x$, the next state has the form $X_{t+1}=\rho x+\sigma
_{x}\varepsilon _{t+1}^{x}$, where $\varepsilon _{t+1}^{x}$ is standard
normal. Hence 
\begin{equation*}
\mathbb{E}_{t}[e^{qX_{t+1}^{2}+bX_{t+1}}]=\frac{1}{\sqrt{2\pi }}\int_{%
\mathbb{R}}\exp \!\left( q(\rho x+\sigma _{x}y)^{2}+b(\rho x+\sigma _{x}y)-%
\frac{y^{2}}{2}\right) \,dy.
\end{equation*}%
Expanding the exponent gives 
\begin{equation*}
\begin{aligned} &q(\rho x+\sigma_x y)^2+b(\rho x+\sigma_x y)-\frac{y^2}{2}\\
&\quad=q\rho^2x^2+b\rho x +\left(q\sigma_x^2-\frac12\right)y^2
+\sigma_x(2q\rho x+b)y. \end{aligned}
\end{equation*}%
Since $d=1-2q\sigma _{x}^{2}>0$, the coefficient of $y^{2}$ is $-d/2$.
Taking the terms independent of $y$ outside the integral therefore yields 
\begin{equation*}
\mathbb{E}_{t}[e^{qX_{t+1}^{2}+bX_{t+1}}]=\frac{e^{q\rho ^{2}x^{2}+b\rho x}}{%
\sqrt{2\pi }}\int_{\mathbb{R}}\exp \!\left( -\frac{d}{2}y^{2}+\sigma
_{x}(2q\rho x+b)y\right) \,dy.
\end{equation*}

To evaluate this integral, set $m=\sigma_x(2q\rho x+b)/d$ and complete the
square: 
\begin{equation*}
-\frac d2y^2+\sigma_x(2q\rho x+b)y =-\frac d2(y-m)^2+\frac d2m^2.
\end{equation*}
The change of variables $u=\sqrt d\,(y-m)$ gives 
\begin{equation*}
\int_{\mathbb{R}}e^{-d(y-m)^2/2}\,dy =\frac1{\sqrt d}\int_{\mathbb{R}%
}e^{-u^2/2}\,du =\sqrt{\frac{2\pi}{d}}.
\end{equation*}
Consequently, 
\begin{equation*}
\mathbb{E}_t[e^{qX_{t+1}^2+bX_{t+1}}] =d^{-1/2}\exp\!\left( q\rho^2x^2+b\rho
x +\frac{\sigma_x^2(2q\rho x+b)^2}{2d} \right).
\end{equation*}

Finally, putting the exponent over the common denominator $d$ and using $%
d+2q\sigma _{x}^{2}=1$, its numerator simplifies to 
\begin{equation*}
\begin{aligned} &d(q\rho^2x^2+b\rho x) +\frac{\sigma_x^2}{2}(2q\rho x+b)^2\\
&\quad=(d+2q\sigma_x^2)(q\rho^2x^2+b\rho x) +\frac{\sigma_x^2b^2}{2}\\
&\quad=q\rho^2x^2+b\rho x+\frac{\sigma_x^2b^2}{2}. \end{aligned}
\end{equation*}%
We therefore obtain 
\begin{equation*}
\mathbb{E}_{t}[e^{qX_{t+1}^{2}+bX_{t+1}}]=d^{-1/2}\exp \!\left( \frac{q\rho
^{2}X_{t}^{2}+b\rho X_{t}+\tfrac{1}{2}\sigma _{x}^{2}b^{2}}{d}\right) ,
\end{equation*}%
which implies 
\begin{equation*}
\begin{aligned} \mathbb E_t[Z_{t+1}^{s}]
&=C_t^s\exp\!\left(sa+\frac{s^2\bar\sigma^2}{2}\right) d^{-1/2}\\
&\quad\times\exp\!\left( sX_t+ \frac{q\rho^2X_t^2+b\rho
X_t+\tfrac12\sigma_x^2b^2}{d} \right). \end{aligned}
\end{equation*}%
Dividing this expression by $C_{t}^{s}\exp ((qX_{t}^{2}+bX_{t})/\beta )$, we
find that the process $B$ in \eqref{eq:DDD} satisfies: 
\begin{equation*}
\begin{aligned} \log B_t &=sa+\frac{s^2\bar\sigma^2}{2} -\frac12\log
d+\frac{\sigma_x^2b^2}{2d}\\
&\quad+q\left(\frac{\rho^2}{d}-\frac1\beta\right)X_t^2
+\left(s+\frac{b\rho}{d}-\frac b\beta\right)X_t. \end{aligned}
\end{equation*}%
Hence $B\in \mathcal{K}(\mathbf{1})$ is equivalent to 
\begin{equation*}
q\left( \frac{\rho ^{2}}{d}-\frac{1}{\beta }\right) =0,\qquad s+\frac{b\rho 
}{d}-\frac{b}{\beta }=0.
\end{equation*}%
If $q=0$, then $d=1$ and the second equation gives $b=\beta s/(1-\beta \rho
) $. If $q\neq 0$, the first equation gives $d=\beta \rho ^{2}$. In view of
the definition of $d$ we get $q=(1-\beta \rho ^{2})/(2\sigma _{x}^{2})$,
while the second equation becomes 
\begin{equation*}
s+\frac{b}{\beta \rho }-\frac{b}{\beta }=0\quad \Longleftrightarrow \quad
\beta s\rho +b(1-\rho )=0.
\end{equation*}%
The two choices are therefore 
\begin{equation*}
(q_{0},b_{0})=\left( 0,\frac{\beta s}{1-\beta \rho }\right) ,\qquad
(q_{1},b_{1})=\left( \frac{1-\beta \rho ^{2}}{2\sigma _{x}^{2}},\frac{\beta
s\rho }{\rho -1}\right) .
\end{equation*}%
They satisfy $q_{1}>0$, $d_{0}=1$, and $d_{1}=\beta \rho ^{2}>0$, where $%
d_{i}=1-2q_{i}\sigma _{x}^{2}$. Thus the conditional moments evaluated above
are finite for both choices. The residuals are the positive constants 
\begin{equation*}
\bar{B}_{i}=\exp \!\left( sa+\frac{s^{2}\bar{\sigma}^{2}}{2}-\frac{1}{%
2}\log d_{i}+\frac{\sigma _{x}^{2}b_{i}^{2}}{2d_{i}}\right) ,\qquad i=0,1.
\end{equation*}

Finally, set $Z_{t}^{(i)}=C_{t}\exp ((q_{i}X_{t}^{2}+b_{i}X_{t})/s)$. By %
\eqref{eq:LSX}, the separable utilities are 
\begin{equation*}
U_{t}^{(i)}=\bar{B}_{i}^{\beta /[s(1-\beta )]}Z_{t}^{(i)}=C_{t}\exp \!\left(
A_{i}+\frac{q_{i}X_{t}^{2}+b_{i}X_{t}}{s}\right) ,\qquad i=0,1,
\end{equation*}%
where $A_{i}=\beta \log \bar{B}_{i}/[s(1-\beta )]$. It remains to check that 
$C,Z^{(i)}\in \mathcal{L}_{++}^{1}(\delta )$.

\paragraph{A sufficient condition for $C,Z^{(i)}\in\mathcal{L}^1_{++}(%
\protect\delta)$.}

Let 
\begin{equation*}
R:=\exp \!\left( a+\frac{\bar{\sigma}^{2}}{2}+\frac{\sigma _{x}^{2}}{%
2(1-\rho )^{2}}\right) .
\end{equation*}%
Now iterating the state equation gives $X_{k}=\rho ^{k}X_{0}+\sigma
_{x}\sum_{j=1}^{k}\rho ^{k-j}\varepsilon _{j}^{x}$. Substituting this
expression into $\log (C_{t}/C_{0})=ta+\sum_{k=0}^{t-1}X_{k}+\bar{%
\sigma}\sum_{j=1}^{t}Y_{j}$ and summing up the geometric series we obtain
for $t\geq 1$ 
\begin{equation*}
\begin{aligned} \log\frac{C_t}{C_0} &=ta+\frac{1-\rho^t}{1-\rho}X_0
+\bar\sigma\sum_{j=1}^tY_j\\ &\quad+\sigma_x\sum_{j=1}^{t-1}
\frac{1-\rho^{t-j}}{1-\rho}\varepsilon^x_j. \end{aligned}
\end{equation*}%
As the shocks are independent standard normal random variables, we obtain 
\begin{equation*}
\begin{aligned} \mathbb E[C_t] &=C_0\exp\!\left\{
t\left(a+\frac{\bar\sigma^2}{2}\right) +\frac{1-\rho^t}{1-\rho}X_0
+\frac{\sigma_x^2}{2}\sum_{j=1}^{t-1}
\left(\frac{1-\rho^{t-j}}{1-\rho}\right)^2 \right\}\\ &\leq
C_0\exp\!\left(\frac{|X_0|}{1-\rho}\right)R^t. \end{aligned}
\end{equation*}%
This inequality uses the inequality $0\leq (1-\rho ^{k})/(1-\rho )\leq
1/(1-\rho )$ and the fact that the sum has at most $t$ terms. Consequently,
whenever $0<\delta <\min \{1,R^{-1}\}$, 
\begin{equation*}
\sum_{t=0}^{\infty }\delta ^{t}\mathbb{E}[C_{t}]\leq C_{0}\exp \!\left( 
\frac{|X_{0}|}{1-\rho }\right) \sum_{t=0}^{\infty }(\delta R)^{t}=\frac{%
C_{0}\exp (|X_{0}|/(1-\rho ))}{1-\delta R}<\infty .
\end{equation*}%
It follows that $\delta R<1$ implies $C\in \mathcal{L}_{++}^{1}(\delta )$.
Incidentally, quick calculations also reveal that $\delta R<1$ is a
necessary condition for $C\in \mathcal{L}_{++}^{1}(\delta )$.

Finally, observe that since $0<b_0/s<1/(1-\rho)$, the Gaussian AR(1)
specification gives, for sufficiently large $n$ and some constant $K>0$, 
\begin{equation*}
Z_t^{(0)}\leq C_t+K\mathbb{E}_t[C_{t+n}].
\end{equation*}
Completing the square also gives $Z_t^{(1)}\leq e^{-b_1^2/(4q_1s)}C_t$.
Taking expectations and summing with weights $\delta^t$ therefore shows that
admissibility of $C$ implies admissibility of both generators and their
constant multiples $U^{(i)}$.

\paragraph{Bansal--Yaron Case I calibration.}

Under \eqref{eq:unit-eis-BY-common-calibration} we have 
\begin{equation*}
\begin{aligned} \sigma_x&=0.044\times0.0078=0.0003432,\\ q_0&=0,
&b_0&=\frac{0.998(-9)}{1-0.998(0.979)},\\
q_1&=\frac{1-0.998(0.979)^2}{2(0.0003432)^2},
&b_1&=\frac{0.998(-9)(0.979)}{0.979-1}. \end{aligned}
\end{equation*}%
Using these exact expressions in the formula for $A_{i}$ above gives 
\begin{equation*}
\begin{aligned} A_0&\simeq0.1121, &b_0/s&\simeq43.47, &q_0/s&=0,\\
A_1&\simeq-1.219, &b_1/s&\simeq-46.53, &q_1/s&\simeq-2.051\times10^4.
\end{aligned}
\end{equation*}%
Substituting these values into the formula found above we obtain the two
recursive utilities 
\begin{equation*}
\begin{aligned} U_t^{(0)} &\simeq C_t\exp\!\left(0.1121+43.47X_t\right),\\
U_t^{(1)} &\simeq C_t\exp\!\left( -1.219-46.53X_t-(2.051\times10^4)X_t^2
\right). \end{aligned}
\end{equation*}

For this calibration, the consumption admissibility condition is $%
0<\delta<R^{-1}$, where 
\begin{equation*}
R^{-1}=\exp\!\left( -0.0015-\frac{0.0078^2}{2} -\frac{0.0003432^2}{%
2(1-0.979)^2}\right) \simeq0.9983.
\end{equation*}
Under the calibration in \eqref{eq:unit-eis-BY-common-calibration},
subtracting the exponents of the two utilities gives 
\begin{equation*}
\log\frac{U_t^{(1)}}{U_t^{(0)}} \simeq
-1.331-90.00X_t-(2.051\times10^4)X_t^2.
\end{equation*}
This quadratic is maximized at 
\begin{equation*}
X_t\simeq-\frac{90.00}{2(2.051\times10^4)} \simeq-0.002194,
\end{equation*}
where its value is 
\begin{equation*}
-1.331+\frac{90.00^2}{4(2.051\times10^4)} \simeq-1.232<0.
\end{equation*}
Hence, we obtain that $U_t^{(1)}<U_t^{(0)}$ as desired.

\end{document}